\documentclass[aps,prl,reprint,groupedaddress,footinbib, nobibnotes]{revtex4-2}

\usepackage{amsmath,amsthm,amssymb}
\usepackage{bibunits}
\defaultbibliographystyle{apsrev4-2}
\defaultbibliography{apssamp}

\usepackage{color,soul}
\usepackage{graphicx}
\usepackage{dcolumn}
\usepackage{bm}
\usepackage[mathlines]{lineno}
\usepackage{hyperref}

\theoremstyle{plain}

\newtheorem{theorem}{Theorem}
\newtheorem{atheorem}{Theorem}

\newtheorem{lemma}{Lemma}
\newtheorem{corollary}{Corollary}

\theoremstyle{remark}
\newtheorem{remark}{Remark}

\theoremstyle{definition}

\newcommand{\R}{\mathbb{R}}
\newcommand{\rr}{{\mathbf{r}}}
\newcommand{\BB}{{\mathcal B}}
\newcommand{\const}{\mathrm{const}}

\DeclareMathOperator{\Int}{int}
\DeclareMathOperator{\Cl}{cl}

\begin{document}
\begin{bibunit}

\title{Fundamental limits on state preparation for an open qubit}
\author{D.A. Abraamian}
\email{david.abraamian@math.msu.ru}
\affiliation{Department of Mathematics and Mechanics, Lomonosov Moscow State University, Leninskiye gory str. 1, Moscow, 119991, Russia}
\author{L.V. Lokutsievskiy}
\email{lion.lokut@gmail.com}
\affiliation{Department of Differential Equations,
Steklov Mathematical Institute of Russian Academy of Sciences, Gubkina str. 8, Moscow, 119991, Russia}
\affiliation{Faculty of Mathematics, National Research University Higher School of Economics, Usacheva Str. 6, Moscow, 119048, Russia}
\author{A.N. Pechen}
\email{apechen@gmail.com}
\affiliation{Department of Mathematical Methods for Quantum Technologies,
Steklov Mathematical Institute of Russian Academy of Sciences, Gubkina str. 8, Moscow, 119991, Russia}

\sloppy

\date{\today}
\begin{abstract}
We analytically determine the ultimate limits of  state preparation in two-level open quantum systems driven by coherent control. For a dissipative qubit governed by a GKSL master equation, we give an exact characterization of the reachable set in the Bloch ball. Dissipation excludes a region of states in the Bloch ball which cannot be approached even under arbitrarily strong coherent driving, and we prove that this region has a nontrivial geometry whose boundary is a surface of revolution around the $x$-axis which is analytic except for two conical singularities. We derive a closed-form control protocol for moving on this boundary, and construct an explicit protocol that steers the system arbitrarily close to any prescribed boundary state. These results provide a complete geometric constructive description of reachable qubit states in the standard dissipative environment, establishing fundamental bounds on controllability and state-preparation fidelity for open two-level quantum systems.
\end{abstract}

\maketitle

While the control of isolated, i.e. closed, quantum systems, governed by the Schr\"odinger equation, has been extensively studied, realistic quantum devices inevitably interact with their surrounding environment~\cite{Schleich_etal_2016,Acin_etal_2018,Koch_etal_2022}. The control of open quantum systems is therefore central to a wide range of quantum technologies~\cite{Brif_Chakrabarti_Rabitz_2010,ShapiroBrumerBook2012,Koch_2016,Cong_2018,AlessandroBook2021,Kurizki_Kofman_2021,KallushDannKosloff2022}. A particularly important task is the preparation of quantum states that are as close as possible to a prescribed target state, either pure or mixed. 

The dynamics of open quantum systems are typically nonunitary and include dissipation and decoherence \cite{BreuerPetruccioneBook2007}. The environment is often viewed as detrimental since it introduces errors and limits coherence times. In other cases, it can be exploited as a useful resource~\cite{Beige_Braun_Tregenna_Knight_2000,Diehl_etal_2008,Verstraete_Wolf_Cirac_2009,Schmidt_Negretti_Ankerhold_Calarco_Stockburger_2011,HarringtonMuellerMurch2022,Coleetal2022,Malinowskietal2022}. An example is incoherent control which exploits  properties of the environment such as its spectral density for affecting the system's evolution \cite{PechenRabitz2006, Pechen2011} to complement traditional coherent control, where the system is driven by external fields such as shaped laser pulses or magnetic fields.

A fundamental question in quantum control is  controllability: which states can be reached from a given initial state? The answer to this question determines the ultimate limits on manipulating a given system. For closed quantum systems, controllability criteria based on the Lie algebra generated by the system Hamiltonians are well established \cite{Huang_Tarn_Clark_1983,Turinici2001, Albertini2001, Schirmer2001, Schirmer2002, Altafini2002, Polack2009}. For open quantum systems governed by the GKSL equations~\cite{GoriniKossakowskiSudarshan1976,Lindblad1976}, the analysis is more subtle because the dynamics are dissipative and generally irreversible. Early studies addressed controllability properties for general finite-dimensional Markovian systems \cite{Altafini1} and, in particular, for two-level systems with specific dissipation mechanisms and control resources \cite{Altafini2}. Further reachability and controllability questions for Markovian open quantum systems were studied using Lie-semigroup methods~\cite{Dirr_Helmke_Kurniawan_SchulteHerbrueggen_2009}.

A related problem is developing time-optimal control protocols which for open systems is complicated by the interplay between coherent driving and dissipation. Various studies have addressed this problem using tools from geometric control theory\cite{AgrachevSachkov}, notably the Pontryagin Maximum Principle~\cite{Sugny_Kontz_Jauslin_2007, Lapert_Zhang_Braun_Glaser_Sugny_2010, BoscainPiccoli12004}, including for optimization of state purity~\cite{Clark_Bloch_Colombo_Rooney_2020}, mitigation of dissipation during state transfer~\cite{Mukherjee_Carlini_Mari_Caneva_Montangero_Calarco_Fazio_Giovannetti_2013}, and speed limit bounds for open system dynamics~\cite{delCampo_Egusquiza_Plenio_Huelga_2013}. 

Two-level quantum systems provide the basic model for many platforms, ranging from qubits in quantum information processing to spin dynamics in NMR~\cite{Zhang_Lapert_Sugny_Braun_Glaser_2011}, quantum optics~\cite{Sugny_Kontz_Jauslin_2007}, chemical dynamics, and biological systems such as light-harvesting complexes~\cite{Fassioli_Dinshaw_Arpin_Scholes_2014,Kozyrev_Pechen_2022}. Much effort has been directed toward the analysis of controllability and high-fidelity control of two-level quantum systems, either closed or open, as well as of qubit gates implemented in the presence of decoherence~\cite{Viola_Lloyd_1998,Viola_Knill_Lloyd_1999,Viola_Lloyd_Knill_1999,Grace_Brif_Rabitz_Walmsley_Kosut_Lidar_2007,West_Lidar_Fong_Gyure_2010,Zahedinejad_Ghosh_Sanders_2015,Caneva_2009,Hegerfeldt_2013,Palao_Kosloff_2002}. Dissipation excludes a region of states in the Bloch ball which cannot be approached even under arbitrarily strong coherent driving, and describing this region determines fundamental limits on state preparation in two-level open quantum systems. For coherently controlled open systems, approximate reachable-set constructions were developed for quantum state engineering in NMR systems~\cite{Li_Lu_Luo_Laflamme_Peng_Du_2016}. Reachable set at the level of its closure was characterized and time-optimal state transfer was studied~\cite{LokutsievskiyPechen1,LokutsievskiyPechen2} by extending classical geometric control theory \cite{Sussmann_1982, Sussmann_1987} to this quantum setting. Despite these advances, a complete geometric and dynamical characterization of the reachable sets and the  coherent controls generating boundary dynamics  have remained unavailable.

In this article, we give a complete analytic solution to the problem of characterizing the ultimate limits of state preparation for a two-level quantum system interacting with an environment. We consider an open qubit driven by coherent control and coupled to a photonic reservoir. Its dynamics are described by the Gorini--Kossakowski--Sudarshan--Lindblad (GKSL) master equation~\cite{BreuerPetruccioneBook2007,AlessandroBook2021}, which combines Hamiltonian evolution with nonunitary dissipation. The standard form of the master equation in quantum-optical models is ~\cite{Scully_Zubairy_1997}
\begin{equation}
\label{eq: main QCS}
	\frac{\mathrm{d}\rho}{\mathrm{d}t} =
	-\mathrm{i}\bigl[H_0+uV, \rho\bigr]
	+\gamma(n+1) \mathcal{D}[\sigma^-]\rho + \gamma n \mathcal{D}[\sigma^+]\rho,
\end{equation}
where we set $\hbar=1$. Here $\rho$ is the $2\times 2$ density matrix of the two level system, $H_0=\frac\omega2 \sigma_z$ is the free Hamiltonian, $u(t)$ is the coherent control, $V=\frac\kappa2 \sigma_x$ is the interaction Hamiltonian, $\sigma^+=(\sigma_x +\mathrm{i}\sigma_y)/2$ and $\sigma^-=(\sigma_x -\mathrm{i}\sigma_y)/2$ are the raising and lowering operators, $\sigma_{x,y,z}$ are the Pauli matrices, and $\mathcal{D}[\sigma]\rho=\sigma\rho \sigma^\dagger - \frac{1}{2} \{ \sigma^\dagger \sigma, \rho \}$ is the Lindblad dissipator. A key parameter is $\alpha = \gamma/\omega$. The choice of $H_0$ and $V$ describes a general dipole-type coupling to a coherent control field. In the most general setting, the control variables are the coherent control $u(t)\in\R$, for example a shaped coherent laser field, and the incoherent control $n(t)\ge 0$, for example the spectral density of incoherent photons. An additional dephasing term $\mathcal{D}[\sigma_z]$ may arise for other environments, for instance for phonon reservoirs. In the photonic case considered here, however, the relevant dissipators are $\mathcal{D}[\sigma^\pm]$, and $n$ represents the photon density at the transition frequency $\omega$.

Our main result is an exact analytical characterization of the set $\mathcal R$ of all quantum states reachable from the equilibrium state $N$ of an open qubit under arbitrary coherent control\footnote{The set $\mathcal{R}(\rho_0)$ of all quantum states reachable from a given initial state $\rho_0$ always contains $\mathcal{R}=\mathcal{R}(N)$. Moreover, if $\rho_0\in\mathcal{R}$, then $\mathcal{R}(\rho_0)=\mathcal{R}$, see~\cite{LokutsievskiyPechen1}.}. Moreover, we construct explicit control protocols that steer the system arbitrarily close to any prescribed state on the boundary of this reachable set. Specifically,

\begin{enumerate}
	\item We determine the global geometry of the reachable-set boundary in Bloch coordinates $(r_x,r_y,r_z)$. We prove that this boundary is a surface of revolution around the $r_x$-axis which is analytic everywhere except at two conical singularities (see Theorem~\ref{thm: main}). We analytically construct a generating curve $\Gamma\subset\partial\mathcal R$ whose rotation about the $r_x$-axis gives the full boundary $\partial\mathcal R$. This reduces the three-dimensional reachable-set problem to the construction of a boundary trajectory.
    
	\item We derive a closed-form feedback control generating the trajectory $\Gamma$ on the boundary of the reachable set,
	\[
		\boxed{u = \frac\omega\kappa \frac{(1+r_z)(2 r_x r_z+2 \alpha r_y+\alpha r_yr_z)}{r_y^2+r_z^2+r_z^3},}
	\]
	which gives an explicit differential equation for the boundary trajectories (see Theorem~\ref{thm: main} for details).

	\item Using this boundary dynamics, we construct explicit coherent controls $u(t)$ that steer any initial state to an arbitrarily small neighborhood of any desired final state on the boundary of the reachable set (see~\hyperref[sec: practical implementation]{Appendix}). This provides a constructive method for preparing extremal reachable states.

	\item We describe the internal structure of the reachable set and identify the region of small-time local controllability, consisting of points whose neighborhoods are reachable in arbitrarily small time. This region is a body of revolution whose boundary, denoted by $\mathcal{T}$, has a torus-like structure. This boundary $\mathcal{T}$ is generated by rotating an explicitly described figure-eight-shaped curve \eqref{eq: 8-figure two equations} around the $r_x$-axis, as shown in Fig.~\ref{fig: final structure of boundary clear}. This provides a separation between locally controllable  and not locally controllable regions in the entire reachable set (see~\hyperref[sec: practical implementation]{Appendix} and Supplementary Online Material (SOM)).
\end{enumerate}

Now we provide precise formulations of these results, including the contributions of the present paper and relevant findings from previous works.

If we set $u=n=0$ in Equation~\eqref{eq: main QCS}, then every solution will approach the unique equilibrium state exponentially fast. In the basis $\{ |e\rangle, |g\rangle \}$, the equilibrium state is
$N = |g\rangle\langle g|$. Therefore, by Krener's theorem (see~\cite{AgrachevSachkov}), the reachable set $\mathcal{R}(\rho_0)$ from any initial state $\rho_0$ contains at least the interior of the reachable set $\mathcal{R}(N)$ from the state $N$ (see \cite[Section~4]{LokutsievskiyPechen1}). Moreover, for any initial state $\rho_0$ in the reachable set $\mathcal{R}(N)$, we have
\[
	\rho_0\in\mathcal{R}(N)
	\Rightarrow
	\operatorname{int} \mathcal{R}(N)\subset \mathcal{R}(\rho_0)\subset\mathcal{R}(N)\subset\operatorname{cl}\operatorname{int}\mathcal{R}(N).
\]
Therefore, determining the structure of the reachable set starting from the initial state $N$ is of primary interest.

Thus, we focus on the reachable set from the equilibrium state $N$. We consider two  already mentioned control scenarios: with only the coherent control (so that $n\equiv0$) and the full model with coherent and incoherent controls:
\begin{align*}
	\mathcal{R}_b & = \Big\{
	\rho_1 \mathrel{\Big|} \exists T\ge 0, u,n\ge 0: \rho(0)=N, \rho(T)=\rho_1
	\Big\},\\
	\mathcal{R}_c & = \Big\{\rho_1 \mathrel{\Big|} \exists T\ge 0, u,n\equiv 0: \rho(0)=N, \rho(T)=\rho_1
	\Big\}.
\end{align*}
Clearly, $\mathcal{R}_c\subset \mathcal{R}_b$. It was shown in \cite{LokutsievskiyPechen1} that setting $n\equiv0$ and steering system~\eqref{eq: main QCS}  only with the coherent control $u$ does not change the reachable set at the level of closure; closures of the reachable sets generated with and without incoherent control are the same. Specifically, the following result was proven:

\begin{theorem}[{see~\cite[Theorem~2]{LokutsievskiyPechen1}}]
\label{thm: cl int partial}
	Consider the six sets $\mathcal{R}_*$, $\operatorname{int}\mathcal{R}_*$, and $\operatorname{cl}\mathcal{R}_*$ where $*$ is $b$ or $c$. Then these six sets share the same boundary, interior, and closure.
\end{theorem}

Thus, Theorem~\ref{thm: cl int partial} prevents the reachable set from being, for example, a full-dimensional body missing an interior point, or the union of a full-dimensional body and a curve.

In the present work, we are interested in the boundary  structure of these reachable sets. Therefore we use the unified notation
\[
	\boxed{\mathcal{R} = \operatorname{cl}\mathcal{R}_b = \operatorname{cl}\mathcal{R}_c.}
\]
With a slight abuse of terminology, we call  $\mathcal{R}$ also the ``reachable set'', although formally from the equilibrium state $N$ one can only reach the interior points $\operatorname{int}\mathcal{R}$, while the boundary points $\partial\mathcal{R}$ are only asymptotically reachable. Thus, the use of incoherent control $n\ge 0$ essentially does not increase the size of the reachable set. However, the inability to use incoherent control (i.e., for $n\equiv0$) can affect the optimal steering time.

Time-optimal steering from an arbitrary initial state $\rho_0$ to a final state $\rho_1$ in the absence of incoherent control was studied in~\cite{LokutsievskiyPechen2}. It was proven that in the absence of incoherent control (i.e.,~$n\equiv 0$), the time-optimal control $u(t)$ consists of two pulses at $t=0$ and $t=T$ and an analytic function on the open interval~$(0,T)$. The optimal steering time $T$ satisfies
\[
	\frac1\gamma\ln\left(\frac{1-s \mu_0}{1-s \mu_1}\right) \le
	T\le
	\frac{\pi}{\omega} + \frac{{\mathrm e}}\gamma +
	\frac1\gamma \ln\left(\frac{1-s \mu_0}{1-s \mu_1}\right),
\]
where $\mu_i = \sqrt{2\operatorname{tr}\rho_i^2 - 1}$ for $i=0,1$ and $s=\operatorname{sgn}(\mu_1-\mu_0)$. The lower bound holds when $\mu_1<1$, and the upper bound when $\mu_1\le 1-\frac\pi2\frac\gamma\omega$.

The characteristic size of the reachable set $\mathcal{R}$ was estimated in~\cite{LokutsievskiyPechen1}, although its exact shape was not determined. To describe the reachable set,  we switch to coordinates $\rr=(r_x,r_y,r_z)$ in the Bloch ball, which are given by the formula
\[
\rho=\frac12\left(\mathbb{I} + r_x\sigma_x + r_y\sigma_y + r_z\sigma_z\right).
\]
In these coordinates, the phase space of the system is the unit Bloch ball $\BB=\{\rr:r_x^2+r_y^2+r_z^2\le 1\}$ and Eq.~\eqref{eq: main QCS} takes the form
\begin{equation}
\label{eq: QCS in Bloch ball}
\dot{\rr} = \omega f_0(\rr) + \kappa f_1(\rr)u + 2\gamma f_2(\rr)n,
\end{equation}
where (see~\cite{LokutsievskiyPechen2} for details\footnote{Note that \cite{LokutsievskiyPechen2} uses a slightly different convention ($r_z\mapsto -r_z$ and $u\mapsto -u$) corresponding to the opposite choice of the raising and lowering operators $\sigma_\pm$.})
\begin{align*}
	f_0(\rr) & =
	\begin{pmatrix}
		0&-1&0\\
		1&0&0\\
		0&0&0
	\end{pmatrix}\rr
	-
	\alpha
	\begin{pmatrix}
		\frac12 & 0 & 0\\
		0 & \frac12 & 0\\
		0 & 0 & 1
	\end{pmatrix}\rr
	-
	\alpha
	\begin{pmatrix}
		0\\0\\1
	\end{pmatrix},\\
	f_1(\rr) & =
	\begin{pmatrix}
		0&0&0\\
		0&0&-1\\
		0&1&0
	\end{pmatrix}\rr,
	\quad
	f_2(\rr)=
	-\begin{pmatrix}
		\frac12 & 0 & 0\\
		0 & \frac12 & 0\\
		0 & 0 & 1
	\end{pmatrix}\rr
\end{align*}
with $\alpha<1$.

\begin{theorem}[{see~\cite[Theorem 3]{LokutsievskiyPechen1}}]
\label{thm: lacunas}
	The reachable set $\mathcal{R}$ is a solid of revolution around the $r_x$-axis and for any $\omega>\gamma>0$ it holds that
	$$
		\left\{\rr:|\rr|\le 1 - c_1\alpha\right\}
		\subset
		\mathcal{R}
		\not\supset
		\left\{\rr:|\rr|\le 1 - c_2\alpha\right\},
	$$
	where the constants $c_1>c_2>0$ are of order 1.
\end{theorem}

\begin{figure}
\centering
\includegraphics[width=0.3\linewidth]{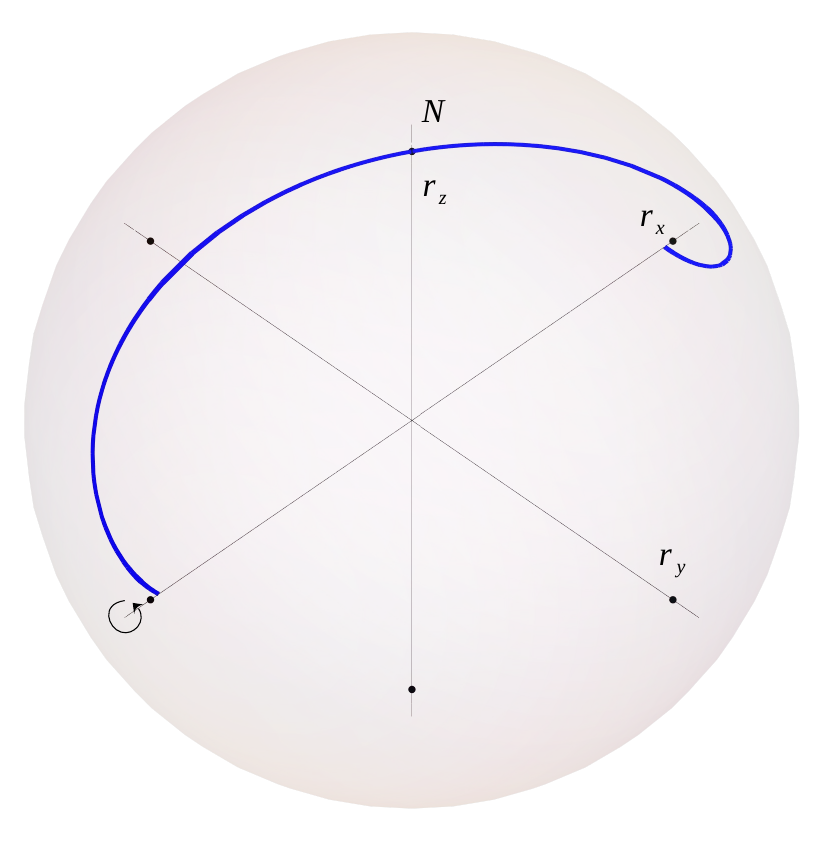}
\includegraphics[width=0.3\linewidth]{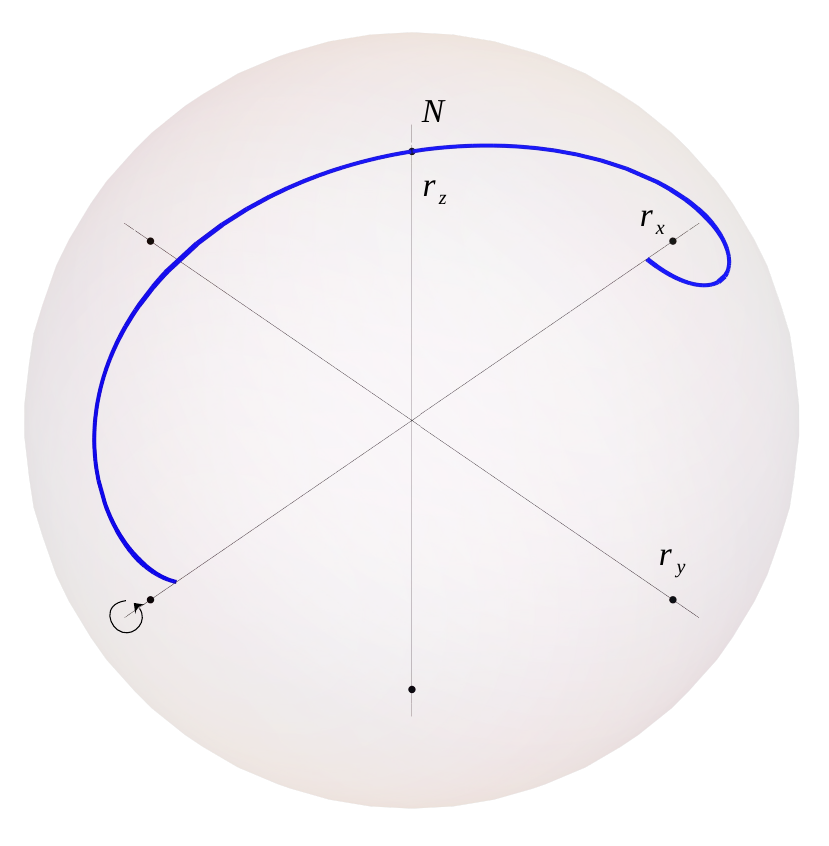}
\includegraphics[width=0.3\linewidth]{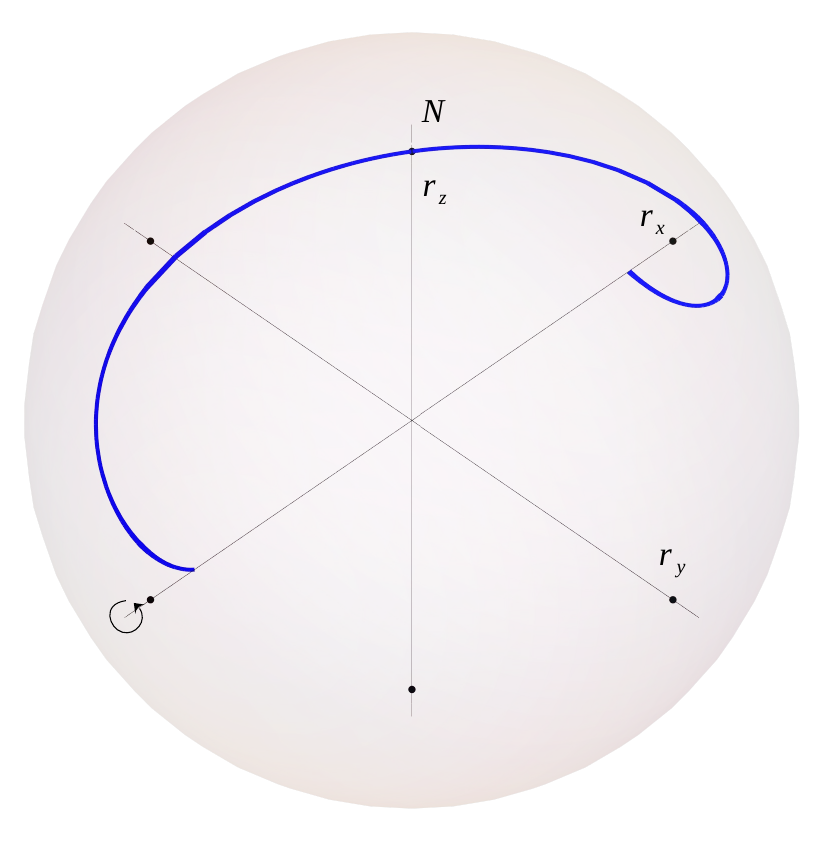}\\
\includegraphics[width=0.3\linewidth]{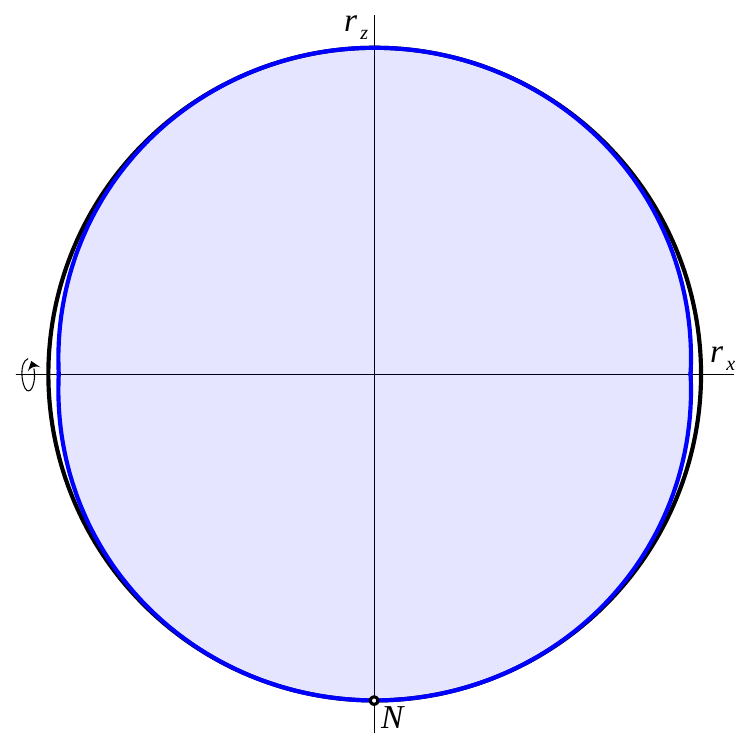}
\includegraphics[width=0.3\linewidth]{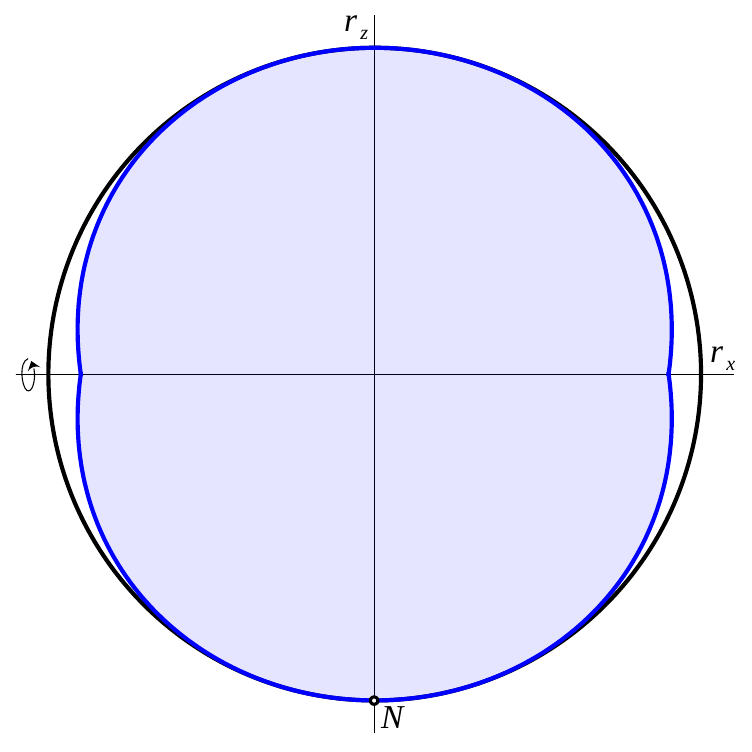}
\includegraphics[width=0.3\linewidth]{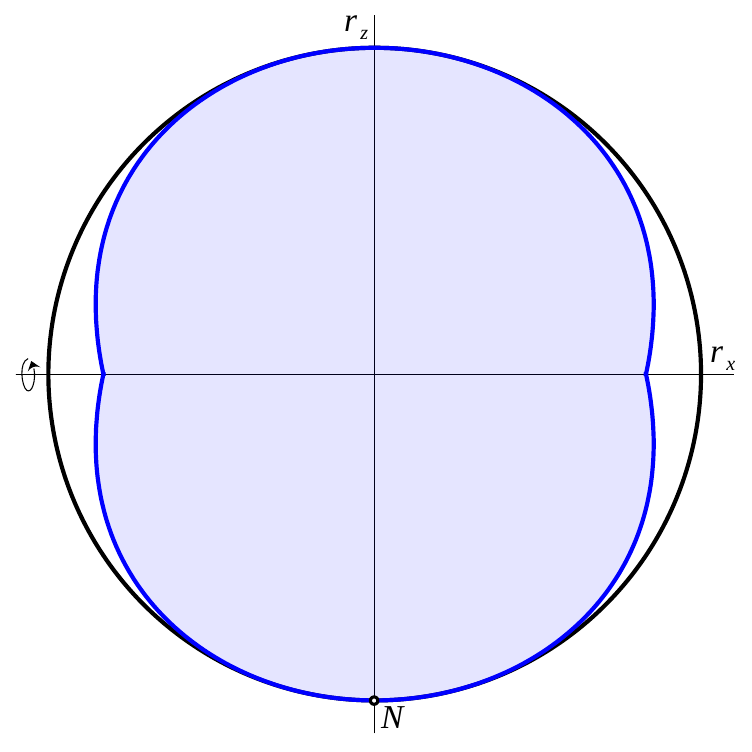}

\caption{\label{fig: 3d separatrix-1}The first row shows the curve $\Gamma\subset\partial\mathcal{R}$ in the Bloch ball for $\alpha=\gamma/\omega=0.1, 0.3, 0.5$ from left to right. The second row shows the intersection of the boundary $\partial\mathcal{R}$ with $r_xr_z$-plane for the same values of $\alpha$.}
\end{figure}

In other words, if $\alpha\ll 1$, then the reachable set is close to the entire Bloch ball $\BB$ (as it contains a ball of radius $1-c_1\alpha$), but never coincides with it (as it never contains the slightly larger ball of radius $1-c_2\alpha$). Thus, the boundary of the reachable set $\partial \mathcal{R}$ has a non-trivial structure, which is the main object of study in the present paper.

A priori, the boundary $\partial\mathcal{R}$ may not be a smooth two-dimensional surface. Here we prove that $\partial\mathcal{R}$ is a two-dimensional surface of revolution, which is smooth except for two conical singularities on the $r_x$-axis.

\textit{The set of reachable states}---Now we provide a precise description of $\partial\mathcal{R}$ as obtained in the present work. According to Theorem~\ref{thm: lacunas}, the boundary $\partial\mathcal{R}$ is a surface of revolution about the $r_x$-axis. Hence, in order to construct $\partial\mathcal{R}$, it is sufficient to find a curve $\Gamma\subset\partial\mathcal{R}$ such that rotating it about the $r_x$-axis forms $\partial\mathcal{R}$. This curve for some parameters is shown in Fig.~\ref{fig: 3d separatrix-1}.

\begin{theorem}
\label{thm: main} The boundary of the reachable set $\partial\mathcal R$ is obtained by rotating an analytic curve $\Gamma$ about the $r_x$-axis. The curve $\Gamma$
lies in the lower half of the Bloch ball $\BB \cap \{r_z \le 0\}$, passes through the equilibrium point $N$, and finishes at the $r_x$-axis (see 
Fig.~\ref{fig: 3d separatrix-1}). Curve $\Gamma$ away from $N$ is a trajectory of~\eqref{eq: QCS in Bloch ball} (or equivalently~\eqref{eq: main QCS}) with the control
\begin{equation}
\label{eq: u_b}
	\boxed{u = u_b = \frac\omega\kappa \frac{(1+r_z)(2 r_x r_z+2 \alpha r_y+\alpha r_yr_z)}{r_y^2+r_z^2+r_z^3}}
\end{equation}
\end{theorem}

This theorem immediately implies that $\partial R$ is analytic except two conical points on the $r_x$-axis.

The curve $\Gamma$ consists of the point $N$ and two symmetric branches, $\Gamma=\Gamma^+ \sqcup \Gamma^- \sqcup \{N\}$. The branches $\Gamma^\pm$ are trajectories of \eqref{eq: QCS in Bloch ball} with $u=u_b$. Moreover, both of these solutions tend to $N$ as $t\to-\infty$. Thus, $\Gamma^\pm$ are the unstable separatrices of the point $N$. We prove that the curve $\Gamma$ lies on the surface
\begin{equation}
\label{eq: S surface}
	\frac{r_z}{1+r_z} + \frac{2r_y^2+2r_z^2+\alpha r_x r_y}{r_x^2+r_y^2} = 0.
\end{equation}
This surface in fact has no singularity at $N$ and is smooth everywhere except at the origin (which always lies in the interior of the reachable set). Indeed, we show that in a neighborhood of $N$, the curve $\Gamma$ in cylindrical coordinates $(r_x,r_y,r_z)=(r_x,R\cos\theta,R\sin\theta)$ is given by the following series expansions:
\begin{align}
	\label{eq: R Taylor}
	R&=1-\frac12r_x^2 -\frac18\beta^2(3\beta^2+2)r_x^4 + O(r_x^6),
	\\
	\label{eq: theta Taylor}
	\theta&=-\frac\pi2 - \beta r_x + O(r_x^3),
\end{align}
where $\beta = \frac16(\alpha+\sqrt{12+\alpha^2})\simeq \frac1{\sqrt{3}}$.

A detailed proof of Theorem~\ref{thm: main} is provided in the Supplementary Online Material (SOM).

\textit{Approaching boundary states}---Now we show that in this system the boundary $\partial\mathcal{R}$ of the reachable set not only can be described, but can be approached with a precise control protocol. Constructing a control protocol to steer the qubit state along the boundary is a non-trivial problem since the control $u_b$ is singular at the equilibrium point $N$, and simply starting in a small vicinity of $N$  and applying the control $u_b$ may move far away from the curve $\Gamma$. For a working protocol, we need to prepare the state on the special torus-like surface $\mathcal T$ obtained by rotating about the $r_x$-axis a figure-eight curve defined  by (see Fig.~\ref{fig: final structure of boundary clear}):
\begin{equation}
\label{eq: 8-figure two equations}
\alpha r_x + 2r_y=(4+\alpha^2)r_y^2 + 2\alpha^2 r_z(1+r_z)=0.
\end{equation}
Once the state on $\mathcal T$ is prepared, we apply control $u_b$.

\begin{figure}[t]
	\centering
	\includegraphics[width=0.5\linewidth]{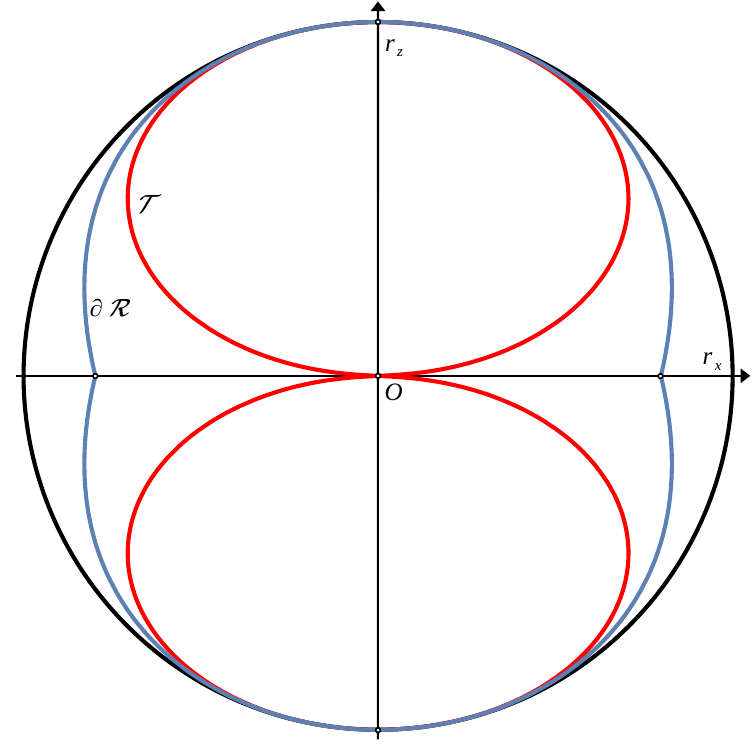}
	\caption{\label{fig: final structure of boundary clear}Structure of the intersections of $\partial\mathcal R$ and $\mathcal T$ with the $r_xr_z$-plane $\{r_y=0\}$ in the Bloch ball~$\BB$.}
\end{figure}

This can be done by a special control 
\[
\frac\kappa\omega u_p(t) = \begin{cases}
	\delta,&\text{for }t\in[0,\tau),\\
	\frac1{R^2}(r_xr_z+\alpha r_y+\frac12\alpha r_yr_z),&\text{for }t\in[\tau,T],\\
	\frac\kappa\omega u_b,&\text{for }t>T.
\end{cases}
\]
where $R=\sqrt{r_y^2+r_z^2}$. In the limit $\varepsilon:=\tau\delta\to0$, $|\delta|\to\infty$, $\tau\to0$, and $T\to\infty$, this control moves the state arbitrarily close to $\mathcal T$ near $N$, and then makes the subsequent trajectory approaching the boundary curve $\Gamma$. Finally, applying a control close to a $\delta$-impulse at a suitable time moves the system state arbitrarily close to any target state on the boundary of the reachable set. Thus the boundary found in Theorem~\ref{thm: main} can be attainable in the asymptotic sense with the specified protocol, and not merely an abstract boundary of the closure of the reachable set. Full details of this procedure are given in the End Matter.

\textit{Conclusions}---In this work, we have obtained a complete analytical description of the boundary of the reachable set for an open two-level quantum control system. This result provides fundamental limits on the ability to manipulate two-level open quantum systems---it shows which states can be created, at least in principle, using arbitrary controls within the considered class, and which states are unreachable in any finite or infinite time. These limits are determined by the ratio of the relaxation rate to the transition frequency of the two-level system. 

We have shown that the reachable set has highly nontrivial geometry as a subset of the Bloch ball. Its boundary is a surface of revolution that is smooth, indeed analytic, except at two conical singularities. We have also derived an explicit control protocol that generates the dynamics on this boundary, and constructed control protocol that steers the system arbitrarily close to any prescribed boundary state. Thus, the results provide both a geometric characterization of the reachable set of an open qubit and a constructive description of extremal state-preparation protocols.

These findings can be used across a wide range of quantum control problems, including establishing upper bounds on achievable fidelities in open-qubit quantum computing models, analyzing limits on spin dynamics in NMR and related systems, and quantifying the ultimate controllability limits imposed by dissipation in two-level quantum systems.

\textit{Acknowledgments}---This work was supported by the Theoretical Physics and Mathematics Advancement Foundation ``BASIS''. 

\textit{Data availability}—No data were created or analyzed in
this Letter.

\putbib[apssamp,main_prlNotes]
\makeatletter
\global\let\@FMN@list\@empty
\makeatother

\section{\label{sec: practical implementation}Derivation of the protocol for motion along the boundary of the reachable set}

Here we describe in details an explicit practical method for steering the qubit state along the boundary $\partial\mathcal{R}$, which solves the problem of reaching qubit states that are as close to the set of pure states as possible. Without loss of generality, we can assume that $\rr(0)=N$. Indeed, if this is not the case, it is sufficient to set $u=0$, and any state $\rr(t)$, as a solution of~\eqref{eq: QCS in Bloch ball} with $u\equiv n\equiv 0$, will tend to the equilibrium $N$ exponentially fast.

The difficulty of initiating the control at $N$ lies in the fact that the control $u_b$ in the Bloch ball $\BB$ has a singularity at the equilibrium point $N$ and on the axis $\{r_y=r_z=0\}$. Therefore, for the practical implementation of motion along the boundary $\partial\mathcal R$ starting from $N$, it is necessary to move slightly away from $N$. Unfortunately, if one moves away from $N$ using an arbitrary control and then switches to the control $u_b$, one can end up very far from the curve $\Gamma$. The reason is that in the Bloch ball $\BB$, there is a torus-like surface $\mathcal T$, obtained by rotating a figure-eight curve about the $r_x$-axis. The figure-eight curve lies in the plane $\alpha r_x + 2r_y=0$ and is defined in that plane by the formula (see Fig.~\ref{fig: final structure of boundary clear}):
\[
(4+\alpha^2)r_y^2 + 2\alpha^2 r_z(1+r_z)=0.
\]
This surface bounds the region of small-time local controllability, consisting of points whose neighborhoods are reachable in arbitrarily small time. (See the SOM).

Formally, the control $u_b$ is defined both inside and outside the surface $\mathcal T$, but physically, the control $u_b$ is meaningful only outside this surface. The point $N$ lies on the surface $\mathcal T$, and the surface $\mathcal T$ itself is tangent to the Bloch sphere at this point. Therefore, if one moves away from $N$ arbitrarily, the state $\rho$ may end up inside $\mathcal T$ (and most likely will), and the control $u_b$ will no longer be physically meaningful.

Below, we propose an explicit form of control that steers the state arbitrarily close to the boundary $\partial\mathcal{R}$. For details regarding the structure of this control, we refer the reader to the proofs in the SOM.

First, it is necessary to reach the surface $\mathcal T$ in the vicinity of $N$, and then switch to the control $u_b$. This can be done as follows. Let us choose a small $\varepsilon$ (possibly negative) and a large $|\delta|$ (with $\delta$ possibly negative), as well as $\tau,T>0$ such that $\tau\delta=\varepsilon$. The limiting case corresponds to $\varepsilon\to 0$, $|\delta|\to\infty$, $\tau\to0$, and $T\to+\infty$. For example, $\delta=\varepsilon^{-1}$, $\tau=\varepsilon^2$, and $T=|\varepsilon|^{-1}$. Although other choices are possible: the choice of parameters $\varepsilon$, $\tau$, $\delta$, $T$ affects the accuracy and the steering time. Let us set
\[
\frac\kappa\omega u_p(t) = \begin{cases}
	\delta,&\text{for }t\in[0,\tau),\\
	\frac1{R^2}(r_xr_z+\alpha r_y+\frac12\alpha r_yr_z),&\text{for }t\in[\tau,T],\\
	\frac\kappa\omega u_b,&\text{for }t>T.
\end{cases}
\]
Here, $R=\sqrt{r_y^2+r_z^2}$.

In the first interval $t\in[0,\tau)$, we approximate the Dirac $\delta$-function to obtain $\theta(\tau)=\frac\pi2+\varepsilon$, where $\theta=\arg(r_y+\mathrm{i}r_z)$. In the second interval, we apply a control such that $\dot\theta=0$, and, consequently, $(r_x,R)\to E(\frac\pi2+\varepsilon)$, where
\[
E(\varphi)=\begin{pmatrix}
	r_x(\varphi)\\ R(\varphi)
\end{pmatrix} 
= 
\frac{2\alpha\sin\varphi}{2 \alpha^2 + (4-\alpha^2) \cos^2\varphi}
\begin{pmatrix}
	2\cos\varphi\\ -\alpha
\end{pmatrix}.
\]
The subsequent application of the control $u_b$ will guide the state close to the curve $\Gamma$.

Under the control $u=u_p$, the state $(r_x,r_y,r_z)$ at time $T$ will be arbitrarily close to a point on the surface $\mathcal T$ with $r_z=-2\alpha^2\cos\varepsilon/(2\alpha^2 + (4-\alpha^2)\sin^2\varepsilon)$. For $t>T$, the state of the system will move arbitrarily close to the curve $\Gamma$. At a finite time $T_1>T$, this trajectory will reach the $r_x$-axis, i.e., $r_y(T_1)=r_z(T_1)=0$. For $t\in(T,T_1)$, this trajectory runs arbitrarily close to $\Gamma\subset\partial\mathcal R$. For $t>T_1$, this trajectory lies strictly inside $\mathcal R$. Therefore, to obtain an arbitrary point on the boundary $\partial\mathcal R$, one must apply a control close to a Dirac $\delta$-function at the required moment $s\in(T,T_1)$, for example:
\[
	\frac\kappa\omega u_p(t) = \frac\theta\sigma,\text{ for }t\in[s,s+\sigma],
\]
where $\theta$ is the required angle of rotation around the $r_x$-axis. Thus, one can obtain a final state $\rho(s+\sigma)$ arbitrarily close to any point $\rho_1\in\partial\mathcal R$ on the boundary of the reachable set. Note that it is impossible to remain at this point due to the drift in~\eqref{eq: QCS in Bloch ball}, so for $t>s+\sigma$ the state will inevitably drift away.

To transfer the initial state $\rho(0)=N$ to a final state $\rho_1$ in the interior of the reachable set, $\rho_1\in\operatorname{int}\mathcal{R}$, one must correctly choose the instant $s$ and ensure that $\theta(s+\sigma)=0$ or $\theta(s+\sigma)=\pi$. If after this, for $t>s+\sigma$, one uses
\begin{equation}
\label{eq: spiral control}
	\frac\kappa\omega u_p(t) = \frac14\alpha\sin2\theta + \frac1R(r_x\sin\theta+\alpha\cos\theta),
\end{equation}
then $\theta(t)\equiv\const$ for $t>s+\sigma$. In both cases ($\theta(t)=0$ and $\theta(t)=\pi$), the solutions of system~\eqref{eq: QCS in Bloch ball} will be logarithmic spirals tending towards the origin. Therefore, if a final state $\rho_1=(r_{x1},r_{y1},r_{z1})$ is given, the initial state $\rho_0=N$ can be transferred to a state $\rho(\tau)$, where $\tau\ge s+\sigma$, with the required $r_{x1}$ and $R_1=\sqrt{r_{y1}^2+r_{z1}^2}$. To reach $\rho_1$ in the interval $[\tau,\tau+\sigma]$, it is necessary to use an analogous to the $\delta$-pulse described above. Note that if $\sqrt{2\mathrm{tr}\,\rho_1^2-1}<1-\frac\pi2\alpha$, the scheme described above is not needed, and one can use a much simpler method proposed in the paper~\cite{LokutsievskiyPechen2}.

\end{bibunit}

\onecolumngrid

\appendix
\begin{bibunit}

\newpage 

\begin{center}
\Large\bf Supplementary material for the manuscript\\ "Fundamental limits on state preparation for an open qubit"
\end{center}

In this supplementary material, we provide all the proofs for the results of the work. Specifically, we investigate the structure of the reachable set for an open qubit subject to coherent and incoherent control. Recall that the dynamics of an open quantum system under coherent and incoherent control is often described by the Gorini--Kossakowski--Sudarshan--Lindblad (GKSL) master equation \cite{AlessandroBook2021}
\[
	\frac{{\mathrm{d}} \rho}{{\mathrm{d}} t} = -\frac{\mathrm{i}}{\hbar}[H_0+u(t)V,\rho] + \mathcal{L}(\rho),
\]
where $\rho$ is the density matrix, $H_0$ is the free Hamiltonian, $u(t)$ is the coherent control, $V$ is the corresponding interaction Hamiltonian, and $\mathcal{L}$ is the Lindblad superoperator describing dissipation. The incoherent control enters the equation through the last term $\mathcal{L}(\rho)$.

For a single open qubit, the GKSL master equation takes the following form\cite{Scully_Zubairy_1997}:
\begin{equation}
\label{eq: main QCS appendix}
\tag{\ref*{eq: main QCS}}
	\frac{\mathrm{d}\rho}{\mathrm{d}t} =
	-\frac{\mathrm{i}}{\hbar}\bigl[H_0+uV, \rho\bigr]
	+\gamma(n+1) \mathcal{D}[\sigma^-]\rho + \gamma n \mathcal{D}[\sigma^+]\rho.
\end{equation}
Here $\rho$ is the $2\times 2$ density matrix, $\sigma_{x,y,z}$ are the Pauli matrices, $\sigma^\pm=\frac12(\sigma_x\pm\mathrm{i}\sigma_y)$ are the raising and lowering operators, $H_0=\frac\omega2 \sigma_z$ and  $V=\frac\kappa2 \sigma_x$, and $\mathcal{D}[\sigma]\rho=\sigma\rho \sigma^\dagger - \frac{1}{2} \Big\{ \sigma^\dagger \sigma, \rho \Big\}$. The main parameters are the transition frequency $\omega$, the coherent control coupling coefficient $\kappa$, and the decoherence rate $\gamma$. The control variables are the coherent control $u(t)\in\R$ and the incoherent control $n(t)\ge 0$. Hereafter, we put $\hbar=1$.

If we set $u=n=0$ in Equation~\eqref{eq: main QCS appendix}, then every solution will approach the unique equilibrium state exponentially fast. In the basis $\{ |e\rangle, |g\rangle \}$, the equilibrium state is
$N = |g\rangle\langle g|$. Therefore, determining the structure of the reachable set starting from the initial state $N$ is of primary interest.

In the present work, we are interested in the boundary structure of the reachable set:
$$
	\boxed{\mathcal{R} = \Cl\mathcal{R}_b = \Cl\mathcal{R}_c},
$$
where
$$
	\mathcal{R}_b = \Big\{
		\rho_1 \mathrel{\Big|} \exists T\ge 0,,\,u(t)\in\R,\,n(t)\ge0\,:\,\rho(0)=N, \rho(T)=\rho_1
	\Big\}.
$$
$$
	\mathcal{R}_c = \Big\{
		\rho_1 \mathrel{\Big|} \exists T\ge 0,\,u(t)\in\R,\,n(t)\equiv 0\,:\,\rho(0)=N, \rho(T)=\rho_1
	\Big\},
$$
Note that $\Cl\mathcal{R}_b=\Cl\mathcal{R}_c$ and $\partial\mathcal{R}_b=\partial\mathcal{R}_c$ due to \cite[Theorem~2]{LokutsievskiyPechen1}. Therefore, without loss of generality, we assume $n\equiv 0$.

In order to investigate system~\eqref{eq: main QCS appendix}, let us write it in Bloch coordinates. Recall that Bloch coordinates $\rr=(r_x,r_y,r_z)$ on the Bloch ball $\BB=\{\rr:r_x^2+r_y^2+r_z^2\le 1\}$ are given by the formula
$$
	\rho=\frac12\left(\mathbb{I} + r_x\sigma_x + r_y\sigma_y + r_z\sigma_z\right).
$$
Direct computations lead to the following control system
\begin{equation}
\label{eq: QCS in Bloch ball appendix}
\tag{\ref*{eq: QCS in Bloch ball}}
	\dot{\rr} = \omega f_0(\rr) + \kappa f_1(\rr)u + 2\gamma f_2(\rr)n
\end{equation}
on the Bloch ball, where (see~\cite{LokutsievskiyPechen2} for details\footnote{Note that \cite{LokutsievskiyPechen2} uses a slightly different convention ($r_z\mapsto -r_z$ and $u\mapsto -u$) corresponding to the opposite choice of the raising and lowering operators $\sigma_\pm$.})
\begin{align*}
	f_0(\rr) & =
	\begin{pmatrix}
		0&-1&0\\
		1&0&0\\
		0&0&0
	\end{pmatrix}\rr
	-
	\alpha
	\begin{pmatrix}
		\frac12 & 0 & 0\\
		0 & \frac12 & 0\\
		0 & 0 & 1
	\end{pmatrix}\rr
	-
	\alpha
	\begin{pmatrix}
		0\\0\\1
	\end{pmatrix},\\
	f_1(\rr) & =
	\begin{pmatrix}
		0&0&0\\
		0&0&-1\\
		0&1&0
	\end{pmatrix}\rr,
	\qquad
	f_2(\rr)=
	-\begin{pmatrix}
		\frac12 & 0 & 0\\
		0 & \frac12 & 0\\
		0 & 0 & 1
	\end{pmatrix}\rr.
\end{align*}
Here and throughout the rest of the paper, we denote
$$
	\boxed{\alpha=\frac\gamma\omega\ll 1}.
$$

\begin{figure}
	\centering
	\includegraphics[width=0.3\linewidth]{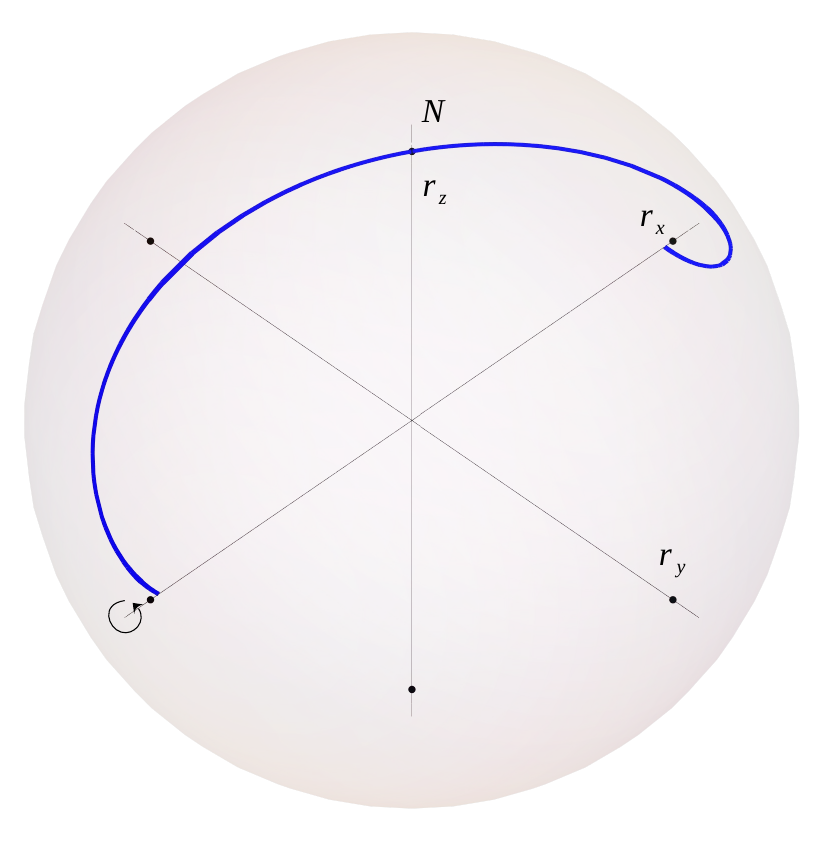}
	\includegraphics[width=0.3\linewidth]{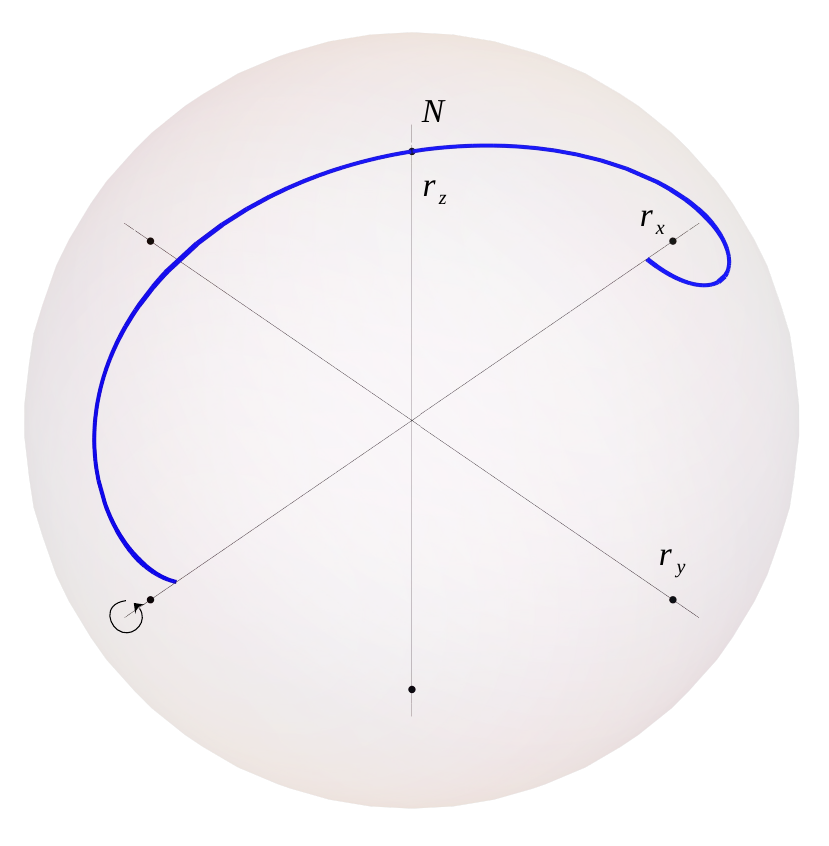}
	\includegraphics[width=0.3\linewidth]{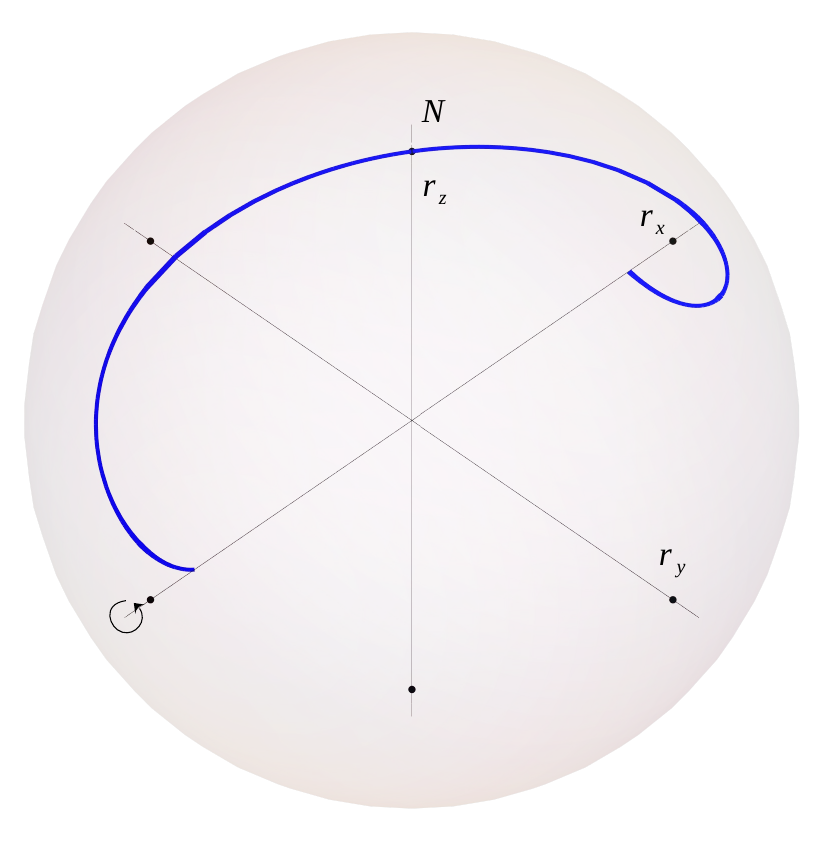}
	
	\caption{\label{fig: 3d separatrix}Curve $\Gamma\subset\partial\mathcal{R}$ in the Bloch ball for $\gamma/\omega=0.1, 0.3, 0.5$ from left to right. The boundary $\partial\mathcal{R}$ is a surface of revolution obtained by rotating $\Gamma$ about the $r_x$-axis. The curve $\Gamma$ is symmetric wrt reflection $(r_x,r_y,r_z)\mapsto(-r_x,-r_y,r_z)$}
\end{figure}

We prove the following main result:\footnote{Theorem numbering is consistent with the main text.}
\setcounter{theorem}{3}

\begin{theorem}
	In the lower half of the Bloch ball $\BB \cap \{r_z \le 0\}$, there exists an analytic curve $\Gamma$ (see Fig.~\ref{fig: 3d separatrix}) that passes through the equilibrium point $N$, finishes at the $r_x$-axis, and is, away from $N$, a trajectory of~\eqref{eq: QCS in Bloch ball appendix} (or equivalently~\eqref{eq: main QCS appendix}) with the control\footnote{The curve $\Gamma$ consists of the point $N$ and two symmetric branches, $\Gamma=\Gamma^+ \sqcup \Gamma^- \sqcup \{N\}$. The branches $\Gamma^\pm$ are trajectories of \eqref{eq: QCS in Bloch ball appendix} with $u=u_b$. Moreover, both of these solutions tend to $N$ as $t\to-\infty$. Thus, $\Gamma^\pm$ are the unstable separatrices of the point $N$.}
	\begin{equation}
	\label{eq: u_b appendix}
	\tag{\ref*{eq: u_b}}
		\boxed{u = u_b = \frac\omega\kappa \frac{(1+r_z)(2 r_x r_z+2 \alpha r_y+\alpha r_yr_z)}{r_y^2+r_z^2+r_z^3}}
	\end{equation}
	The boundary of the reachable set $\mathcal{R}$ is obtained from $\Gamma$ by rotating the (connected) part of the curve $\Gamma$ about the $r_x$-axis (see Fig.~\ref{fig: 3d separatrix}).
\end{theorem}

Moreover, we prove that the curve $\Gamma$ lies on the surface
\begin{equation}
\label{eq: S surface appendix}
\tag{\ref*{eq: S surface}}
	\frac{r_z}{1+r_z} + \frac{2r_y^2+2r_z^2+\alpha r_x r_y}{r_x^2+r_y^2} = 0.
\end{equation}
We show that in a neighborhood of $N$, the curve $\Gamma$ in cylindrical coordinates $(r_x,r_y,r_z)=(r_x,R\cos\theta,R\sin\theta)$ is given by the following series expansions:
\begin{align}
\label{eq: R Taylor appendix}
\tag{\ref*{eq: R Taylor}}
	R&=1-\frac12r_x^2 -\frac18\beta^2(3\beta^2+2)r_x^4 + O(r_x^6),
	\\
\label{eq: theta Taylor appendix}
\tag{\ref*{eq: theta Taylor}}
	\theta&=-\frac\pi2 - \beta r_x + O(r_x^3),
\end{align}
where $\beta = \frac16(\alpha+\sqrt{12+\alpha^2})\simeq \frac1{\sqrt{3}}$. Note that $\beta\ge \frac1{\sqrt{3}}$, so
$$
	r_x^2+r_y^2+r_z^2=r_x^2+R^2=1-\frac14(\beta^2(3\beta^2+2)-1)r_x^4+O(r_x^6)\le 1
$$
as expected.

\setcounter{equation}{7}

\section{Reduction of the system dimension}

Following~\cite{LokutsievskiyPechen1}, we use a key idea for reducing the dimensionality of system~\eqref{eq: main QCS appendix}. Since the use of incoherent control $n\ge 0$ does not change the boundary of the reachable set, we will assume without loss of generality that $n\equiv0$. Building on the framework established in~\cite{LokutsievskiyPechen1,LokutsievskiyPechen2}, we switch to cylindrical coordinates, which cleverly decouples the control $u$ from the radial and longitudinal components. Namely, if we introduce the variables\footnote{In the previous papers \cite{LokutsievskiyPechen1} and \cite{LokutsievskiyPechen2}, we have used slightly different notation. Namely, $r_z\mapsto -r_z$, $u\mapsto -u$, and $\theta\mapsto-\theta$.}
$$
	r_y=R\cos\theta, \quad r_z=R\sin\theta,
$$
system~\eqref{eq: main QCS appendix} (or, equivalently,~\eqref{eq: QCS in Bloch ball appendix}) can be rewritten in a form where the control $u$ does not enter the right-hand sides of $\dot r_x$ and $\dot R$:
\begin{equation}
\label{eq: QCS in cylindrical coordinates}
	\dfrac{1}{\omega}
	\begin{pmatrix}
		\dot r_x\\
		\dot R 
	\end{pmatrix}
	=
	\begin{pmatrix}
		-\frac12\alpha & -\cos\theta\\
		\cos\theta  & -\frac12\alpha(1+\sin^2\theta)\\
	\end{pmatrix}
	\begin{pmatrix}
		r_x\\ 
		R
	\end{pmatrix}
	-
	\alpha
	\begin{pmatrix}
		0\\
		\sin\theta
	\end{pmatrix},
\end{equation}
but enters only the right-hand side of $\dot\theta$:
\begin{equation}
\label{eq: QCS in cylindrical coordinates theta}
	\dfrac{1}{\omega}\dot \theta =\frac\kappa\omega u-
	\frac14\alpha\sin2\theta-
	\frac1R\Big(
	r_x\sin\theta+
	\alpha\cos\theta
	\Big).
\end{equation}

Note that the phase space of system~\eqref{eq: QCS in cylindrical coordinates} is the disk $r_x^2 + R^2\le 1$, which can be naturally associated with the disk in the plane $(r_x,r_y=0,r_z=R)$ within the Bloch ball $\BB$. Each point $(r_x,R)$ in this disk represents an entire circle (of the form $r_y^2+r_z^2=R^2$) or a point (if $r_y^2+r_z^2=0$) in the original Bloch ball~$\BB$. Note that we omit the restriction $R\ge 0$ as system~\eqref{eq: QCS in cylindrical coordinates} is symmetric with respect to the usual symmetry $(R,\theta)\mapsto(-R,\theta+\pi)$

\begin{remark}
	Let us give a physical interpretation of the transition to cylindrical coordinates (for precise mathematical results, we refer to the papers~\cite{LokutsievskiyPechen1} and \cite{LokutsievskiyPechen2}). Specifically, applying an unbounded control $u$ in system~\eqref{eq: main QCS appendix} rotates the point $(r_x,r_y,r_z)$ along the circle $r_y^2+r_z^2=R^2$. Therefore, in system~\eqref{eq: QCS in cylindrical coordinates}, $\theta$ can effectively be chosen arbitrarily and thus can be considered a new control. Moreover, it is clear that the reachable set $\mathcal{R}$ is the solid of revolution obtained by rotating the reachable set $\bar{\mathcal{R}}$ of system~\eqref{eq: QCS in cylindrical coordinates} (where $\theta$ acts as the control) about the $r_x$-axis:
	$$
		\bar{\mathcal{R}} = \Cl\left\{
		(r_x^1,R^1)\,\big|\,\exists T\ge 0\text{ and }(r_x(t),R(t),\theta(t)):r_x(0)=0, R(0)=1, r_x(T)=r_x^1,R(T)=R^1
		\right\},
	$$
	where the initial point $N$ corresponds to $r_x(0)=0$ and $R(0)=1$, with $\theta(t)$ acting as a measurable control function that dictates the dynamics of the state $(r_x(t),R(t))$ in~\eqref{eq: QCS in cylindrical coordinates}. It is convenient to identify $\bar{\mathcal{R}}$ with the intersection of $\mathcal{R}$ with the plane $r_y=0$ (see details in~\cite{LokutsievskiyPechen1}):
	$$
		\boxed{\bar{\mathcal{R}} = \mathcal{R} \cap \{r_y=0\} \subset \R^2 = \{(r_x,r_y=0,r_z=R)\}}.
	$$
\end{remark}

Thus, to find the boundary of the reachable set $\mathcal{R}$ in the Bloch ball $\BB$, it is sufficient to find the boundary of the set $\bar{\mathcal{R}}$ in the disk $r_x^2+R^2\le 1$ and rotate it around the $r_x$-axis.

\section{Set of points of local controllability}
\begin{figure}
	\centering
	\begin{minipage}[t]{0.4\textwidth}
		\centering
		\includegraphics[width=\linewidth]{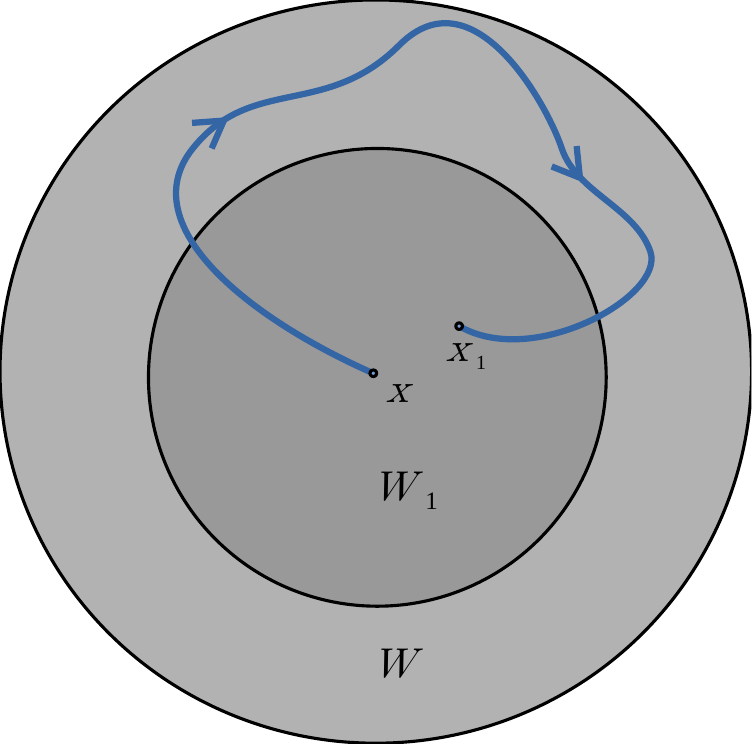}
		\caption{\label{fig: local controllability}STL-local controllability}
	\end{minipage}
	$\qquad\qquad$
	\begin{minipage}[t]{0.4\textwidth}
		\centering
		\includegraphics[width=\linewidth]{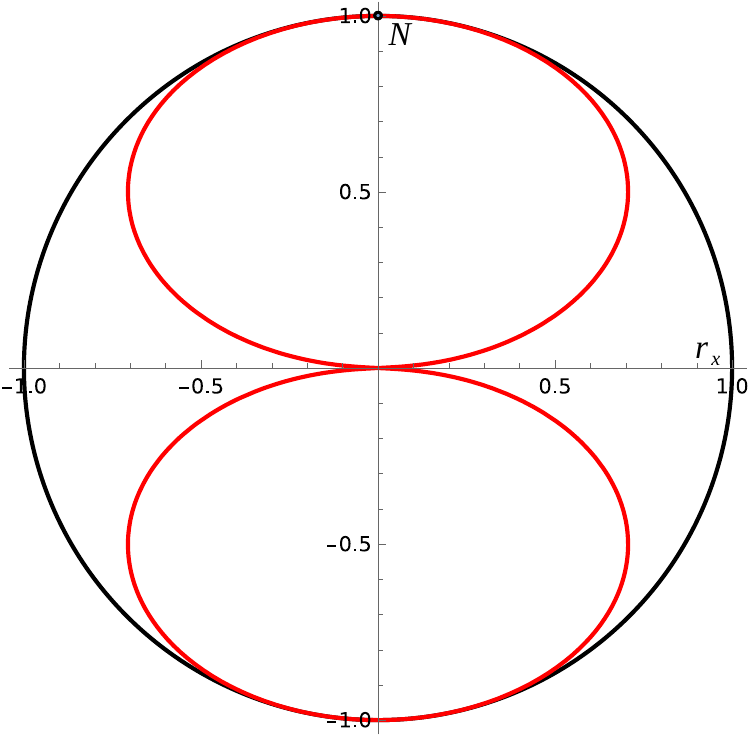}
		\caption{\label{fig: 8 curve}figure-eight curve $E(\varphi)$ for $\alpha=0.1$}
	\end{minipage}
\end{figure}

Recall that a control system is called STL-locally controllable\footnote{STL = small-time localized} at some point $x$ if for any neighborhood $W$ of $x$ and any $T>0$, there exists a smaller neighborhood $W_1\subset W$ of $x$ such that every point $x_1\in W_1$ is reachable from $x$ in time less than $T$ along a trajectory that does not leave $W$ (see Fig.~\ref{fig: local controllability}). For a detailed introduction to the types of local controllability, we refer the reader to~\cite{BoscainControllability}.

In this section, we find the boundary of the set of STL-locally controllable points for system~\eqref{eq: QCS in cylindrical coordinates}. First, we provide a geometric characterization of these points.

As mentioned above, the control $u$ rotates the solution of system~\eqref{eq: QCS in Bloch ball appendix} around the $r_x$-axis. Since the control $u$ is unbounded, the entire Bloch ball $\BB$ is partitioned into equivalence classes with respect to these rotations: circles $\Sigma=\{r_x=\const,r_y^2+r_z^2=R^2=\const>0\}$ and points $\{r_x=\const,r_y=r_z=0\}$. Fix some non-degenerate circle $\Sigma$ and a point $\rho=(r_x,r_y,r_z)\in\Sigma$ on it. The set of admissible velocities at this point (i.e., all possible values of the right-hand side of the control system~\eqref{eq: QCS in Bloch ball appendix} for various $u$ and fixed $(r_x,r_y,r_z)$) is a line $\ell(\rho)$ parallel to the tangent line of the circle $\Sigma$ at the point $\rho$. Choose a vector $v(\rho)\in \ell(\rho)$ of minimal length, and compute the linking number $\sigma$ of the curves $\Sigma$ and $\Sigma_v=\{\rho+\varepsilon v(\rho)\mid\rho\in\Sigma\}$ for a sufficiently small $\varepsilon$. It is easy to see that if $0\not\in\ell(\rho)$ for all $\rho\in\Sigma$, then the curves $\Sigma$ and $\Sigma_v$ do not intersect, and the linking number $\sigma$ is well-defined (as long as the number $\varepsilon$ is sufficiently small). We claim that three cases are possible:

\begin{itemize}
	\item $0\not\in\ell(\rho)$ for all $\rho\in\Sigma$ and $\sigma=1$. Then the vector $\varepsilon v(\rho)$ completes one revolution around~$\Sigma$ as $\rho$ traverses $\Sigma$. In this case, the corresponding point $(r_x,R)$ lies in the interior of the set of STL-locally controllable points.
	
	\item $0\not\in\ell(\rho)$ for all $\rho\in\Sigma$ and $\sigma=0$. Then the vector $\varepsilon v(\rho)$ completes zero revolutions around~$\Sigma$ as $\rho$ traverses $\Sigma$. In this case, the corresponding point $(r_x,R)$ does not lie in the set of STL-locally controllable points.
	
	\item $0\in\ell(\rho)$ for at least one $\rho\in\Sigma$. Then the linking number $\sigma$ is not defined, as the curves $\Sigma$ and $\Sigma_v$ intersect. In this case, the corresponding point $(r_x,R)$ lies on the boundary of the set of STL-locally controllable points.
\end{itemize}

This is the geometric interpretation of the following theorem.

\begin{atheorem}
\label{thm: local controllability}
	Let
	$$
	E(\varphi)=\begin{pmatrix}
		r_x(\varphi)\\ R(\varphi)
	\end{pmatrix} 
	= 
	\frac{2\alpha\sin\varphi}{2 \alpha^2 + (4-\alpha^2) \cos^2\varphi}
	\begin{pmatrix}
		2\cos\varphi\\ -\alpha
	\end{pmatrix},
	\quad \varphi\in\R/2\pi\mathbb{Z}
	$$
	be a parametrically defined closed curve on the $(r_x,R)$ plane, forming a smooth figure-eight with a single self-intersection point (see Fig.~\ref{fig: 8 curve}). This curve bounds the set of STL-locally controllable points of the system\footnote{In fact, the interiors of the sets of ST-local controllability and L-local controllability coincide here, but this is not essential for what follows.}~\eqref{eq: QCS in cylindrical coordinates}. Moreover, the set of STL-locally controllable points is contained in the reachable set $\bar{\mathcal{R}}$.
\end{atheorem}

\begin{proof}	
	Fix a point $(r_x,R)$ with $r_x^2 + R^2 \le 1$. Then for $\theta\in[0,2\pi]$, the right-hand side of \eqref{eq: QCS in cylindrical coordinates} is a closed curve. This curve is the boundary of some convex compact set $U(r_x,R)$ (see \cite{LokutsievskiyPechen1}), which we will call the set of admissible velocities. According to the saturation method, the reachable sets of the control system~\eqref{eq: QCS in cylindrical coordinates} (with control $\theta$) and the extended system $(\dot r_x,\dot R)\in U(r_x,R)$ coincide (see~\cite{AgrachevSachkov}). Therefore, if $0\in \Int U(r_x,R)$, then system~\eqref{eq: QCS in cylindrical coordinates} is STL-locally controllable at the point $(r_x,R)$. If $0\not\in U(r_x,R)$, then it can be strictly separated from $U(r_x,R)$, and therefore such a point is not a point of STL-local controllability. If $0\in\partial U(r_x,R)$, then such a point is actually a point of L-local, but not of STL-local controllability (see~\cite[Ex.~5]{BoscainControllability}), but this is not important for what follows.

	Thus, the points of the disk $r_x^2+R^2\le 1$ are divided into two open sets depending on the position of the zero velocity relative to the set of admissible velocities $U(r_x,R)$ (either $0\in\Int U(r_x,R)$ or $0\not\in U(r_x,R)$), which are separated by the curve of points $(r_x,R)$ such that $0\in\partial U(r_x,R)$. This curve is easy to find: if $0\in\partial U(r_x,R)$, then there exists an angle $\theta$ such that the right-hand side of system~\eqref{eq: QCS in cylindrical coordinates} becomes zero, i.e.,
	$$
	\begin{pmatrix}
		-\frac12\alpha & -\cos\theta\\
		\cos\theta  & -\frac12\alpha(1+\sin^2\theta)\\
	\end{pmatrix}
	\begin{pmatrix}
		r_x\\ 
		R
	\end{pmatrix}
	-
	\alpha
	\begin{pmatrix}
		0\\
		\sin\theta
	\end{pmatrix}
	=
	\begin{pmatrix}
		0\\
		0
	\end{pmatrix}.
	$$
	From this, we find the expressions for the curve $E$ stated in the theorem. We replace the parameter $\theta$ with $\varphi$ to avoid confusion with the control parameter.
	
	Let us show that the curve $E(\varphi)$ is smooth. For this, denote the scalar factor in the definition of the curve $E(\varphi)$ by $g(\varphi) = 2\alpha\sin\varphi/(2 \alpha^2 + (4-\alpha^2) \cos^2 \varphi)$. If $\varphi$ is a point of non-smoothness, then
	$$
	g(\varphi)\begin{pmatrix}
		-2\sin\varphi\\0
	\end{pmatrix} 
	+ 
	g'(\varphi)\begin{pmatrix}
		2\cos\varphi\\ -\alpha
	\end{pmatrix} 
	=
	\begin{pmatrix}
		0\\0
	\end{pmatrix}.
	$$
	Since $\alpha>0$, we get $g'(\varphi)=0$ and, consequently, $g(\varphi)\sin\varphi=0$. From the definition of the function $g(\varphi)$, we obtain that $\sin\varphi=0$, i.e., $\varphi=0$ or $\pi$. Direct calculation shows that $g'(0)\ne 0$ and $g'(\pi)\ne 0$, i.e., the specified curve has no points of non-smoothness.
	
	Let us now find the self-intersection points. Suppose the curve $E(\varphi)$ passes through a common point $(r_x,R)$ for two different parameters $\varphi\ne\tilde\varphi$. Then
	$$
	g(\varphi)\begin{pmatrix}
		2\cos\varphi\\-\alpha
	\end{pmatrix} 
	=
	g(\tilde\varphi)\begin{pmatrix}
		2\cos\tilde\varphi\\-\alpha
	\end{pmatrix}.
	$$
	Since $\alpha>0$, we get $g(\varphi)=g(\tilde\varphi)$. Consider two cases: $g(\varphi)\ne 0$ and $g(\varphi)=0$. In the first case, where $g(\varphi)\ne 0$, we obtain $\cos\varphi=\cos\tilde\varphi$. Thus, the equality $g(\varphi)=g(\tilde\varphi)$ implies $\sin\varphi=\sin\tilde\varphi$, and, consequently, in this case $\varphi=\tilde\varphi$, and there is no self-intersection point. In the second case, where $g(\varphi)=g(\tilde\varphi)=0$, we get $\sin\varphi=\sin\tilde\varphi=0$, that is, $\varphi=0$ and $\tilde\varphi=\pi$ (or vice versa), and the curve has a self-intersection at the origin.
	
	Let us show that if a point $(r_x,R)$ is located inside the region bounded by the curve $E(\varphi)$, then the zero velocity lies inside the set $U(r_x,R)$ and vice versa. The curve $E(\varphi)$ divides the plane into three connected components, as it has exactly one self-intersection point (see Fig.~\ref{fig: 8 curve}). Each connected component either consists entirely of STL-locally controllable points or contains no such points. The interior of the specified curve is the union of two open connected regions. The first region contains the vertical segment joining $(0,0)$ and $(0,1)$, and the second contains the vertical segment joining $(0,0)$ and $(0,-1)$. Both interior regions consist of STL-locally controllable points. Indeed, consider the $R$-axis ($r_x=0$) with $|R|<1$: for $\theta=\pm\pi/2$, the velocity vector on the right-hand side of~\eqref{eq: QCS in cylindrical coordinates} is non-zero and vertical, pointing up or down depending on the sign of the product $R\theta$. Thus, $0\in U(0,R)$ due to convexity. Since $0\not\in\partial U(0,R)$ at these points, we obtain $0\in\Int U(0,R)$. The exterior region contains no STL-locally controllable points, because on the $r_x$-axis for $r_x\ne 0$, the set $U(r_x,0)$ is a vertical segment that does not contain 0.
	
	It remains to show that both interior connected components lie in $\bar{\mathcal{R}}$. It is clear that each such component either lies entirely in $\bar{\mathcal R}$ or entirely outside it (since they consist of points of local controllability). Indeed, setting $\theta=\frac\pi2$ and $r_x=0$ in system~\eqref{eq: QCS in cylindrical coordinates}, we have $\dot r_x=0$ and $\dot R = -\alpha(R+1)$. Thus, the control $\theta=\frac\pi2$ moves the point $(0,1)$ vertically downwards, approaching the point $(0,-1)$ exponentially fast. This motion passes through both interior components. Hence, they both lie in $\bar{\mathcal{R}}$.
\end{proof}

\begin{remark}
\label{rm: fixed points}
	The curve $E(\varphi)$ has the following meaning from the point of view of practical control: if we set $\theta(t)=\varphi=\const$ in system~\eqref{eq: QCS in cylindrical coordinates}, then the resulting affine differential equation for $(r_x,R)$ will have a unique equilibrium point $E(\varphi)$, and this point will be exponentially attracting.
\end{remark}

\section{Exact shape of the boundary of the reachable set}

The boundary of the reachable set $\bar{\mathcal R}$ from the point $(0,1)$ in the two-dimensional disk $\{r_x^2+R^2\le1\}$ is described by the following system of differential equations:

\begin{figure}
	\centering
	\includegraphics[width=0.3\linewidth]{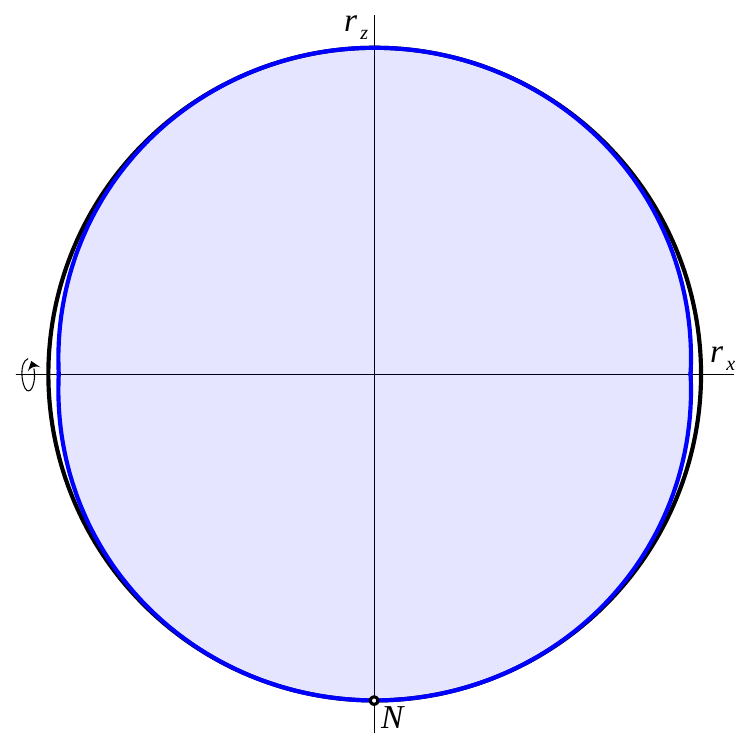}
	\includegraphics[width=0.3\linewidth]{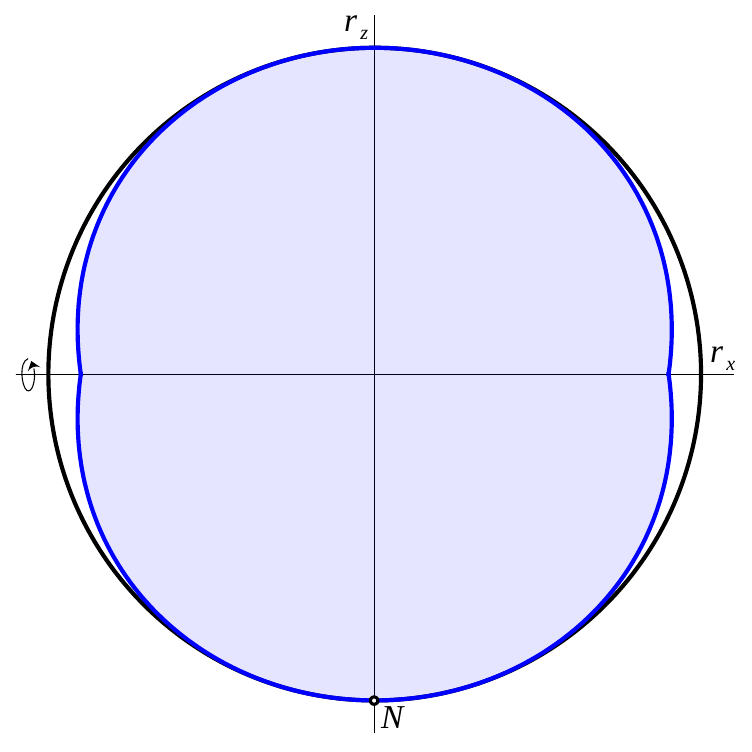}
    \includegraphics[width=0.3\linewidth]{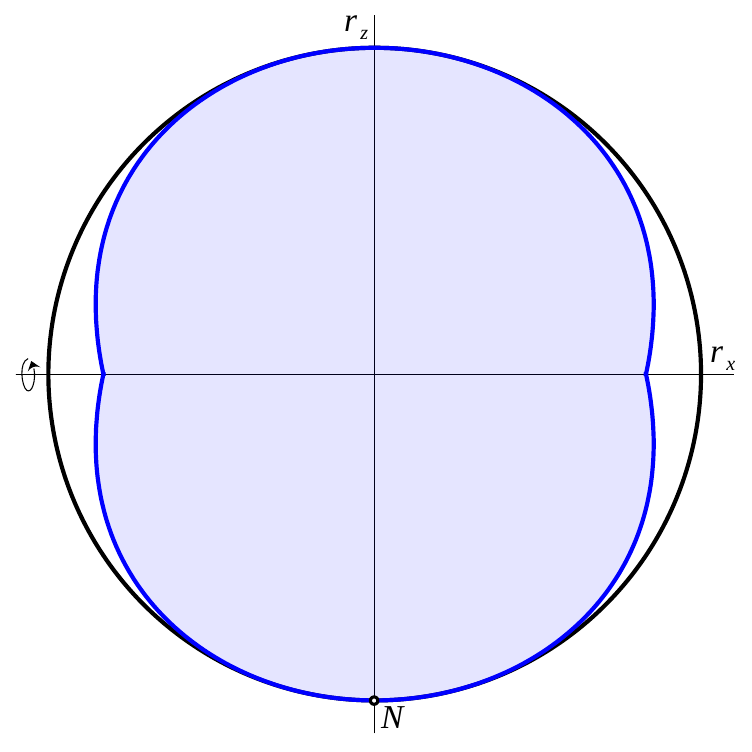}

	\caption{\label{fig: 2d separatrix}Reachable sets in the disk $r_x^2+R^2\le 1$ for $\alpha=0.1, 0.3, 0.5$ from left to right.}
\end{figure}

\begin{atheorem}
\label{thm: reachable set boundary}
	The reachable set $\bar{\mathcal R}$ contains the figure-eight curve $E$ and its interior (see Theorem~\ref{thm: local controllability}). The set $\bar{\mathcal{R}}$ is symmetric with respect to the reflections $r_x\to -r_x$ and $R\to-R$. The boundary $\partial\bar{\mathcal {R}}$ is a closed curve, has two points of non-smoothness on the $r_x$-axis and is analytic everywhere else (see Fig.~\ref{fig: 2d separatrix}). At all points except the two sigular ones, the boundary $\partial\bar{\mathcal {R}}$ satisfies the differential equation~\eqref{eq: QCS in cylindrical coordinates}, where the angle $\theta$ is a solution of the equation
	\begin{equation}
		\label{eq: theta main equation}
		(1+R\sin\theta)(2R+\alpha r_x \cos\theta) + (R^2 \cos^2\theta+r_x^2)\sin\theta = 0.
	\end{equation}
	Near the point $(r_x,R)=(0,1)$, the boundary $\partial\bar{\mathcal{R}}$ has the asymptotic expansions~\eqref{eq: R Taylor appendix} and~\eqref{eq: theta Taylor appendix}.
\end{atheorem}

The proof is based on several lemmas and is closely related to ideas presented in~\cite[Chapter 2]{Davydov}.

\begin{lemma}
\label{lemma: boundary is outside the eight}
	The interior of the reachable set $\Int\bar{\mathcal{R}}$ contains the interior of the figure-eight curve $E$ and the entire curve $E$ except for two points: the north and south poles $(r_x,R)=(0,\pm 1)$.
\end{lemma}

\begin{proof}
	Indeed, the region inside $E$ lies in $\bar{\mathcal{R}}$ by Theorem~\ref{thm: local controllability}. Let us show that every point of $E$ (except for the points $(r_x,R)=(0,\pm1)$) lies in the interior of $\bar{\mathcal{R}}$. Recall that we denoted by $U(r_x,R)$ the convex set bounded by the curve given by the right-hand side of~\eqref{eq: QCS in cylindrical coordinates} for $\theta\in[0,2\pi]$. Thus, for any $\varphi\ne0,\pi$ at the point $E(\varphi)$, the cone $C(r_x,R)=\R_+U(r_x,R)$ is a half-space by the definition of the curve $E$. Note that for $R\ne 0$, the set $U(r_x,R)$ is strictly convex with a smooth boundary (see \cite{LokutsievskiyPechen1}), and if $(r_x,R)=E(\varphi)$, then $0\in\partial U(r_x,R)$. Let us check that the line defining the half-space $C(E(\varphi))$ is transversal to the curve $E$ itself at the point $E(\varphi)$ --- this is obviously sufficient to prove the required statement $E(\varphi)\in\Int\bar{\mathcal{R}}$ for $\varphi\ne 0$ and $\pi$. If $\varphi=0$ or $\pi$, then $E(\varphi)=(0,0)\in\Int\bar{\mathcal R}$ by Theorem~\ref{thm: lacunas} in the text of the main paper (see also \cite{LokutsievskiyPechen1}).
	
	Thus, at the point $E(\varphi)$, the line defining $C(E(\varphi))$ is tangent to the set $U(E(\varphi))$ at the point $\theta=\varphi$, and therefore is directed along the vector:
	$$
	\frac1\omega\begin{pmatrix}
		\partial\dot r_x/\partial\theta\\\partial\dot R/\partial\theta
	\end{pmatrix}
	=
	\begin{pmatrix}
		0 & \sin\theta\\
		-\sin\theta  & -\frac12\alpha\sin2\theta\\
	\end{pmatrix}
	E(\theta)
	-
	\alpha
	\begin{pmatrix}
		0\\
		\cos\theta
	\end{pmatrix}.
	$$
	It is easy to check that this vector is parallel to $E'(\theta)$ only when $\theta=\pm\frac\pi2$. Indeed,
	$$
	\frac1\omega\left[
	\begin{pmatrix}
		\partial\dot r_x/\partial\theta\\\partial\dot R/\partial\theta
	\end{pmatrix} \times E'(\theta)\right]=
	-
	\frac{4 \alpha^2 (4+\alpha^2) \cos^3\theta (4+\alpha^2 \cos 2\theta)}{((4-\alpha ^2) \cos^2\theta+2 \alpha^2)^3}
	$$
	Therefore, for $\theta\ne\pm\frac\pi2$, the right-hand side is not zero, and hence the line is transversal to the figure-eight curve $E$.
	
\end{proof}

Thus, the boundary of the reachable set $\partial\bar{\mathcal{R}}$ can intersect the figure-eight curve $E$ only at the poles $(r_x,R)=(0,\pm1)$.

In the following lemma, we will prove that the boundary $\partial\bar{\mathcal{R}}$ is a locally Lipschitz curve. This statement is very important because every locally Lipschitz curve is differentiable almost everywhere and can be uniquely determined if it satisfies a known ordinary differential equation in canonical form.

\begin{lemma}
\label{lm: boundary is lipschitz}
	Let $q=(r_x,R)\in \partial\bar{\mathcal{R}}$ be a boundary point of the reachable set $\bar{\mathcal{R}}$ that is not a pole, i.e.~$q\ne (0,\pm 1)$. Then there exists a convex neighborhood $F$ of the point $q$ in which one can introduce coordinates such that the boundary $\partial\bar{\mathcal{R}}$ is the graph of a Lipschitz function. Moreover, the boundary $\partial\bar{\mathcal{R}}$ divides $F$ into two connected open sets $F_1$ and $F_2$ (that is, $F=F_1\sqcup F_2\sqcup (\partial\bar{\mathcal{R}}\cap F)$) such that $F_1\subset \Int\bar{\mathcal{R}}$ and $F_2\cap \bar{\mathcal{R}}=\emptyset$.
\end{lemma}

\begin{proof}
	By Lemma~\ref{lemma: boundary is outside the eight}, the boundary $\partial\bar{\mathcal{R}}$ lies outside the closed set bounded by the curve $E(\varphi)$ (with the exception of the two points $(r_x,R)=(0,\pm 1)$). Thus, $q=(r_x,R) \in E_+$, where $E_+=\Int E_+$ is the exterior of the set bounded by the figure-eight curve $E(\varphi)$.     
	
	Consider the set of admissible velocities $U(q)$ at the point $q$ and the cone $C(q)=\R_+U(q)$. Since $q\in E_+$, the cone $C(q)$ cannot be the entire plane, a half-plane, a line, a half-line, or a point. In other words, $C(q)$ is a proper cone with a non-empty interior. Due to the continuous dependence of the set $U(q)$ on the point $q$, there exist a convex neighborhood $F$ of the point $q$ and a proper subcone $H\subset C(q)$ with a non-empty interior, such that for all $q_1 \in F$, it holds that $H \subset C(q_1)$ (see Fig.~\ref{fig: cone H}). For convenience, we will consider the cone $H$ to be open, i.e., $H=\Int H$. 
	
	Further, we will repeatedly use the following simple property of the reachable set $\bar{\mathcal R}$: if $q'\in\bar{\mathcal R}$ and an admissible trajectory of the control system~\eqref{eq: QCS in cylindrical coordinates} leads from $q'$ to some point $q''$, then $q''\in\bar{\mathcal R}$.
	
	\begin{figure}
		\centering
		\includegraphics[width = 6cm]{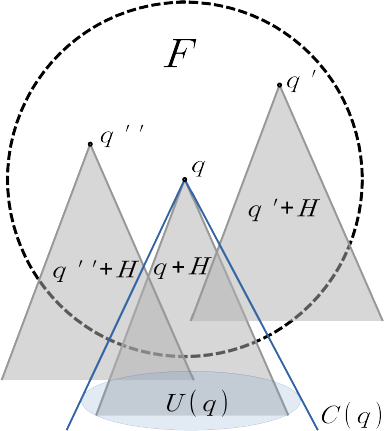}
		\caption{\label{fig: cone H}Cone $H$ of admissible velocities in the neighborhood $F$ of the point $q$.}
	\end{figure}
	
	\textbf{1.} We claim that if a point $q'\in F$ belongs to $\bar{\mathcal{R}}$, then $(q'+H)\cap F \subset\Int\bar{\mathcal{R}}$. Indeed, from the point $q'$, one can reach any point $q''\in(q'+H)\cap F$ by simple motion along a straight line, since $H\subset C(q_1)$ for all $q_1\in F$ and the set $F$ is convex. Thus, $q''\in \bar{\mathcal{R}}$, i.e., $(q'+H)\cap F\subset \bar{\mathcal{R}}$. It remains to note that the set $(q'+H)\cap F$ is open. 
	
	\textbf{2.} We claim that if a point $q'\in F$ does not belong to $\Int\bar{\mathcal{R}}$, then $\big((q'-H)\cap F\big)\cap\bar{\mathcal{R}}=\emptyset$. Indeed, if there existed a point $q''\in \big((q'-H)\cap F\big)\cap\bar{\mathcal{R}}$, then from the two facts $q''\in F\cap\bar{\mathcal{R}}$ and $q'\in q''+H$, we would obtain $q'\in\Int\bar{\mathcal{R}}$ by Step 1, a contradiction.
	
	\textbf{3.} We claim that if a point $q'\in F$ lies on the boundary $\partial\bar{\mathcal{R}}$, then both Steps 1 and 2 hold, and thus neither of the sets $(q'\pm H)\cap F$ contains other boundary points of $\partial\bar{\mathcal{R}}$ (see Fig.~\ref{fig: points alpha beta gamma}). Indeed, by Step 1, the set $(q'+H)\cap F$ is contained in the interior $\Int\bar{\mathcal{R}}$ and thus does not contain boundary points of $\bar{\mathcal{R}}$. By Step 2, the set $(q'-H)\cap F$ does not intersect $\bar{\mathcal{R}}$; since it is open, it therefore cannot contain boundary points of $\bar{\mathcal{R}}$.
	
	\begin{figure}[h]
		\centering
		\includegraphics[height=7cm]{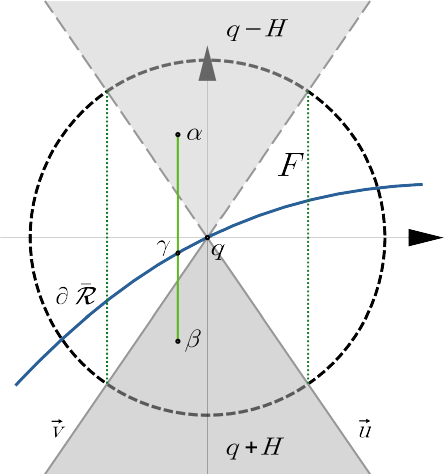} 
		
		\caption{\label{fig: points alpha beta gamma}Positions of points $\alpha$, $\beta$, and $\gamma$ in the neighborhood $F$.}
		
	\end{figure}
	
	Let us now show that $\partial\bar{\mathcal{R}}\cap F$ is the graph of some function in a suitable coordinate system. To do this, we fix the generators $\vec u, \vec v$ of the cone $H$. Then the generators of the cone $-H$ are the vectors $-\vec u, -\vec v$. Choose the following coordinate system: the origin is the point $q$, and the axes are parallel to the vectors $\frac{\vec v - \vec u}{2}$ and $-\frac{\vec u + \vec v}{2}$. The direction $-\frac{\vec u + \vec v}{2}$ will be called vertical.
	
	Shrink $F$ (preserving convexity) if necessary, so that any vertical line either does not intersect $F$, or intersects both $F\cap (q+H)$ and $F\cap(q-H)$ (see Fig.~\ref{fig: points alpha beta gamma}). In this coordinate system, choose any vertical line intersecting $F$. Then on this line, there will be two points $\alpha,\beta\in F$ such that $\alpha\in q-H$ and $\beta\in q+H$. Consequently, by Steps 1 and 2, we have $\alpha\not\in\bar{\mathcal{R}}$ and $\beta\in\Int\bar{\mathcal{R}}$.
	
	Connect $\alpha$ and $\beta$ with a vertical segment. Since both sets $\Int\bar{\mathcal R}$ and $F\setminus\bar{\mathcal{R}}$ are open and do not intersect, by the connectedness of the segment, there must be a point $\gamma$ on it that does not belong to either of these sets and, consequently, lies on the boundary, $\gamma \in \partial \bar{\mathcal{R}}$ (see Fig.~\ref{fig: points alpha beta gamma}). Therefore, on any vertical line intersecting $F$, there is at least one point from $\partial \bar{\mathcal{R}}$. Such a point is unique by Step 3. Thus, in the neighborhood $F$ of $q$, in the specified coordinate system, the boundary $\partial \bar{\mathcal{R}}$ is the graph of some function. By Step 3, this function must be a Lipschitz function, dividing $F$ into two regions $F_1\subset \Int \bar{\mathcal R}$ and $F_2\subset F\setminus \bar{\mathcal{R}}$: the region $F_1$ below its graph, and $F_2$ above.
	
\end{proof}

\begin{corollary}
\label{cor: boundary is locally lipschitz}
	The boundary $\partial\bar{\mathcal R}$ excluding the poles $(r_x,R)=(0,\pm 1)$ is a locally Lipschitz curve.
\end{corollary}
A locally Lipschitz curve is locally absolutely continuous and therefore differentiable almost everywhere. Let us now write down the ordinary differential equation that the curve $\partial\bar{\mathcal{R}}$ satisfies.

\begin{lemma}
\label{lm: boundary obey ODE}
	

    For almost every point $q=(r_x,R)\in\partial\bar{\mathcal{R}}$, the boundary curve $\partial\bar{\mathcal{R}}$ is tangent to one of the generators of the cone spanned by the set of admissible velocities of system~\eqref{eq: QCS in cylindrical coordinates} at $q$. The velocity vector along this generator is parameterized by an angle $\theta$, which must be a solution to equation~\eqref{eq: theta main equation}.
\end{lemma}

\begin{proof}
	
	The set $E_+$ consists of points $q$ for which zero velocity does not belong to the set of admissible velocities $U(q)$.
	
	Let $q \in \partial{\bar{\mathcal{R}}}\cap E_+$ be a point such that the boundary $\partial\bar{\mathcal{R}}$ is differentiable at $q$. Let $\vec v$ be a unit tangent vector to $\partial{\bar{\mathcal{R}}}$ at the point $q$. We prove by contradiction that $\vec v\in\partial C(q)$ or $-\vec v\in\partial C(q)$, where $C(q)=\R_+U(q)$.
	
	\begin{enumerate}
		
		\item Suppose the vector $\vec v$ lies in the interior of $C(q)$. In this case, we immediately arrive at a contradiction, as this implies that points of $\partial{\bar{\mathcal R}}$ close to $q$ would belong to the interior $\Int\bar{\mathcal{R}}$ (analogous to Step 1 in Lemma~\ref{lm: boundary is lipschitz}).
		
		\item Suppose the vector $\vec v$ lies in the interior of $-C(q)$. Then $-\vec v$ is also tangent to $\partial{\bar{\mathcal{R}}}$ and lies in the interior of $C(q)$, which similarly leads to a contradiction.
		
		\item The vector $\vec v$ belongs to the interior of the complement of $C(q)\cup (-C(q))$. For the given point $q$, take the neighborhood $F$ from Lemma~\ref{lm: boundary is lipschitz}, i.e., $F=F_1\sqcup F_2\sqcup (F\cap\partial\bar{\mathcal{R}})$, where $F_1\subset\Int\bar{\mathcal{R}}$ and $F_2\cap \bar{\mathcal{R}}=\emptyset$. The case where the cone points towards $F_2$ obviously leads to a contradiction, since any motion from the point $q$ with a velocity in $C(q)$ must simultaneously lie in $F_2$ and in $\bar{\mathcal{R}}$, which is impossible. Let us consider in more detail the case where $C(q)$ points towards $F_1$. Since the vector $\vec v$ does not belong to $C(q)$, there exists a strictly larger cone $\hat H$ such that $C(q) \subset \Int\hat H$ and $\vec v \not \in \Cl\hat H$. By shrinking the neighborhood $F$ of the point $q$ if necessary, we can ensure that for any point $q_1$ in this neighborhood, $C(q_1)\subset\Int\hat H$. Further, assuming $q$ is not a pole (implied by the context $q \in E_+$ and Lemma \ref{lm: boundary is lipschitz}), there exists a point $q'\in\Cl F_1$ from which a trajectory reaches $q$ without leaving $F$:
		$$
			\exists \ \gamma:[0,\tau]\rightarrow \bar{\mathcal{R}} \cap F: \quad \gamma(0) = q', \ \gamma(\tau) = q, \ \forall t \in [0,\tau]: \ \dot\gamma(t) \in C(\gamma(t)) \subset \Int\hat H
		$$
		Thus:
		$$
			q - q' = \int_{0}^{\tau} \dot{\gamma}(t)\,dt \in \Int\hat H\ \implies\ q' \in q - \Int\hat H 
		$$
		By construction, we have $(q - \Int\hat H) \cap F \subset F_2$, which means $q' \not \in \bar{\mathcal{R}}$, a contradiction.
	\end{enumerate}
	
	\begin{figure}
		\centering
		\includegraphics[width=0.5\linewidth]{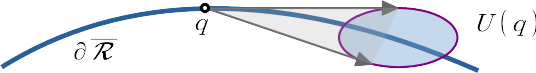}
		\caption{\label{fig: boundary tangent line}Tangent line to the boundary $\partial\bar{\mathcal{R}}$.}
	\end{figure}
	
	Thus, one of the generators of the cone of admissible velocities $C(q)=\R_+ U(q)$ is tangent to $\partial{\bar{\mathcal{R}}}$ at the point $q$ (see Fig.~\ref{fig: boundary tangent line}). The boundary of the set $U(q)$ is given by the right-hand side of~\eqref{eq: QCS in cylindrical coordinates} for $\theta\in[0,2\pi]$. Let us find the generators of the cone. Recall that
	$$
	\partial U(q) = \partial U(r_x,R) = \left\{
	V(r_x,R,\theta) = 
	\begin{pmatrix}
		\dot{r}_x(r_x,R,\theta)\\
		\dot{R}(r_x,R,\theta)
	\end{pmatrix}
	\,\Bigg|\,
	\theta\in[0,2\pi]
	\right\}.
	$$
	Consequently, the point $V(r_x,R,\theta)$ lies on the boundary of the cone $C(q)=\R_+U(q)$ if and only if the tangent to the boundary $\partial U(q)$ at the point $V(r_x,R,\theta)$ passes through the origin, i.e.,
	$$
		V(r_x,R,\theta) \parallel \frac{\partial}{\partial\theta} V(r_x,R,\theta)
		\quad\Longleftrightarrow\quad
		\det\left[
			V(r_x,R,\theta), \frac{\partial}{\partial\theta} V(r_x,R,\theta)
		\right]=0.
	$$
	After direct substitution of the right-hand side of~\eqref{eq: QCS in cylindrical coordinates} into the last equality and subsequent simplification, we obtain~\eqref{eq: theta main equation}, which was to be proved.
\end{proof}

Note that at each point $(r_x,R) \in E_{+}$, equation~\eqref{eq: theta main equation} has exactly two solutions $\theta_1$ and $\theta_2$ (since the cone $C(r_x,R)$ has two generators), and only one of them, when substituted into~\eqref{eq: QCS in cylindrical coordinates}, gives the tangent vector to the boundary $\partial\bar{\mathcal{R}}$. Furthermore, for points $q\in E_+$, the implicit function theorem readily guarantees that both solutions $\theta_1$ and $\theta_2$ are smooth functions of $(r_x,R)$ in a neighborhood of $q$. However, the result of Lemma~\ref{lm: boundary obey ODE} does not yet allow us to uniquely determine the curve $\partial\bar{\mathcal{R}}$, since equation~\eqref{eq: theta main equation} has two solutions.

Let $S(r_x,R,\theta)$ denote the left-hand side of equation~\eqref{eq: theta main equation} multiplied by $R$, i.e.\ equation~\eqref{eq: theta main equation} takes the form $S/R=0$. The axis $R=0$ is of no importance at the moment, so we will work with the equation $S=0$, which is more convenient for subsequent analysis. This equation has two solutions on $E_+$, and the function $\theta(t)$ (as a control) is not necessarily Lipschitz continuous. Thus, the classical Picard theorem on the existence and uniqueness of ODE solutions cannot be applied, since the pair of equations~\eqref{eq: QCS in cylindrical coordinates} and \eqref{eq: theta main equation} defines an ODE in non-canonical form.

Nevertheless, Lemma~\ref{lm: boundary obey ODE} provides the key equation~\eqref{eq: theta main equation} for the angle $\theta$, which will allow us to find the boundary $\partial\bar{\mathcal{R}}$ analytically. To demonstrate the method for finding the boundary, let us assume for simplicity that we have already proved that the boundary $\partial\bar{\mathcal{R}}$ is a smooth curve in a neighborhood of some point $q\in \partial\bar{\mathcal{R}}$. Then one can choose its parameterization such that the functions $r_x(t)$, $R(t)$, $\theta(t)$ will be smooth solutions of system~\eqref{eq: QCS in cylindrical coordinates}. With a suitable choice of control $u$, equation~\eqref{eq: QCS in cylindrical coordinates theta} can also be satisfied. Further, equation~\eqref{eq: theta main equation} defines a two-dimensional surface $\{S=0\}$ in the three-dimensional space $(r_x,R,\theta)$, on which the curve $(r_x(t), R(t), \theta(t))$ must lie. This fact uniquely determines the control $u_b$, i.e., upon substituting $u=u_b$ into \eqref{eq: QCS in cylindrical coordinates}, \eqref{eq: QCS in cylindrical coordinates theta}, we obtain a three-dimensional system of ODEs, for which the two-dimensional surface $\{S=0\}$ is invariant. We will show that the point $N$ is an equilibrium point of this two-dimensional ODE system on $\{S=0\}$ and is of hyperbolic type. Next, we will show that its unstable separatrices are precisely the boundary $\partial\bar{\mathcal R}$.

\begin{lemma}
	\label{lm: ub}
	Let $r_x(t)$, $R(t)$, and $\theta(t)$ be Lipschitz functions and $u(t)\in L_\infty$. If equations~\eqref{eq: QCS in cylindrical coordinates}, \eqref{eq: QCS in cylindrical coordinates theta}, \eqref{eq: theta main equation} hold for these functions, then
	$$
		\frac\kappa\omega u = \frac\kappa\omega u_b 
		\stackrel{def}{=}
		\frac{(1+R \sin\theta ) (2r_x \sin \theta + 2\alpha \cos\theta+\alpha R\sin\theta\cos\theta)}{R(1+R \sin ^3\theta)}
		=
		\frac{(1+r_z)(2 r_x r_z+2 \alpha r_y+\alpha r_y r_z)}{r_y^2+r_z^2+r_z^3}.
	$$
\end{lemma}

\begin{proof}
	
	Since the Lipschitz functions $r_x(t),R(t),\theta(t)$ satisfy the equation $S=0$ for all $t$, the derivative $\dot S$ along the vector field of system\eqref{eq: QCS in cylindrical coordinates}, \eqref{eq: QCS in cylindrical coordinates theta} must also be equal to $0$ for almost all $t$. Let us compute it in the original coordinates $r_x,r_y,r_z$. In these coordinates, equation \eqref{eq: theta main equation} takes the form
	$$
		S=r_z(r_x^2+r_y^2)+(1+r_z)(2r_y^2+2r_z^2+\alpha r_x r_y)=0.
	$$
	Note that in the coordinates $(r_x,r_y,r_z)$, the surface $\{S=0\}$ has only one singular point $r_x=r_y=r_z=0$ and is smooth outside it.
	
	Differentiating the obtained formula for $S$ along the vector field of system~\eqref{eq: QCS in Bloch ball appendix} (which is equivalent to \eqref{eq: QCS in cylindrical coordinates}, \eqref{eq: QCS in cylindrical coordinates theta}), we find the control
	$$
		\frac\kappa\omega u=\frac{\alpha  r_x^2r_z+2 (\alpha ^2-2) r_x r_y (1+r_z)+\alpha r_y^2 (7 r_z+6)+2 \alpha r_z (3r_z^2+5 r_z+2)}{r_x^2 r_y+\alpha r_x (r_y^2-r_z^2-r_z)+3 r_y^3}\stackrel{def}{=}\frac{P}{Q}.
	$$
	This formula for the control is rather cumbersome, but it can be greatly simplified by using the fact that $S=0$, namely, on the surface $S=0$ we have
	$$
		\frac\kappa\omega u=\frac{P}{Q}=\frac{r_z P}{r_z Q} = 
		\frac{r_z P - \alpha(3r_z+2)S}{r_z Q - r_y S} = \frac{(\alpha r_x+2 r_y)(1+r_z) (2 r_x r_z+2 \alpha r_y+\alpha r_y r_z)}{(\alpha r_x+2 r_y)(r_y^2+r_z^3+r_z^2)}
	$$
	and we obtain the formula stated in the lemma.
\end{proof}

Note that the obtained formula for the control $u_b$ has a singularity if the denominator vanishes. This fact introduces some difficulties in investigating the behavior of system~\eqref{eq: QCS in cylindrical coordinates}, \eqref{eq: QCS in cylindrical coordinates theta} with $u=u_b$ both on the surface $S=0$ and inside the Bloch ball $\BB$. Therefore, we will begin this investigation by isolating the zeros of the function $r_y^2+r_z^2+r_z^3$.

\begin{lemma}
	\label{lm: ub denominator zeros}
	Assume $0<\alpha<1$. Suppose that $r_y^2+r_z^2+r_z^3=0$ and either $S=0$ or $r_x^2+r_y^2+r_z^2\le 1$. Then either $(r_x,r_y,r_z)=N=(0,0,-1)$ or $r_y=r_z=0$.
\end{lemma}

\begin{proof}
	First, let us study the structure of the set of solutions $r_y^2 + r_z^2 + r_z^3=0$ in the Bloch ball $\BB$. If $r_x^2+r_y^2+r_z^2 \le 1$, then $|r_z|\le 1$, and, consequently, $r_y^2 = -r_z^2-r_z^3 \le 0$. Therefore, if $r_y^2 + r_z^2 + r_z^3=0$ holds for a point in the Bloch ball $\BB$, then $r_y=0$ and either $r_z=0$ or $r_z=-1$. Thus, in the Bloch ball $\BB$, the denominator of $u_b$ becomes zero at the equilibrium point $N$ and on the axis $r_y=r_z=0$.
	
	Now let us turn to the surface $\{S=0\}$. In general, this surface does not lie entirely within the Bloch ball $\BB$. Therefore, the zeros of the denominator of $u_b$ on the surface $\{S=0\}$ require separate investigation. We will show that the denominator of $u_b$ becomes zero on the surface $S=0$ only at the point $N=(0,0,-1)$ and on the axis $\{r_y=r_z=0\}=\{R=0\}$.
	
	First, we show that there is only one point on the surface $\{S=0\}$ with $r_z=-1$. Indeed, if $r_z=-1$, then from $S=0$ we find $r_x^2+r_y^2=0$. Suppose $r_z\ne -1$. If the denominator of $u_b$ vanishes, then $r_y^2=-r_z^2-r_z^3$, and, substituting $r_y^2$ into $S=0$, we obtain
	$$
	r_x^2r_z+\alpha r_x r_y (1+r_z)-3 (1+r_z) r_z^3=0.
	$$
	We will show that this equation implies $r_z=0$. To do this, assume the contrary, i.e., $r_z\ne 0$. Then, since $r_z\ne 0,-1$, we get $r_x\ne 0$. Therefore, from the equality above, one can express $r_y$ in terms of $r_x$ and $r_z$ and substitute it into the equality $r_y^2+r_z^2+r_z^3=0$. As a result, we obtain the following equation:
	$$
	\alpha^2(1+r_z)^3r_x^2 + (r_x^2-3r_z^2(1+r_z))^2=0.
	$$
	The obtained equation can be viewed as a biquadratic equation in $r_x$. The discriminant of this equation has the form:
	$$
	D=\alpha^2(1+r_z)^4(\alpha^2(1+r_z)^2-12r_z^2).
	$$
	This discriminant is negative if $\alpha|1+r_z| < 2\sqrt3|r_z|$, which is certainly true if $|r_z|>\frac12\alpha$ (since $\alpha<1$). 
	
	Thus, if $r_z\ne 0,-1$, and $r_y^2+r_z^2+r_z^3=0$ at some point on the surface $S=0$, then $|r_z|\le \frac12\alpha$. But in this case $|r_z|^3 < r_z^2$ (since $\alpha<1$ and $r_z\ne 0$), and, thus, $r_y^2+r_z^2+r_z^3 >0$. We arrive at a contradiction.
	
	Therefore, if $r_z\ne -1$, $r_y^2+r_z^2+r_z^3=0$ and $S=0$, then $r_z=0$. Consequently, $r_y=0$, which was to be proved.
	
\end{proof}

\begin{lemma}
	\label{lm: separatrix}
	For $u=u_b$, the surface $\{S=0\}$ is invariant with respect to the ODE system\footnote{Recall that system~\eqref{eq: QCS in Bloch ball appendix} is equivalent to~\eqref{eq: QCS in cylindrical coordinates}, \eqref{eq: QCS in cylindrical coordinates theta}.}~\eqref{eq: QCS in Bloch ball appendix}. The point $N$ is a removable singularity. The point $N\in\{S=0\}$ itself is an equilibrium point, and system~\eqref{eq: QCS in Bloch ball appendix} with $u=u_b$ on $\{S=0\}$ in the neighborhood of $N$ exhibits the dynamics of a hyperbolic equilibrium point: there are two stable separatrices and two unstable ones. On the unstable separatrices in the neighborhood of $N$, the asymptotics~\eqref{eq: theta Taylor appendix} and~\eqref{eq: R Taylor appendix} hold. The stable separatrices do not lie in the Bloch ball $\BB$.
\end{lemma}

\begin{proof}
	
	Upon substituting $u_b$ into the control system~\eqref{eq: main QCS appendix}, we obtain an ODE system with respect to which the two-dimensional surface $S=0$ is invariant, since $\dot S=0$ by the construction of $u_b$ (see Lemma~\ref{lm: ub}). Note that this ODE system is not defined at $N$, but is defined in a punctured neighborhood of $N$ on the surface $S=0$ by Lemma~\ref{lm: ub denominator zeros}. Let us show that the following limit exists and equals zero:
	$$
		\lim_{\substack{(r_x,r_y,r_z)\to(0,0,-1)\\S(r_x,r_y,r_z)=0}} u_b=0
	$$
	To do this, let us introduce local coordinates on $S$ in the neighborhood of $N$. Since $\frac{\partial}{\partial R}S(r_x=0,R=1,\theta=-\frac\pi2)=1$, in the neighborhood of $N$, the equation $S=0$ can be smoothly solved for $R$:\footnote{An explicit formula for $R(r_x,\theta)$ is also easy to write: the equation $S=0$ is quadratic in $R$, but it is not important for what follows.}
	\begin{equation}
		\label{eq:R taylor long}
		\begin{aligned}
			R(r_x,\theta) =& 
			1-\frac12r_x^2+
			\frac38\left(\theta+\frac\pi2\right)^4+
			\frac14 \alpha r_x\left(\theta+\frac\pi2\right)^3+\\
			&+\frac{1}{4}r_x^2\left(\theta+\frac{\pi }{2}\right)^2+
			\frac14\alpha
			r_x^3\left(\theta +\frac\pi2\right)-\frac14r_x^4 + O\left(r_x^6+\left(\theta +\frac\pi2\right)^6\right).
		\end{aligned}
	\end{equation}
	Therefore, $(r_x,\theta)$ are local coordinates on $S$ in the neighborhood of the point $N$. In these local coordinates, the control $u=u_b$ takes the form
	$$
		\frac\kappa\omega u_b = -\frac{
			\left(2r_x-\alpha\left(\theta+\frac\pi2\right)\right)
			\left(r_x^2 + \left(\theta+\frac\pi2\right)^2
			\right)+
			O\left(r_x^4+\left(\theta+\frac\pi2\right)^4\right)
		}{
			r_x^2 +
			3 \left(\theta+\frac\pi2\right)^2 +
			O\left(|r_x|^3+\left|\theta+\frac\pi2\right|^3\right)
		}.
	$$
	Consequently, on $S=0$, the control $u_b$ has a limit at $N$ equal to zero. However, if we define $u_b(N)=0$, the resulting function will not be smooth at the point $N$, since the ratio $(r_x^2 + (\theta+\pi/2)^2)/(r_x^2 + 3(\theta+\pi/2)^2)$ is bounded in the neighborhood of the point $(r_x=0,\theta=-\pi/2)$, but does not have a limit there. Nevertheless, the function $u_b|_{S=0}$ is Lipschitz in the neighborhood of the point $N$. Therefore, upon substituting $u=u_b$ into system~\eqref{eq: main QCS appendix} on the surface $S=0$, we obtain an ODE system with a Lipschitz right-hand side. Thus, the conditions of the Picard existence and uniqueness theorem hold for it.
	
	On the surface $S=0$, the point $N$ is an equilibrium point of the ODE system~\eqref{eq: main QCS appendix} with $u=u_b$. However, this system is not smooth at $N$. Nevertheless, below we will show that the behavior of solutions of~\eqref{eq: main QCS appendix} with $u=u_b$ on $S=0$ in the neighborhood of the fixed point $N$ is completely analogous to that of a hyperbolic equilibrium point on the plane.
	
	On the surface $S=0$, system~\eqref{eq: main QCS appendix} with $u=u_b$ takes the form
	$$
	\begin{cases}
		\frac1\omega\dot r_x = -R \cos \theta-\alpha r_x/2\\
		\frac1\omega\dot\theta = \frac14(1+R \sin \theta (1+\cos ^2\theta)) (4 \alpha  \cos\theta+\alpha R \sin 2\theta +4 r_x\sin\theta )/[R (1+R \sin ^3\theta)].
	\end{cases}
	$$
	Substituting the Taylor expansion for $R$ on $S=0$, we obtain
	$$
	\begin{cases}
		\frac1\omega\dot r_x=-\frac12\alpha r_x - \left(\theta+\frac\pi2\right)+O\left(|r_x|^3+\left|\theta+\frac\pi2\right|^3\right)\\
		\frac1\omega\dot\theta=-\dfrac{
			\left(r_x-\frac12\alpha\left(\theta+\frac\pi2\right)\right)\left(r_x^2 - \left(\theta+\frac\pi2\right)^2\right)
		}{
			r_x^2+3\left(\theta+\frac\pi2\right)^2
		}
		+ O(|r_x|^3+\left|\theta+\frac\pi2\right|^3).
	\end{cases}
	$$
	
	To resolve the singularity at the point $(r_x=0,\theta=-\frac\pi2)$, we perform a standard blow-up procedure and glue in a circle instead of the point $N$. To do this, we set $r_x=\eta\cos\varphi$, $\theta+\frac\pi2=\eta\sin\varphi$. As a result, we obtain the system
	\begin{equation}
		\label{eq: blow up N}
		\begin{cases}
			\frac1\omega\dot\eta = -\dfrac{\frac\alpha2+\sin2\varphi}{2-\cos2\varphi}\eta + O(\eta^2)\\[0.5cm]
			\frac1\omega\dot\varphi = \dfrac{1-2 \cos2\varphi+\frac\alpha2\sin2\varphi}{2-\cos2\varphi} + O(\eta^2).
		\end{cases}
	\end{equation}
	Thus, the circle $\eta=0$ is an invariant manifold of this system. In the punctured neighborhood of this circle, there are no fixed points, since the point $N$ was an isolated singular point. To find the fixed points on the circle $\eta=0$, it is sufficient to solve the equation $1-2 \cos2\varphi+\frac\alpha2\sin2\varphi=0$ or equivalently $3\tan^2\varphi + \alpha\tan\varphi -1=0$. From this, we find four fixed points:
	\begin{equation}
		\label{eq: varphi on separatrix}
		\varphi = -\arctan\frac{\alpha\pm\sqrt{12+\alpha^2}}{6}+\pi k = \pm\frac\pi6 - \frac18\alpha + O(\alpha^2)+\pi k,\quad k=0,1.
	\end{equation}
	
	Let us now study the type of these singular points:
	$$
		\frac1\omega
		\begin{pmatrix}
			\partial\dot\eta/\partial\eta & \partial\dot\eta/\partial\varphi\\
			\partial\dot\varphi/\partial\eta & \partial\dot\varphi/\partial\varphi\\
		\end{pmatrix}\Big|_{\eta=0} = 
		\begin{pmatrix}
			-\dfrac{\sin2\varphi+\frac\alpha2}{2-\cos2\varphi} & 0 \\
			0 & \dfrac{6 \sin2\varphi + \alpha (2 \cos2\varphi-1)}{(2-\cos2 \varphi)^2}
		\end{pmatrix}.
	$$
	Therefore, for $\varphi$ from~\eqref{eq: varphi on separatrix}, we obtain
	$$
		\frac1\omega
		\begin{pmatrix}
			\partial\dot\eta/\partial\eta & \partial\dot\eta/\partial\varphi\\
			\partial\dot\varphi/\partial\eta & \partial\dot\varphi/\partial\varphi\\
		\end{pmatrix}\Big|_{\eta=0} = 
		\dfrac13\begin{pmatrix}
			-\alpha\pm\frac12\sqrt{12+\alpha ^2}&0\\
			0 &\alpha\mp 2\sqrt{12+\alpha ^2}
		\end{pmatrix}.
	$$
	
	Consequently, for $0<\alpha<2$, all four points turn out to be hyperbolic. Each of them has a stable and an unstable manifold. We are interested in the points whose unstable manifolds are transversal to the circle $\eta=0$ --- these are two points with $\varphi_{0,1}=-\arctan\frac{\alpha+\sqrt{12+\alpha^2}}{6}+\pi k$, $k=0,1$, near the points $\varphi=-\frac\pi6$ and $\varphi=\pi-\frac\pi6$. Note that the system is analytic and does not change under the reflection $(\eta,\varphi)\mapsto(-\eta,\varphi+\pi)$. Therefore, these two unstable manifolds transform into each other under such a reflection. Since on each of them it holds that $\varphi=\varphi_{0,1}+O(\eta^2)$, we obtain 
	$$
		r_x = \eta\cos\varphi_{0,1}+O(\eta^3);
		\quad
		\theta+\frac\pi2 = \eta\sin\varphi_{0,1}+O(\eta^3)
	$$
	Therefore,
	$$
		\eta = \frac1{\cos\varphi_{0,1}}r_x + O(r_x^3)
		\quad\text{and}\quad
		\theta+\frac\pi2=r_x\tan\varphi_{0,1}+O(r_x^3).
	$$
	This is precisely the Taylor expansion~\eqref{eq: theta Taylor appendix}. Formula~\eqref{eq: R Taylor appendix} is obtained by substituting this expansion for $\theta$ into~\eqref{eq:R taylor long}.
	
	Similar calculations can be done for the stable separatrices of $N$, but the expansion for $R$ will have the form:
	$$
		R=1-\frac12r_x^2 - \frac{1}{432} \left(54+18 \alpha ^2+\alpha ^4-12  \alpha\sqrt{\alpha ^2+12} - \alpha ^3\sqrt{\alpha ^2+12}\right) r_x^4 + O(r_x^6).
	$$
	Since the coefficient of $r_x^4$ here is greater than $-\frac18$, on the stable separatrices in the neighborhood of $N$, it holds that $r_x^2+R^2=r_x^2+r_y^2+r_z^2>1$, and they lie outside the Bloch ball $\BB$.
\end{proof}

\begin{figure} 
	\centering
	\includegraphics[width=0.5\linewidth]{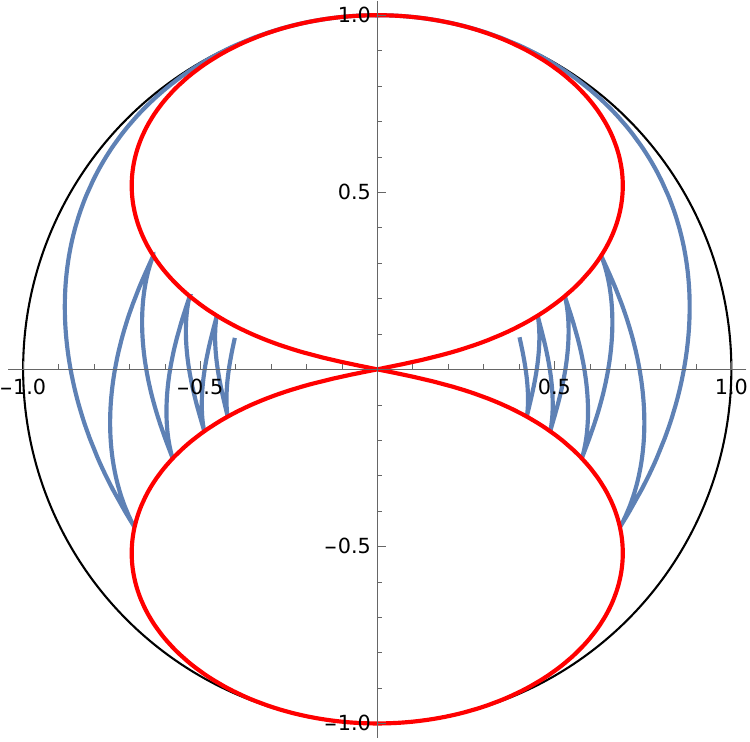}
	\caption{\label{fig: full separatrix}Projection of the separatrices on the $(r_x,R)$ plane.}
\end{figure}

The separatrices found in Lemma~\ref{lm: separatrix} lie on the surface $\{S=0\}$. When projected onto the $(r_x,R)$ plane, the image of this surface never enters the interior of the figure-eight curve $E(\varphi)$, and the projection of the separatrix has the form shown in Fig.~\ref{fig: full separatrix}.

In order to show that the unstable separatrices obtained above emanating from the point $N$ on the surface $\{S=0\}$ indeed define the boundary of the set $\bar{\mathcal{R}}$, it will be convenient for us to prove that both of these separatrices intersect the surface $\{R=0\}$ (see Fig.~\ref{fig: 2d separatrix}). To do this, let us track the sign change of $\dot R$.

\begin{lemma}
	\label{lm: sign dot R}
	If $u=u_b$, then on the surface $\{S=0\}$, the equality $\dot R=0$ for $|R|\le 1$ holds only at the points of the figure-eight curve $E(\varphi)$ from Theorem~\ref{thm: local controllability} when $\theta=\varphi$.
\end{lemma}

\begin{proof}
	
	The equations $\dot R=0$ and $S=0$ have the form
	$$
	\begin{cases}
		r_x\cos\theta-\frac12 \alpha R (1+\sin^2\theta)-\alpha\sin\theta=0,\\
		(1+R\sin\theta)(2R+\alpha r_x \cos\theta) + (R^2 \cos^2\theta+r_x^2)\sin\theta = 0.\\
	\end{cases}
	$$
	Let us first consider the case when $\cos\theta=0$. Then $\theta=\pm\frac\pi2$. In this case, from the first equation, we find $R=\mp 1$. Thus, the second equation gives $r_x=0$. Therefore, these points correspond to the north and south poles and lie on the figure-eight curve.
	
	Consider the case $\cos\theta\ne 0$. The first equation is linear in $r_x$. Expressing $r_x$ from the first equation and substituting it into the second, we get
	$$
		\frac{1}{\cos^2\theta}  ((9\sin\theta+ \sin3\theta)R+8)(2\alpha^2 \sin\theta+(2\alpha^2+(4-\alpha^2)\cos^2\theta)R) = 0.
	$$
	Since $|9\sin\theta+ \sin3\theta|\le 8$, the first bracket can vanish only if $\cos\theta=0$, which contradicts our assumption. Hence, the second bracket must be zero, and we obtain
	$$
		R = \frac{-2\alpha^2\sin\theta}{2\alpha^2+(4-\alpha^2) \cos^2\theta}.
	$$
	Substituting this expression for $R$ into the equation $\dot R=0$, we get
	$$
		r_x=\frac{4\alpha \sin\theta\cos\theta}{2 \alpha^2+(4-\alpha ^2) \cos^2\theta}.
	$$
	The obtained formulas coincide with the definition of the figure-eight curve $E(\varphi)$ from Theorem~\ref{thm: local controllability} when $\theta=\varphi$.
\end{proof}

Now we are ready to prove Theorem~\ref{thm: reachable set boundary}.

\begin{proof}[Proof of Theorem~\ref{thm: reachable set boundary}]
	
	Denote the two unstable separatrices obtained in Lemma~\ref{lm: separatrix} by $s_l$ and $s_r$ with $s_{l,r}\subset \{S=0\}$. For definiteness, let $r_x>0$ on $s_r$ and $r_x<0$ on $s_l$ in the neighborhood of $N$. Obviously, $s_l$ and $s_r$ transform into each other under the reflection $(r_x,\theta)\mapsto(-r_x,\pi-\theta)$.
	
	\begin{figure}
		\centering
		\includegraphics[width=0.5\linewidth]{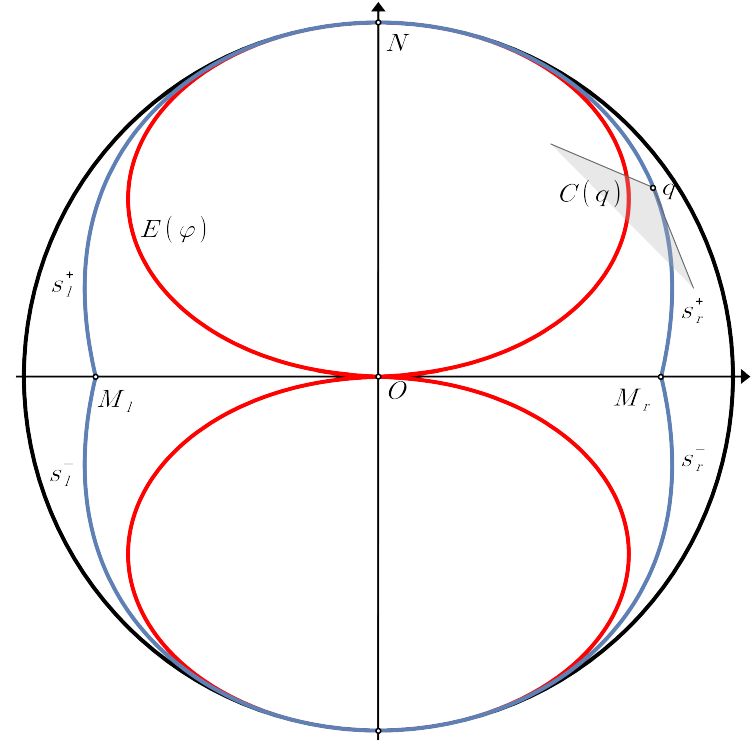}
		\caption{	\label{fig: final structure of boundary}Structure of $\partial\bar{\mathcal{R}}$}
	\end{figure}
	
	First, let us show that the projections of $s_l$ and $s_r$ onto the $(r_x,R)$ plane (obtained by dropping the third coordinate $\theta$) lie in the reachable set $\bar{\mathcal{R}}$. In a sufficiently small neighborhood $V$ of the point $N$ on $\{S=0\}$, the dynamics of the system~\eqref{eq: QCS in cylindrical coordinates} and \eqref{eq: QCS in cylindrical coordinates theta} with $u=u_b$ is analogous to linear hyperbolic dynamics: if a point from $V$ does not lie on the stable separatrix, then as long as it does not leave $V$, as time increases, it approaches one of the two unstable separatrices $s_l$ or $s_r$ (this follows from the blow-up resolution~\eqref{eq: blow up N}). By Lemma~\ref{lemma: boundary is outside the eight}, in an arbitrarily small neighborhood of $N$, there exists an open set of points in\footnote{Recall that $E_+$ is the exterior of the figure-eight curve $E$.} $E_+$ that lie in the reachable set $\bar{\mathcal{R}}$. Therefore, on the surface $\{S=0\}$ in any neighborhood of $N$, there exists an open set of points whose projections lie in $\bar{\mathcal{R}}$. Note that if the projection of a point from $\{S=0\}$ lies in $\bar{\mathcal{R}}$, then its future trajectory under the flow of~\eqref{eq: QCS in cylindrical coordinates}, \eqref{eq: QCS in cylindrical coordinates theta} with $u=u_b$ also projects into $\bar{\mathcal R}$ by the definition of the reachable set $\bar{\mathcal R}$. Thus, under the action of the flow~\eqref{eq: QCS in cylindrical coordinates}, \eqref{eq: QCS in cylindrical coordinates theta} with $u=u_b$, the projections of the images of these points always lie in $\bar{\mathcal R}$, and at the moment they leave $V$, they are arbitrarily close to the points $s_l\cap \partial V$ and $s_r\cap\partial V$. Thus, the projections of the points $s_l\cap \partial V$ and $s_r\cap\partial V$ lie in $\bar{\mathcal{R}}$. Moreover, the same reasoning can be applied with any smaller neighborhood $V_1$ of the point $N$ instead of $V$. That is, the projections of the points $s_l\cap \partial V_1$ and $s_r\cap\partial V_1$ also lie in $\bar{\mathcal R}$. Since the neighborhood $V_1$ can be chosen arbitrarily small, we find that on each of the separatrices $s_l$ and $s_r$, there are points arbitrarily close to $N$ whose projections lie in $\bar{\mathcal R}$. We then again use the fact that the flow preserves this property, meaning the unstable separatrices $s_l$ and $s_r$ project entirely into $\bar{\mathcal{R}}$.
	
	\begin{figure}
		\centering
		\includegraphics[width=0.4\linewidth]{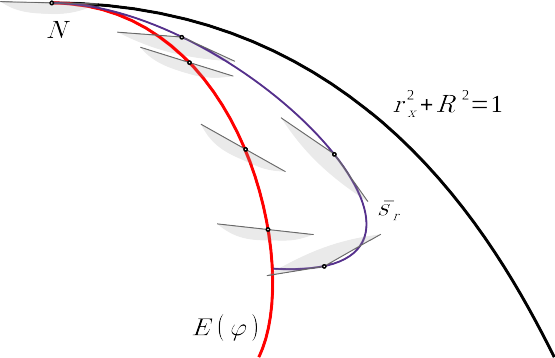}
		\caption{\label{fig: intersection E sr}Structure of cones $C(q)$ along $E(\varphi)$ and $s^+_r$ in the case of their impossible intersection.}
	\end{figure}
	
	Let us now investigate the behavior of these separatrices in the region $\{R\ge0\}$. Denote by $\bar s_r$ the connected component of the separatrix $s_r$ that lies in the region $\{R\ge 0\}$ (and tends to $N$ as $t\to-\infty$). Let $\epsilon$ denote the part of the figure-eight curve $E(\varphi)$ that lies in the same quadrant $\{r_x\ge0,R\ge0\}$. We claim that the curve $\epsilon$ has a neighborhood which the projection of the curve $\bar s_r$ never enters. First let us show that $\bar s_r$ does not contain the origin. Indeed, by the construction of the surface $\{S=0\}$, its projection onto the $(r_x, R)$-plane lies outside the interior of the set bounded by the figure-eight curve $E(\varphi)$. Therefore, if the projection $\bar s_r$ reaches the origin, its velocity at this point cannot be vertical by the structure of the curve $E(\varphi)$ found in Theorem~\ref{thm: local controllability} (see Fig.~\ref{fig: 8 curve}). This is impossible, since all admissible velocities in system~\eqref{eq: QCS in cylindrical coordinates} are vertical at the origin. Now let us show that the projection of $\bar s_r$ cannot intersect the figure-eight curve points of $\epsilon\setminus\{N\}$ for topological reasons. Indeed, at the points of the curve $E(\varphi)$ (except for the origin), the cone $C(q)$ is a half-plane, and the line defining this half-plane is tangent to the curve $E(\varphi)$ itself only at the points $(0,\pm 1)$. Moreover, at the origin, there is a singularity: the cone $C(q)$ degenerates into a line. At the points of $\bar s_r$, one of the generators of $C(q)$ is tangent to $\bar s_r$ (and points in the direction of the forward time flow). Therefore, if an intersection of $E(\varphi)$ and $\bar s_r$ existed in $\{R>0\}$, the generators of the cone would have to swap places (see Fig.~\ref{fig: intersection E sr}). Note that the projection of the entire separatrix $s_r$ might intersect the figure-eight curve $E(\varphi)$, but the first intersection must be in the region $\{R<0\}$ (see Fig.~\ref{fig: full separatrix}). This is because at the point $(0,0)=E(0)=E(\pi)$, the cone $C(q)$ degenerates into a line.
	
	Let us now show that both separatrices intersect $\{R=0\}$. The separatrices $s_l$ and $s_r$ are symmetric, so it is sufficient to give the argument for only one of them. Indeed, by Lemma~\ref{lm: sign dot R}, the sign of $\dot R$ is constant on $\bar s_r$. The case $\dot R>0$ is impossible, because in this case, starting from the point $N$, the separatrix $s_r$ would leave the Bloch ball $\BB$, contradicting the fact that it lies in the reachable set. Let $V$ be a small neighborhood of $N$. Then, due to the continuity of $\dot R$, we obtain that in the compact region $\{|r_x|\le 1, 0\le R\le 1, \theta\in[0;2\pi]\}\setminus V$ outside an open neighborhood of $\epsilon$, the derivative $\dot R$ on $s_r$ is bounded away from zero, and thus, after leaving $V$, the separatrix $s_r$ will intersect $\{R=0\}$ in finite time.
	
	Let $M_l$ and $M_r$ denote the projections of the first intersection points of $s_l$ and $s_r$ with $\{R=0\}$, respectively, where $M_l\in\{R=0,-1<r_x<0\}$, $M_r\in\{R=0,0<r_x<1\}$ (see Fig.~\ref{fig: final structure of boundary}). Further, let $s^+_l$ and $s^+_r$ be the projections onto the $(r_x,R)$ plane of the (open) connected segments of the separatrices $s_l$ and $s_r$ between $N$ and their respective first intersection points with $\{R=0\}$, so that $s^+_{l,r}\subset \{r_x^2+R^2\le 1, R>0\}$. Finally, denote by $s_l^-$ and $s_r^-$ the reflections of $s^+_l$ and $s^+_r$ with respect to the $r_x$-axis, i.e., $s^-_{l,r}\subset \{r_x^2+R^2\le 1, R<0\}$.
	
	We claim that (see Fig.~\ref{fig: final structure of boundary})
	\begin{equation}
	\label{eq: exact formula for boundary}
		\partial \bar{\mathcal R} = s_l^+\cup s_l^-\cup s_r^+ \cup s_r^- \cup \{N,-N\}\cup \{M_l,M_r\}.
	\end{equation}
	To prove this equality, we denote the closed curve on the right-hand side of~\eqref{eq: exact formula for boundary} by $\gamma$ and will show that $\partial\bar{\mathcal{R}}=\gamma$. We have shown that $\gamma\subset\bar{\mathcal R}$. Let us now show that (i) all points inside $\gamma$ lie in $\bar{\mathcal R}$ and (ii) no admissible curve starting from the point $N$ can exit $\gamma$.
	
	To prove (i), observe that if we substitute $\theta=0$ into~\eqref{eq: QCS in cylindrical coordinates}, we get an ODE system of the following form:
	$$
	\dfrac{1}{\omega}
	\begin{pmatrix}
		\dot r_x\\
		\dot R 
	\end{pmatrix}
	=
	\begin{pmatrix}
		-\frac12\alpha & - 1\\
		1  & -\frac12\alpha\\
	\end{pmatrix}
	\begin{pmatrix}
		r_x\\ 
		R
	\end{pmatrix}.
	$$
	The solutions of this system are logarithmic spirals, exponentially tending towards the origin as $t\to+\infty$ and escaping to infinity as $t\to-\infty$. For an arbitrary point $q=(r_x,R)\ne(0,0)$, there exists a unique spiral passing through this point. If the point $q$ lies inside $\gamma$, then in backward time, this spiral will intersect $\gamma$ at some point $q'$, since $\gamma$ is a closed bounded curve, and the spiral tends to infinity in backward time. Consequently, $q'\in\gamma\subset\bar{\mathcal R}$. The point $q$ is reachable from $q'$ using the control $\theta=0$, and, therefore, $q\in\bar{\mathcal R}$. The origin lies in $\bar{\mathcal{R}}$ because a punctured neighborhood of the origin lies in $\bar{\mathcal R}$, and the set $\bar{\mathcal R}$ is closed.
	
	To prove (ii), we show that no admissible curve can leave the region bounded by $\gamma$. Let us check how the cone of admissible velocities $C(q)=\R_+U(q)$ is positioned at each point (say, at a point $q\in s_r^+$, see Fig.~\ref{fig: final structure of boundary}). By the construction of the control $u_b$, one of the generators of the cone $C(q)$ is tangent to $s_r^+$ at the point $q$. Therefore, the cone $C(q)$ either points entirely  into the interior of the curve $\gamma$ or entirely outward, because $s_r^+$ lies in $\Int E_+$, and in $\Int E_+$, the generators of the cone $C(q)$ are never parallel (for $q\in\Int E_+$, the cone $C(q)$ is always strictly smaller than a half-space and has a non-empty interior). The cone $C(q)$ cannot point outward, because for $\theta=0$ or $\theta=\pi$, the solutions of~\eqref{eq: QCS in cylindrical coordinates} are logarithmic spirals tending towards the origin. Thus, at any point $q\in s_{l,r}^\pm$, one generator of the cone of admissible velocities $C(q)$ is tangent to $s_{l,r}^\pm$, and the other is directed inside the curve $\gamma$. By continuity, at the points $M_{l,r}$, both generators of $C(M_{l,r})$ are directed inward. Therefore, no admissible curve from a point on $\gamma$ or inside it can cross $\gamma$ outward. Thus, no point outside the region bounded by $\gamma$ belongs to $\bar{\mathcal R}$, and, consequently, equality~\eqref{eq: exact formula for boundary} holds.
\end{proof}

\begin{proof}[Proof of Theorem~\ref{thm: main}]
	
	Theorem~\ref{thm: main} can be easily derived from Theorem~\ref{thm: reachable set boundary} and the lemmata used in its proof. Indeed, the formula for the control in Theorem~\ref{thm: main} exactly coincides with the formula in Lemma~\ref{lm: ub}. The curve $\Gamma$ in the statement of Theorem~\ref{thm: main} is the union of the separatrices $s_l$ and $s_r$ and the equilibrium $N$, but written in the original coordinates $(r_x, r_y, r_z)$ in the Bloch ball. Its connected component in the half-space $r_z \le 0$ is a union of the initial connected parts of separatrices $s_{l,r}$ in the domain $R \ge 0$ and the point $N$. Hence, the surface of revolution obtained by rotating $\Gamma$ about the $r_x$-axis coincides with the one obtained by rotations of $s_l^+ \cup s_r^+ \cup\{N, M_l, M_r\}$ about the $r_x$-axis, if we consider the $r_xR$-plane as the $r_xr_z$-plane via $(r_x, r_y, r_z) = (r_x, 0, R)$.
\end{proof}

\putbib[apssamp,main_prlNotes]
\end{bibunit}

\providecommand{\noopsort}[1]{}\providecommand{\singleletter}[1]{#1}%
\begin{thebibliography}{59}%
\makeatletter
\providecommand \@ifxundefined [1]{%
 \@ifx{#1\undefined}
}%
\providecommand \@ifnum [1]{%
 \ifnum #1\expandafter \@firstoftwo
 \else \expandafter \@secondoftwo
 \fi
}%
\providecommand \@ifx [1]{%
 \ifx #1\expandafter \@firstoftwo
 \else \expandafter \@secondoftwo
 \fi
}%
\providecommand \natexlab [1]{#1}%
\providecommand \enquote  [1]{``#1''}%
\providecommand \bibnamefont  [1]{#1}%
\providecommand \bibfnamefont [1]{#1}%
\providecommand \citenamefont [1]{#1}%
\providecommand \href@noop [0]{\@secondoftwo}%
\providecommand \href [0]{\begingroup \@sanitize@url \@href}%
\providecommand \@href[1]{\@@startlink{#1}\@@href}%
\providecommand \@@href[1]{\endgroup#1\@@endlink}%
\providecommand \@sanitize@url [0]{\catcode `\\12\catcode `\$12\catcode
  `\&12\catcode `\#12\catcode `\^12\catcode `\_12\catcode `\%12\relax}%
\providecommand \@@startlink[1]{}%
\providecommand \@@endlink[0]{}%
\providecommand \url  [0]{\begingroup\@sanitize@url \@url }%
\providecommand \@url [1]{\endgroup\@href {#1}{\urlprefix }}%
\providecommand \urlprefix  [0]{URL }%
\providecommand \Eprint [0]{\href }%
\providecommand \doibase [0]{https://doi.org/}%
\providecommand \selectlanguage [0]{\@gobble}%
\providecommand \bibinfo  [0]{\@secondoftwo}%
\providecommand \bibfield  [0]{\@secondoftwo}%
\providecommand \translation [1]{[#1]}%
\providecommand \BibitemOpen [0]{}%
\providecommand \bibitemStop [0]{}%
\providecommand \bibitemNoStop [0]{.\EOS\space}%
\providecommand \EOS [0]{\spacefactor3000\relax}%
\providecommand \BibitemShut  [1]{\csname bibitem#1\endcsname}%
\let\auto@bib@innerbib\@empty
\bibitem [{\citenamefont {Schleich}\ \emph {et~al.}(2016)\citenamefont
  {Schleich} \emph {et~al.}}]{Schleich_etal_2016}%
  \BibitemOpen
  \bibfield  {author} {\bibinfo {author} {\bibfnamefont {W.~P.}\ \bibnamefont
  {Schleich}} \emph {et~al.},\ }\href
  {https://doi.org/10.1007/s00340-016-6353-8} {\bibfield  {journal} {\bibinfo
  {journal} {Applied Physics B}\ }\textbf {\bibinfo {volume} {122}},\ \bibinfo
  {pages} {130} (\bibinfo {year} {2016})}\BibitemShut {NoStop}%
\bibitem [{\citenamefont {Ac{\'i}n}\ \emph {et~al.}(2018)\citenamefont
  {Ac{\'i}n} \emph {et~al.}}]{Acin_etal_2018}%
  \BibitemOpen
  \bibfield  {author} {\bibinfo {author} {\bibfnamefont {A.}~\bibnamefont
  {Ac{\'i}n}} \emph {et~al.},\ }\href
  {https://doi.org/10.1088/1367-2630/aad1ea} {\bibfield  {journal} {\bibinfo
  {journal} {New Journal of Physics}\ }\textbf {\bibinfo {volume} {20}},\
  \bibinfo {pages} {080201} (\bibinfo {year} {2018})}\BibitemShut {NoStop}%
\bibitem [{\citenamefont {Koch}\ \emph {et~al.}(2022)\citenamefont {Koch} \emph
  {et~al.}}]{Koch_etal_2022}%
  \BibitemOpen
  \bibfield  {author} {\bibinfo {author} {\bibfnamefont {C.~P.}\ \bibnamefont
  {Koch}} \emph {et~al.},\ }\href
  {https://doi.org/10.1140/epjqt/s40507-022-00138-x} {\bibfield  {journal}
  {\bibinfo  {journal} {EPJ Quantum Technology}\ }\textbf {\bibinfo {volume}
  {9}},\ \bibinfo {pages} {19} (\bibinfo {year} {2022})}\BibitemShut {NoStop}%
\bibitem [{\citenamefont {Brif}\ \emph {et~al.}(2010)\citenamefont {Brif},
  \citenamefont {Chakrabarti},\ and\ \citenamefont
  {Rabitz}}]{Brif_Chakrabarti_Rabitz_2010}%
  \BibitemOpen
  \bibfield  {author} {\bibinfo {author} {\bibfnamefont {C.}~\bibnamefont
  {Brif}}, \bibinfo {author} {\bibfnamefont {R.}~\bibnamefont {Chakrabarti}},\
  and\ \bibinfo {author} {\bibfnamefont {H.}~\bibnamefont {Rabitz}},\ }\href
  {https://doi.org/10.1088/1367-2630/12/7/075008} {\bibfield  {journal}
  {\bibinfo  {journal} {New Journal of Physics}\ }\textbf {\bibinfo {volume}
  {12}},\ \bibinfo {pages} {075008} (\bibinfo {year} {2010})}\BibitemShut
  {NoStop}%
\bibitem [{\citenamefont {Shapiro}\ and\ \citenamefont
  {Brumer}(2012)}]{ShapiroBrumerBook2012}%
  \BibitemOpen
  \bibfield  {author} {\bibinfo {author} {\bibfnamefont {M.}~\bibnamefont
  {Shapiro}}\ and\ \bibinfo {author} {\bibfnamefont {P.}~\bibnamefont
  {Brumer}},\ }\href@noop {} {\emph {\bibinfo {title} {{Quantum Control of
  Molecular Processes}}}},\ \bibinfo {edition} {2nd}\ ed.\ (\bibinfo
  {publisher} {Wiley-VCH},\ \bibinfo {year} {2012})\BibitemShut {NoStop}%
\bibitem [{\citenamefont {Koch}(2016)}]{Koch_2016}%
  \BibitemOpen
  \bibfield  {author} {\bibinfo {author} {\bibfnamefont {C.~P.}\ \bibnamefont
  {Koch}},\ }\href {https://doi.org/10.1088/0953-8984/28/21/213001} {\bibfield
  {journal} {\bibinfo  {journal} {Journal of Physics: Condensed Matter}\
  }\textbf {\bibinfo {volume} {28}},\ \bibinfo {pages} {213001} (\bibinfo
  {year} {2016})}\BibitemShut {NoStop}%
\bibitem [{\citenamefont {Cong}(2018)}]{Cong_2018}%
  \BibitemOpen
  \bibfield  {author} {\bibinfo {author} {\bibfnamefont {S.}~\bibnamefont
  {Cong}},\ }\href@noop {} {\emph {\bibinfo {title} {Control of Quantum
  Systems. Theory and Methods}}}\ (\bibinfo  {publisher} {John Wiley and Sons
  Limited},\ \bibinfo {year} {2018})\BibitemShut {NoStop}%
\bibitem [{\citenamefont {D'Alessandro}(2021)}]{AlessandroBook2021}%
  \BibitemOpen
  \bibfield  {author} {\bibinfo {author} {\bibfnamefont {D.}~\bibnamefont
  {D'Alessandro}},\ }\href@noop {} {\emph {\bibinfo {title} {Introduction to
  Quantum Control and Dynamics}}},\ \bibinfo {edition} {2nd}\ ed.\ (\bibinfo
  {publisher} {Chapman and Hall/CRC},\ \bibinfo {address} {Boca Raton},\
  \bibinfo {year} {2021})\BibitemShut {NoStop}%
\bibitem [{\citenamefont {Kurizki}\ and\ \citenamefont
  {Kofman}(2021)}]{Kurizki_Kofman_2021}%
  \BibitemOpen
  \bibfield  {author} {\bibinfo {author} {\bibfnamefont {G.}~\bibnamefont
  {Kurizki}}\ and\ \bibinfo {author} {\bibfnamefont {A.~G.}\ \bibnamefont
  {Kofman}},\ }\href {https://doi.org/10.1017/9781316798454} {\emph {\bibinfo
  {title} {Thermodynamics and Control of Open Quantum Systems}}},\ \bibinfo
  {edition} {1st}\ ed.\ (\bibinfo  {publisher} {Cambridge University Press},\
  \bibinfo {year} {2021})\BibitemShut {NoStop}%
\bibitem [{\citenamefont {Kallush}\ \emph {et~al.}(2022)\citenamefont
  {Kallush}, \citenamefont {Dann},\ and\ \citenamefont
  {Kosloff}}]{KallushDannKosloff2022}%
  \BibitemOpen
  \bibfield  {author} {\bibinfo {author} {\bibfnamefont {S.}~\bibnamefont
  {Kallush}}, \bibinfo {author} {\bibfnamefont {R.}~\bibnamefont {Dann}},\ and\
  \bibinfo {author} {\bibfnamefont {R.}~\bibnamefont {Kosloff}},\ }\href
  {https://doi.org/10.1126/sciadv.add0828} {\bibfield  {journal} {\bibinfo
  {journal} {Science Advances}\ }\textbf {\bibinfo {volume} {8}},\ \bibinfo
  {pages} {eadd0828} (\bibinfo {year} {2022})}\BibitemShut {NoStop}%
\bibitem [{\citenamefont {Breuer}\ and\ \citenamefont
  {Petruccione}(2007)}]{BreuerPetruccioneBook2007}%
  \BibitemOpen
  \bibfield  {author} {\bibinfo {author} {\bibfnamefont {H.-P.}\ \bibnamefont
  {Breuer}}\ and\ \bibinfo {author} {\bibfnamefont {F.}~\bibnamefont
  {Petruccione}},\ }\href@noop {} {\emph {\bibinfo {title} {The Theory of Open
  Quantum Systems}}}\ (\bibinfo  {publisher} {Oxford University Press},\
  \bibinfo {address} {Oxford},\ \bibinfo {year} {2007})\BibitemShut {NoStop}%
\bibitem [{\citenamefont {Beige}\ \emph {et~al.}(2000)\citenamefont {Beige},
  \citenamefont {Braun}, \citenamefont {Tregenna},\ and\ \citenamefont
  {Knight}}]{Beige_Braun_Tregenna_Knight_2000}%
  \BibitemOpen
  \bibfield  {author} {\bibinfo {author} {\bibfnamefont {A.}~\bibnamefont
  {Beige}}, \bibinfo {author} {\bibfnamefont {D.}~\bibnamefont {Braun}},
  \bibinfo {author} {\bibfnamefont {B.}~\bibnamefont {Tregenna}},\ and\
  \bibinfo {author} {\bibfnamefont {P.~L.}\ \bibnamefont {Knight}},\ }\href
  {https://doi.org/10.1103/PhysRevLett.85.1762} {\bibfield  {journal} {\bibinfo
   {journal} {Physical Review Letters}\ }\textbf {\bibinfo {volume} {85}},\
  \bibinfo {pages} {1762} (\bibinfo {year} {2000})}\BibitemShut {NoStop}%
\bibitem [{\citenamefont {Diehl}\ \emph {et~al.}(2008)\citenamefont {Diehl},
  \citenamefont {Micheli}, \citenamefont {Kantian}, \citenamefont {Kraus},
  \citenamefont {B{\"u}chler},\ and\ \citenamefont {Zoller}}]{Diehl_etal_2008}%
  \BibitemOpen
  \bibfield  {author} {\bibinfo {author} {\bibfnamefont {S.}~\bibnamefont
  {Diehl}}, \bibinfo {author} {\bibfnamefont {A.}~\bibnamefont {Micheli}},
  \bibinfo {author} {\bibfnamefont {A.}~\bibnamefont {Kantian}}, \bibinfo
  {author} {\bibfnamefont {B.}~\bibnamefont {Kraus}}, \bibinfo {author}
  {\bibfnamefont {H.~P.}\ \bibnamefont {B{\"u}chler}},\ and\ \bibinfo {author}
  {\bibfnamefont {P.}~\bibnamefont {Zoller}},\ }\href
  {https://doi.org/10.1038/nphys1073} {\bibfield  {journal} {\bibinfo
  {journal} {Nature Physics}\ }\textbf {\bibinfo {volume} {4}},\ \bibinfo
  {pages} {878} (\bibinfo {year} {2008})}\BibitemShut {NoStop}%
\bibitem [{\citenamefont {Verstraete}\ \emph {et~al.}(2009)\citenamefont
  {Verstraete}, \citenamefont {Wolf},\ and\ \citenamefont
  {Cirac}}]{Verstraete_Wolf_Cirac_2009}%
  \BibitemOpen
  \bibfield  {author} {\bibinfo {author} {\bibfnamefont {F.}~\bibnamefont
  {Verstraete}}, \bibinfo {author} {\bibfnamefont {M.~M.}\ \bibnamefont
  {Wolf}},\ and\ \bibinfo {author} {\bibfnamefont {J.~I.}\ \bibnamefont
  {Cirac}},\ }\href {https://doi.org/10.1038/nphys1342} {\bibfield  {journal}
  {\bibinfo  {journal} {Nature Physics}\ }\textbf {\bibinfo {volume} {5}},\
  \bibinfo {pages} {633} (\bibinfo {year} {2009})}\BibitemShut {NoStop}%
\bibitem [{\citenamefont {Schmidt}\ \emph {et~al.}(2011)\citenamefont
  {Schmidt}, \citenamefont {Negretti}, \citenamefont {Ankerhold}, \citenamefont
  {Calarco},\ and\ \citenamefont
  {Stockburger}}]{Schmidt_Negretti_Ankerhold_Calarco_Stockburger_2011}%
  \BibitemOpen
  \bibfield  {author} {\bibinfo {author} {\bibfnamefont {R.}~\bibnamefont
  {Schmidt}}, \bibinfo {author} {\bibfnamefont {A.}~\bibnamefont {Negretti}},
  \bibinfo {author} {\bibfnamefont {J.}~\bibnamefont {Ankerhold}}, \bibinfo
  {author} {\bibfnamefont {T.}~\bibnamefont {Calarco}},\ and\ \bibinfo {author}
  {\bibfnamefont {J.~T.}\ \bibnamefont {Stockburger}},\ }\href
  {https://doi.org/10.1103/PhysRevLett.107.130404} {\bibfield  {journal}
  {\bibinfo  {journal} {Physical Review Letters}\ }\textbf {\bibinfo {volume}
  {107}},\ \bibinfo {pages} {130404} (\bibinfo {year} {2011})}\BibitemShut
  {NoStop}%
\bibitem [{\citenamefont {Harrington}\ \emph {et~al.}(2022)\citenamefont
  {Harrington}, \citenamefont {Mueller},\ and\ \citenamefont
  {Murch}}]{HarringtonMuellerMurch2022}%
  \BibitemOpen
  \bibfield  {author} {\bibinfo {author} {\bibfnamefont {P.~M.}\ \bibnamefont
  {Harrington}}, \bibinfo {author} {\bibfnamefont {E.~J.}\ \bibnamefont
  {Mueller}},\ and\ \bibinfo {author} {\bibfnamefont {K.~W.}\ \bibnamefont
  {Murch}},\ }\href {https://doi.org/10.1038/s42254-022-00494-8} {\bibfield
  {journal} {\bibinfo  {journal} {Nature Reviews Physics}\ }\textbf {\bibinfo
  {volume} {4}},\ \bibinfo {pages} {660} (\bibinfo {year} {2022})}\BibitemShut
  {NoStop}%
\bibitem [{\citenamefont {Cole}\ \emph {et~al.}(2022)\citenamefont {Cole},
  \citenamefont {Erickson}, \citenamefont {Zarantonello}, \citenamefont {Horn},
  \citenamefont {Hou}, \citenamefont {Wu}, \citenamefont {Slichter},
  \citenamefont {Reiter}, \citenamefont {Koch},\ and\ \citenamefont
  {Leibfried}}]{Coleetal2022}%
  \BibitemOpen
  \bibfield  {author} {\bibinfo {author} {\bibfnamefont {D.~C.}\ \bibnamefont
  {Cole}}, \bibinfo {author} {\bibfnamefont {S.~D.}\ \bibnamefont {Erickson}},
  \bibinfo {author} {\bibfnamefont {G.}~\bibnamefont {Zarantonello}}, \bibinfo
  {author} {\bibfnamefont {K.~P.}\ \bibnamefont {Horn}}, \bibinfo {author}
  {\bibfnamefont {P.-Y.}\ \bibnamefont {Hou}}, \bibinfo {author} {\bibfnamefont
  {J.~J.}\ \bibnamefont {Wu}}, \bibinfo {author} {\bibfnamefont {D.~H.}\
  \bibnamefont {Slichter}}, \bibinfo {author} {\bibfnamefont {F.}~\bibnamefont
  {Reiter}}, \bibinfo {author} {\bibfnamefont {C.~P.}\ \bibnamefont {Koch}},\
  and\ \bibinfo {author} {\bibfnamefont {D.}~\bibnamefont {Leibfried}},\ }\href
  {https://doi.org/10.1103/PhysRevLett.128.080502} {\bibfield  {journal}
  {\bibinfo  {journal} {Physical Review Letters}\ }\textbf {\bibinfo {volume}
  {128}},\ \bibinfo {pages} {080502} (\bibinfo {year} {2022})}\BibitemShut
  {NoStop}%
\bibitem [{\citenamefont {Malinowski}\ \emph {et~al.}(2022)\citenamefont
  {Malinowski}, \citenamefont {Zhang}, \citenamefont {Negnevitsky},
  \citenamefont {Rojkov}, \citenamefont {Reiter}, \citenamefont {Nguyen},
  \citenamefont {Stadler}, \citenamefont {Kienzler}, \citenamefont {Mehta},\
  and\ \citenamefont {Home}}]{Malinowskietal2022}%
  \BibitemOpen
  \bibfield  {author} {\bibinfo {author} {\bibfnamefont {M.}~\bibnamefont
  {Malinowski}}, \bibinfo {author} {\bibfnamefont {C.}~\bibnamefont {Zhang}},
  \bibinfo {author} {\bibfnamefont {V.}~\bibnamefont {Negnevitsky}}, \bibinfo
  {author} {\bibfnamefont {I.}~\bibnamefont {Rojkov}}, \bibinfo {author}
  {\bibfnamefont {F.}~\bibnamefont {Reiter}}, \bibinfo {author} {\bibfnamefont
  {T.-L.}\ \bibnamefont {Nguyen}}, \bibinfo {author} {\bibfnamefont
  {M.}~\bibnamefont {Stadler}}, \bibinfo {author} {\bibfnamefont
  {D.}~\bibnamefont {Kienzler}}, \bibinfo {author} {\bibfnamefont {K.~K.}\
  \bibnamefont {Mehta}},\ and\ \bibinfo {author} {\bibfnamefont {J.~P.}\
  \bibnamefont {Home}},\ }\href
  {https://doi.org/10.1103/PhysRevLett.128.080503} {\bibfield  {journal}
  {\bibinfo  {journal} {Physical Review Letters}\ }\textbf {\bibinfo {volume}
  {128}},\ \bibinfo {pages} {080503} (\bibinfo {year} {2022})}\BibitemShut
  {NoStop}%
\bibitem [{\citenamefont {Pechen}\ and\ \citenamefont
  {Rabitz}(2006)}]{PechenRabitz2006}%
  \BibitemOpen
  \bibfield  {author} {\bibinfo {author} {\bibfnamefont {A.}~\bibnamefont
  {Pechen}}\ and\ \bibinfo {author} {\bibfnamefont {H.}~\bibnamefont
  {Rabitz}},\ }\href {https://doi.org/10.1103/PhysRevA.73.062102} {\bibfield
  {journal} {\bibinfo  {journal} {Phys. Rev. A}\ }\textbf {\bibinfo {volume}
  {73}},\ \bibinfo {pages} {062102} (\bibinfo {year} {2006})}\BibitemShut
  {NoStop}%
\bibitem [{\citenamefont {Pechen}(2011)}]{Pechen2011}%
  \BibitemOpen
  \bibfield  {author} {\bibinfo {author} {\bibfnamefont {A.}~\bibnamefont
  {Pechen}},\ }\href {https://doi.org/10.1103/PhysRevA.84.042106} {\bibfield
  {journal} {\bibinfo  {journal} {Phys. Rev. A}\ }\textbf {\bibinfo {volume}
  {84}},\ \bibinfo {pages} {042106} (\bibinfo {year} {2011})}\BibitemShut
  {NoStop}%
\bibitem [{\citenamefont {Huang}\ \emph {et~al.}(1983)\citenamefont {Huang},
  \citenamefont {Tarn},\ and\ \citenamefont {Clark}}]{Huang_Tarn_Clark_1983}%
  \BibitemOpen
  \bibfield  {author} {\bibinfo {author} {\bibfnamefont {G.~M.}\ \bibnamefont
  {Huang}}, \bibinfo {author} {\bibfnamefont {T.-J.}\ \bibnamefont {Tarn}},\
  and\ \bibinfo {author} {\bibfnamefont {J.~W.}\ \bibnamefont {Clark}},\ }\href
  {https://doi.org/10.1063/1.525634} {\bibfield  {journal} {\bibinfo  {journal}
  {Journal of Mathematical Physics}\ }\textbf {\bibinfo {volume} {24}},\
  \bibinfo {pages} {2608} (\bibinfo {year} {1983})}\BibitemShut {NoStop}%
\bibitem [{\citenamefont {Turinici}\ and\ \citenamefont
  {Rabitz}(2001)}]{Turinici2001}%
  \BibitemOpen
  \bibfield  {author} {\bibinfo {author} {\bibfnamefont {G.}~\bibnamefont
  {Turinici}}\ and\ \bibinfo {author} {\bibfnamefont {H.}~\bibnamefont
  {Rabitz}},\ }\href {https://doi.org/10.1016/S0301-0104(01)00216-6} {\bibfield
   {journal} {\bibinfo  {journal} {Chemical Physics}\ }\textbf {\bibinfo
  {volume} {267}},\ \bibinfo {pages} {1} (\bibinfo {year} {2001})}\BibitemShut
  {NoStop}%
\bibitem [{\citenamefont {Albertini}\ and\ \citenamefont
  {D'Alessandro}(2003)}]{Albertini2001}%
  \BibitemOpen
  \bibfield  {author} {\bibinfo {author} {\bibfnamefont {F.}~\bibnamefont
  {Albertini}}\ and\ \bibinfo {author} {\bibfnamefont {D.}~\bibnamefont
  {D'Alessandro}},\ }\href {https://doi.org/10.1109/TAC.2003.815027} {\bibfield
   {journal} {\bibinfo  {journal} {IEEE Transactions on Automatic Control}\
  }\textbf {\bibinfo {volume} {48}},\ \bibinfo {pages} {1399} (\bibinfo {year}
  {2003})}\BibitemShut {NoStop}%
\bibitem [{\citenamefont {Schirmer}\ \emph {et~al.}(2001)\citenamefont
  {Schirmer}, \citenamefont {Fu},\ and\ \citenamefont
  {Solomon}}]{Schirmer2001}%
  \BibitemOpen
  \bibfield  {author} {\bibinfo {author} {\bibfnamefont {S.~G.}\ \bibnamefont
  {Schirmer}}, \bibinfo {author} {\bibfnamefont {H.}~\bibnamefont {Fu}},\ and\
  \bibinfo {author} {\bibfnamefont {A.~I.}\ \bibnamefont {Solomon}},\ }\href
  {https://doi.org/10.1103/PhysRevA.63.063410} {\bibfield  {journal} {\bibinfo
  {journal} {Phys. Rev. A}\ }\textbf {\bibinfo {volume} {63}},\ \bibinfo
  {pages} {063410} (\bibinfo {year} {2001})}\BibitemShut {NoStop}%
\bibitem [{\citenamefont {Schirmer}\ \emph {et~al.}(2002)\citenamefont
  {Schirmer}, \citenamefont {Solomon},\ and\ \citenamefont
  {Leahy}}]{Schirmer2002}%
  \BibitemOpen
  \bibfield  {author} {\bibinfo {author} {\bibfnamefont {S.~G.}\ \bibnamefont
  {Schirmer}}, \bibinfo {author} {\bibfnamefont {A.~I.}\ \bibnamefont
  {Solomon}},\ and\ \bibinfo {author} {\bibfnamefont {J.~V.}\ \bibnamefont
  {Leahy}},\ }\href {https://doi.org/10.1088/0305-4470/35/40/313} {\bibfield
  {journal} {\bibinfo  {journal} {Journal of Physics A: Mathematical and
  General}\ }\textbf {\bibinfo {volume} {35}},\ \bibinfo {pages} {8551}
  (\bibinfo {year} {2002})}\BibitemShut {NoStop}%
\bibitem [{\citenamefont {Altafini}(2002)}]{Altafini2002}%
  \BibitemOpen
  \bibfield  {author} {\bibinfo {author} {\bibfnamefont {C.}~\bibnamefont
  {Altafini}},\ }\href {https://doi.org/10.1063/1.1467611} {\bibfield
  {journal} {\bibinfo  {journal} {J. Math. Phys.}\ }\textbf {\bibinfo {volume}
  {43}},\ \bibinfo {pages} {2051} (\bibinfo {year} {2002})}\BibitemShut
  {NoStop}%
\bibitem [{\citenamefont {Polack}\ \emph {et~al.}(2009)\citenamefont {Polack},
  \citenamefont {Suchowski},\ and\ \citenamefont {Tannor}}]{Polack2009}%
  \BibitemOpen
  \bibfield  {author} {\bibinfo {author} {\bibfnamefont {T.}~\bibnamefont
  {Polack}}, \bibinfo {author} {\bibfnamefont {H.}~\bibnamefont {Suchowski}},\
  and\ \bibinfo {author} {\bibfnamefont {D.~J.}\ \bibnamefont {Tannor}},\
  }\href {https://doi.org/10.1103/PhysRevA.79.053403} {\bibfield  {journal}
  {\bibinfo  {journal} {Phys. Rev. A}\ }\textbf {\bibinfo {volume} {79}},\
  \bibinfo {pages} {053403} (\bibinfo {year} {2009})}\BibitemShut {NoStop}%
\bibitem [{\citenamefont {Gorini}\ \emph {et~al.}(1976)\citenamefont {Gorini},
  \citenamefont {Kossakowski},\ and\ \citenamefont
  {Sudarshan}}]{GoriniKossakowskiSudarshan1976}%
  \BibitemOpen
  \bibfield  {author} {\bibinfo {author} {\bibfnamefont {V.}~\bibnamefont
  {Gorini}}, \bibinfo {author} {\bibfnamefont {A.}~\bibnamefont
  {Kossakowski}},\ and\ \bibinfo {author} {\bibfnamefont {E.~C.~G.}\
  \bibnamefont {Sudarshan}},\ }\href {https://doi.org/10.1063/1.522979}
  {\bibfield  {journal} {\bibinfo  {journal} {Journal of Mathematical Physics}\
  }\textbf {\bibinfo {volume} {17}},\ \bibinfo {pages} {821} (\bibinfo {year}
  {1976})}\BibitemShut {NoStop}%
\bibitem [{\citenamefont {Lindblad}(1976)}]{Lindblad1976}%
  \BibitemOpen
  \bibfield  {author} {\bibinfo {author} {\bibfnamefont {G.}~\bibnamefont
  {Lindblad}},\ }\href {https://doi.org/10.1007/BF01608499} {\bibfield
  {journal} {\bibinfo  {journal} {Communications in Mathematical Physics}\
  }\textbf {\bibinfo {volume} {48}},\ \bibinfo {pages} {119} (\bibinfo {year}
  {1976})}\BibitemShut {NoStop}%
\bibitem [{\citenamefont {Altafini}(2003{\natexlab{a}})}]{Altafini1}%
  \BibitemOpen
  \bibfield  {author} {\bibinfo {author} {\bibfnamefont {C.}~\bibnamefont
  {Altafini}},\ }\href {https://doi.org/10.1063/1.1571221} {\bibfield
  {journal} {\bibinfo  {journal} {J. Math. Phys.}\ }\textbf {\bibinfo {volume}
  {44}},\ \bibinfo {pages} {2357} (\bibinfo {year}
  {2003}{\natexlab{a}})}\BibitemShut {NoStop}%
\bibitem [{\citenamefont {Altafini}(2003{\natexlab{b}})}]{Altafini2}%
  \BibitemOpen
  \bibfield  {author} {\bibinfo {author} {\bibfnamefont {C.}~\bibnamefont
  {Altafini}},\ }in\ \href {https://doi.org/10.1109/PHYCON.2003.1236992} {\emph
  {\bibinfo {booktitle} {2003 IEEE International Workshop on Workload
  Characterization (IEEE Cat. No.03EX775)}}},\ Vol.~\bibinfo {volume} {3}\
  (\bibinfo  {publisher} {IEEE},\ \bibinfo {year} {2003})\ pp.\ \bibinfo
  {pages} {710--714}\BibitemShut {NoStop}%
\bibitem [{\citenamefont {Dirr}\ \emph {et~al.}(2009)\citenamefont {Dirr},
  \citenamefont {Helmke}, \citenamefont {Kurniawan},\ and\ \citenamefont
  {Schulte-Herbr{\"u}ggen}}]{Dirr_Helmke_Kurniawan_SchulteHerbrueggen_2009}%
  \BibitemOpen
  \bibfield  {author} {\bibinfo {author} {\bibfnamefont {G.}~\bibnamefont
  {Dirr}}, \bibinfo {author} {\bibfnamefont {U.}~\bibnamefont {Helmke}},
  \bibinfo {author} {\bibfnamefont {I.}~\bibnamefont {Kurniawan}},\ and\
  \bibinfo {author} {\bibfnamefont {T.}~\bibnamefont
  {Schulte-Herbr{\"u}ggen}},\ }\href
  {https://doi.org/10.1016/S0034-4877(09)90022-2} {\bibfield  {journal}
  {\bibinfo  {journal} {Reports on Mathematical Physics}\ }\textbf {\bibinfo
  {volume} {64}},\ \bibinfo {pages} {93} (\bibinfo {year} {2009})}\BibitemShut
  {NoStop}%
\bibitem [{\citenamefont {Agrachev}\ and\ \citenamefont
  {Sachkov}(2004)}]{AgrachevSachkov}%
  \BibitemOpen
  \bibfield  {author} {\bibinfo {author} {\bibfnamefont {A.~A.}\ \bibnamefont
  {Agrachev}}\ and\ \bibinfo {author} {\bibfnamefont {Y.~L.}\ \bibnamefont
  {Sachkov}},\ }\href@noop {} {\emph {\bibinfo {title} {Control Theory from the
  Geometric Viewpoint}}}\ (\bibinfo  {publisher} {Springer},\ \bibinfo {year}
  {2004})\BibitemShut {NoStop}%
\bibitem [{\citenamefont {Sugny}\ \emph {et~al.}(2007)\citenamefont {Sugny},
  \citenamefont {Kontz},\ and\ \citenamefont
  {Jauslin}}]{Sugny_Kontz_Jauslin_2007}%
  \BibitemOpen
  \bibfield  {author} {\bibinfo {author} {\bibfnamefont {D.}~\bibnamefont
  {Sugny}}, \bibinfo {author} {\bibfnamefont {C.}~\bibnamefont {Kontz}},\ and\
  \bibinfo {author} {\bibfnamefont {H.~R.}\ \bibnamefont {Jauslin}},\ }\href
  {https://doi.org/10.1103/PhysRevA.76.023419} {\bibfield  {journal} {\bibinfo
  {journal} {Phys. Rev. A}\ }\textbf {\bibinfo {volume} {76}},\ \bibinfo
  {pages} {023419} (\bibinfo {year} {2007})}\BibitemShut {NoStop}%
\bibitem [{\citenamefont {Lapert}\ \emph {et~al.}(2010)\citenamefont {Lapert},
  \citenamefont {Zhang}, \citenamefont {Braun}, \citenamefont {Glaser},\ and\
  \citenamefont {Sugny}}]{Lapert_Zhang_Braun_Glaser_Sugny_2010}%
  \BibitemOpen
  \bibfield  {author} {\bibinfo {author} {\bibfnamefont {M.}~\bibnamefont
  {Lapert}}, \bibinfo {author} {\bibfnamefont {Y.}~\bibnamefont {Zhang}},
  \bibinfo {author} {\bibfnamefont {M.}~\bibnamefont {Braun}}, \bibinfo
  {author} {\bibfnamefont {S.~J.}\ \bibnamefont {Glaser}},\ and\ \bibinfo
  {author} {\bibfnamefont {D.}~\bibnamefont {Sugny}},\ }\href
  {https://doi.org/10.1103/PhysRevLett.104.083001} {\bibfield  {journal}
  {\bibinfo  {journal} {Phys. Rev. Lett.}\ }\textbf {\bibinfo {volume} {104}},\
  \bibinfo {pages} {083001} (\bibinfo {year} {2010})}\BibitemShut {NoStop}%
\bibitem [{\citenamefont {Boscain}\ and\ \citenamefont
  {Piccoli}(2003)}]{BoscainPiccoli12004}%
  \BibitemOpen
  \bibfield  {author} {\bibinfo {author} {\bibfnamefont {U.}~\bibnamefont
  {Boscain}}\ and\ \bibinfo {author} {\bibfnamefont {B.}~\bibnamefont
  {Piccoli}},\ }\href@noop {} {\emph {\bibinfo {title} {Optimal Syntheses for
  Control Systems on 2-D manifolds}}}\ (\bibinfo  {publisher} {Springer},\
  \bibinfo {year} {2003})\BibitemShut {NoStop}%
\bibitem [{\citenamefont {Clark}\ \emph {et~al.}(2020)\citenamefont {Clark},
  \citenamefont {Bloch}, \citenamefont {Colombo},\ and\ \citenamefont
  {Rooney}}]{Clark_Bloch_Colombo_Rooney_2020}%
  \BibitemOpen
  \bibfield  {author} {\bibinfo {author} {\bibfnamefont {W.}~\bibnamefont
  {Clark}}, \bibinfo {author} {\bibfnamefont {A.}~\bibnamefont {Bloch}},
  \bibinfo {author} {\bibfnamefont {L.}~\bibnamefont {Colombo}},\ and\ \bibinfo
  {author} {\bibfnamefont {P.}~\bibnamefont {Rooney}},\ }\href
  {https://doi.org/10.3934/dcdss.2020063} {\bibfield  {journal} {\bibinfo
  {journal} {Discrete and Continuous Dynamical Systems - S}\ }\textbf {\bibinfo
  {volume} {13}},\ \bibinfo {pages} {1061} (\bibinfo {year}
  {2020})}\BibitemShut {NoStop}%
\bibitem [{\citenamefont {Mukherjee}\ \emph {et~al.}(2013)\citenamefont
  {Mukherjee}, \citenamefont {Carlini}, \citenamefont {Mari}, \citenamefont
  {Caneva}, \citenamefont {Montangero}, \citenamefont {Calarco}, \citenamefont
  {Fazio},\ and\ \citenamefont
  {Giovannetti}}]{Mukherjee_Carlini_Mari_Caneva_Montangero_Calarco_Fazio_Giovannetti_2013}%
  \BibitemOpen
  \bibfield  {author} {\bibinfo {author} {\bibfnamefont {V.}~\bibnamefont
  {Mukherjee}}, \bibinfo {author} {\bibfnamefont {A.}~\bibnamefont {Carlini}},
  \bibinfo {author} {\bibfnamefont {A.}~\bibnamefont {Mari}}, \bibinfo {author}
  {\bibfnamefont {T.}~\bibnamefont {Caneva}}, \bibinfo {author} {\bibfnamefont
  {S.}~\bibnamefont {Montangero}}, \bibinfo {author} {\bibfnamefont
  {T.}~\bibnamefont {Calarco}}, \bibinfo {author} {\bibfnamefont
  {R.}~\bibnamefont {Fazio}},\ and\ \bibinfo {author} {\bibfnamefont
  {V.}~\bibnamefont {Giovannetti}},\ }\href
  {https://doi.org/10.1103/PhysRevA.88.062326} {\bibfield  {journal} {\bibinfo
  {journal} {Phys. Rev. A}\ }\textbf {\bibinfo {volume} {88}},\ \bibinfo
  {pages} {062326} (\bibinfo {year} {2013})}\BibitemShut {NoStop}%
\bibitem [{\citenamefont {del Campo}\ \emph {et~al.}(2013)\citenamefont {del
  Campo}, \citenamefont {Egusquiza}, \citenamefont {Plenio},\ and\
  \citenamefont {Huelga}}]{delCampo_Egusquiza_Plenio_Huelga_2013}%
  \BibitemOpen
  \bibfield  {author} {\bibinfo {author} {\bibfnamefont {A.}~\bibnamefont {del
  Campo}}, \bibinfo {author} {\bibfnamefont {I.~L.}\ \bibnamefont {Egusquiza}},
  \bibinfo {author} {\bibfnamefont {M.~B.}\ \bibnamefont {Plenio}},\ and\
  \bibinfo {author} {\bibfnamefont {S.~F.}\ \bibnamefont {Huelga}},\ }\href
  {https://doi.org/10.1103/PhysRevLett.110.050403} {\bibfield  {journal}
  {\bibinfo  {journal} {Phys. Rev. Lett.}\ }\textbf {\bibinfo {volume} {110}},\
  \bibinfo {pages} {050403} (\bibinfo {year} {2013})}\BibitemShut {NoStop}%
\bibitem [{\citenamefont {Zhang}\ \emph {et~al.}(2011)\citenamefont {Zhang},
  \citenamefont {Lapert}, \citenamefont {Sugny}, \citenamefont {Braun},\ and\
  \citenamefont {Glaser}}]{Zhang_Lapert_Sugny_Braun_Glaser_2011}%
  \BibitemOpen
  \bibfield  {author} {\bibinfo {author} {\bibfnamefont {Y.}~\bibnamefont
  {Zhang}}, \bibinfo {author} {\bibfnamefont {M.}~\bibnamefont {Lapert}},
  \bibinfo {author} {\bibfnamefont {D.}~\bibnamefont {Sugny}}, \bibinfo
  {author} {\bibfnamefont {M.}~\bibnamefont {Braun}},\ and\ \bibinfo {author}
  {\bibfnamefont {S.~J.}\ \bibnamefont {Glaser}},\ }\href
  {https://doi.org/10.1063/1.3543796} {\bibfield  {journal} {\bibinfo
  {journal} {J. Chem. Phys.}\ }\textbf {\bibinfo {volume} {134}},\ \bibinfo
  {pages} {054103} (\bibinfo {year} {2011})}\BibitemShut {NoStop}%
\bibitem [{\citenamefont {Fassioli}\ \emph {et~al.}(2014)\citenamefont
  {Fassioli}, \citenamefont {Dinshaw}, \citenamefont {Arpin},\ and\
  \citenamefont {Scholes}}]{Fassioli_Dinshaw_Arpin_Scholes_2014}%
  \BibitemOpen
  \bibfield  {author} {\bibinfo {author} {\bibfnamefont {F.}~\bibnamefont
  {Fassioli}}, \bibinfo {author} {\bibfnamefont {R.}~\bibnamefont {Dinshaw}},
  \bibinfo {author} {\bibfnamefont {P.~C.}\ \bibnamefont {Arpin}},\ and\
  \bibinfo {author} {\bibfnamefont {G.~D.}\ \bibnamefont {Scholes}},\ }\href
  {https://doi.org/10.1098/rsif.2013.0901} {\bibfield  {journal} {\bibinfo
  {journal} {J. R. Soc. Interface}\ }\textbf {\bibinfo {volume} {11}},\
  \bibinfo {pages} {20130901} (\bibinfo {year} {2014})}\BibitemShut {NoStop}%
\bibitem [{\citenamefont {Kozyrev}\ and\ \citenamefont
  {Pechen}(2022)}]{Kozyrev_Pechen_2022}%
  \BibitemOpen
  \bibfield  {author} {\bibinfo {author} {\bibfnamefont {S.~V.}\ \bibnamefont
  {Kozyrev}}\ and\ \bibinfo {author} {\bibfnamefont {A.~N.}\ \bibnamefont
  {Pechen}},\ }\href {https://doi.org/10.1103/PhysRevA.106.032218} {\bibfield
  {journal} {\bibinfo  {journal} {Phys. Rev. A}\ }\textbf {\bibinfo {volume}
  {106}},\ \bibinfo {pages} {032218} (\bibinfo {year} {2022})}\BibitemShut
  {NoStop}%
\bibitem [{\citenamefont {Viola}\ and\ \citenamefont
  {Lloyd}(1998)}]{Viola_Lloyd_1998}%
  \BibitemOpen
  \bibfield  {author} {\bibinfo {author} {\bibfnamefont {L.}~\bibnamefont
  {Viola}}\ and\ \bibinfo {author} {\bibfnamefont {S.}~\bibnamefont {Lloyd}},\
  }\href {https://doi.org/10.1103/PhysRevA.58.2733} {\bibfield  {journal}
  {\bibinfo  {journal} {Phys. Rev. A}\ }\textbf {\bibinfo {volume} {58}},\
  \bibinfo {pages} {2733} (\bibinfo {year} {1998})}\BibitemShut {NoStop}%
\bibitem [{\citenamefont {Viola}\ \emph
  {et~al.}(1999{\natexlab{a}})\citenamefont {Viola}, \citenamefont {Knill},\
  and\ \citenamefont {Lloyd}}]{Viola_Knill_Lloyd_1999}%
  \BibitemOpen
  \bibfield  {author} {\bibinfo {author} {\bibfnamefont {L.}~\bibnamefont
  {Viola}}, \bibinfo {author} {\bibfnamefont {E.}~\bibnamefont {Knill}},\ and\
  \bibinfo {author} {\bibfnamefont {S.}~\bibnamefont {Lloyd}},\ }\href
  {https://doi.org/10.1103/PhysRevLett.82.2417} {\bibfield  {journal} {\bibinfo
   {journal} {Phys. Rev. Lett.}\ }\textbf {\bibinfo {volume} {82}},\ \bibinfo
  {pages} {2417} (\bibinfo {year} {1999}{\natexlab{a}})}\BibitemShut {NoStop}%
\bibitem [{\citenamefont {Viola}\ \emph
  {et~al.}(1999{\natexlab{b}})\citenamefont {Viola}, \citenamefont {Lloyd},\
  and\ \citenamefont {Knill}}]{Viola_Lloyd_Knill_1999}%
  \BibitemOpen
  \bibfield  {author} {\bibinfo {author} {\bibfnamefont {L.}~\bibnamefont
  {Viola}}, \bibinfo {author} {\bibfnamefont {S.}~\bibnamefont {Lloyd}},\ and\
  \bibinfo {author} {\bibfnamefont {E.}~\bibnamefont {Knill}},\ }\href
  {https://doi.org/10.1103/PhysRevLett.83.4888} {\bibfield  {journal} {\bibinfo
   {journal} {Phys. Rev. Lett.}\ }\textbf {\bibinfo {volume} {83}},\ \bibinfo
  {pages} {4888} (\bibinfo {year} {1999}{\natexlab{b}})}\BibitemShut {NoStop}%
\bibitem [{\citenamefont {Grace}\ \emph {et~al.}(2007)\citenamefont {Grace},
  \citenamefont {Brif}, \citenamefont {Rabitz}, \citenamefont {Walmsley},
  \citenamefont {Kosut},\ and\ \citenamefont
  {Lidar}}]{Grace_Brif_Rabitz_Walmsley_Kosut_Lidar_2007}%
  \BibitemOpen
  \bibfield  {author} {\bibinfo {author} {\bibfnamefont {M.}~\bibnamefont
  {Grace}}, \bibinfo {author} {\bibfnamefont {C.}~\bibnamefont {Brif}},
  \bibinfo {author} {\bibfnamefont {H.}~\bibnamefont {Rabitz}}, \bibinfo
  {author} {\bibfnamefont {I.~A.}\ \bibnamefont {Walmsley}}, \bibinfo {author}
  {\bibfnamefont {R.~L.}\ \bibnamefont {Kosut}},\ and\ \bibinfo {author}
  {\bibfnamefont {D.~A.}\ \bibnamefont {Lidar}},\ }\href
  {https://doi.org/10.1088/0953-4075/40/9/S06} {\bibfield  {journal} {\bibinfo
  {journal} {J. Phys. B: At. Mol. Opt. Phys.}\ }\textbf {\bibinfo {volume}
  {40}},\ \bibinfo {pages} {S103} (\bibinfo {year} {2007})}\BibitemShut
  {NoStop}%
\bibitem [{\citenamefont {West}\ \emph {et~al.}(2010)\citenamefont {West},
  \citenamefont {Lidar}, \citenamefont {Fong},\ and\ \citenamefont
  {Gyure}}]{West_Lidar_Fong_Gyure_2010}%
  \BibitemOpen
  \bibfield  {author} {\bibinfo {author} {\bibfnamefont {J.~R.}\ \bibnamefont
  {West}}, \bibinfo {author} {\bibfnamefont {D.~A.}\ \bibnamefont {Lidar}},
  \bibinfo {author} {\bibfnamefont {B.~H.}\ \bibnamefont {Fong}},\ and\
  \bibinfo {author} {\bibfnamefont {M.~F.}\ \bibnamefont {Gyure}},\ }\href
  {https://doi.org/10.1103/PhysRevLett.105.230503} {\bibfield  {journal}
  {\bibinfo  {journal} {Phys. Rev. Lett.}\ }\textbf {\bibinfo {volume} {105}},\
  \bibinfo {pages} {230503} (\bibinfo {year} {2010})}\BibitemShut {NoStop}%
\bibitem [{\citenamefont {Zahedinejad}\ \emph {et~al.}(2015)\citenamefont
  {Zahedinejad}, \citenamefont {Ghosh},\ and\ \citenamefont
  {Sanders}}]{Zahedinejad_Ghosh_Sanders_2015}%
  \BibitemOpen
  \bibfield  {author} {\bibinfo {author} {\bibfnamefont {E.}~\bibnamefont
  {Zahedinejad}}, \bibinfo {author} {\bibfnamefont {J.}~\bibnamefont {Ghosh}},\
  and\ \bibinfo {author} {\bibfnamefont {B.~C.}\ \bibnamefont {Sanders}},\
  }\href {https://doi.org/10.1103/PhysRevLett.114.200502} {\bibfield  {journal}
  {\bibinfo  {journal} {Phys. Rev. Lett.}\ }\textbf {\bibinfo {volume} {114}},\
  \bibinfo {pages} {200502} (\bibinfo {year} {2015})}\BibitemShut {NoStop}%
\bibitem [{\citenamefont {Caneva}\ \emph {et~al.}(2009)\citenamefont {Caneva},
  \citenamefont {Murphy}, \citenamefont {Calarco}, \citenamefont {Fazio},
  \citenamefont {Montangero}, \citenamefont {Giovannetti},\ and\ \citenamefont
  {Santoro}}]{Caneva_2009}%
  \BibitemOpen
  \bibfield  {author} {\bibinfo {author} {\bibfnamefont {T.}~\bibnamefont
  {Caneva}}, \bibinfo {author} {\bibfnamefont {M.}~\bibnamefont {Murphy}},
  \bibinfo {author} {\bibfnamefont {T.}~\bibnamefont {Calarco}}, \bibinfo
  {author} {\bibfnamefont {R.}~\bibnamefont {Fazio}}, \bibinfo {author}
  {\bibfnamefont {S.}~\bibnamefont {Montangero}}, \bibinfo {author}
  {\bibfnamefont {V.}~\bibnamefont {Giovannetti}},\ and\ \bibinfo {author}
  {\bibfnamefont {G.~E.}\ \bibnamefont {Santoro}},\ }\href
  {https://doi.org/10.1103/PhysRevLett.103.240501} {\bibfield  {journal}
  {\bibinfo  {journal} {Phys. Rev. Lett.}\ }\textbf {\bibinfo {volume} {103}},\
  \bibinfo {pages} {240501} (\bibinfo {year} {2009})}\BibitemShut {NoStop}%
\bibitem [{\citenamefont {Hegerfeldt}(2013)}]{Hegerfeldt_2013}%
  \BibitemOpen
  \bibfield  {author} {\bibinfo {author} {\bibfnamefont {G.~C.}\ \bibnamefont
  {Hegerfeldt}},\ }\href {https://doi.org/10.1103/PhysRevLett.111.260501}
  {\bibfield  {journal} {\bibinfo  {journal} {Phys. Rev. Lett.}\ }\textbf
  {\bibinfo {volume} {111}},\ \bibinfo {pages} {260501} (\bibinfo {year}
  {2013})}\BibitemShut {NoStop}%
\bibitem [{\citenamefont {Palao}\ and\ \citenamefont
  {Kosloff}(2002)}]{Palao_Kosloff_2002}%
  \BibitemOpen
  \bibfield  {author} {\bibinfo {author} {\bibfnamefont {J.~P.}\ \bibnamefont
  {Palao}}\ and\ \bibinfo {author} {\bibfnamefont {R.}~\bibnamefont
  {Kosloff}},\ }\href {https://doi.org/10.1103/PhysRevLett.89.188301}
  {\bibfield  {journal} {\bibinfo  {journal} {Phys. Rev. Lett.}\ }\textbf
  {\bibinfo {volume} {89}},\ \bibinfo {pages} {188301} (\bibinfo {year}
  {2002})}\BibitemShut {NoStop}%
\bibitem [{\citenamefont {Li}\ \emph {et~al.}(2016)\citenamefont {Li},
  \citenamefont {Lu}, \citenamefont {Luo}, \citenamefont {Laflamme},
  \citenamefont {Peng},\ and\ \citenamefont
  {Du}}]{Li_Lu_Luo_Laflamme_Peng_Du_2016}%
  \BibitemOpen
  \bibfield  {author} {\bibinfo {author} {\bibfnamefont {J.}~\bibnamefont
  {Li}}, \bibinfo {author} {\bibfnamefont {D.}~\bibnamefont {Lu}}, \bibinfo
  {author} {\bibfnamefont {Z.}~\bibnamefont {Luo}}, \bibinfo {author}
  {\bibfnamefont {R.}~\bibnamefont {Laflamme}}, \bibinfo {author}
  {\bibfnamefont {X.}~\bibnamefont {Peng}},\ and\ \bibinfo {author}
  {\bibfnamefont {J.}~\bibnamefont {Du}},\ }\href
  {https://doi.org/10.1103/PhysRevA.94.012312} {\bibfield  {journal} {\bibinfo
  {journal} {Phys. Rev. A}\ }\textbf {\bibinfo {volume} {94}},\ \bibinfo
  {pages} {012312} (\bibinfo {year} {2016})}\BibitemShut {NoStop}%
\bibitem [{\citenamefont {Lokutsievskiy}\ and\ \citenamefont
  {Pechen}(2021)}]{LokutsievskiyPechen1}%
  \BibitemOpen
  \bibfield  {author} {\bibinfo {author} {\bibfnamefont {L.}~\bibnamefont
  {Lokutsievskiy}}\ and\ \bibinfo {author} {\bibfnamefont {A.}~\bibnamefont
  {Pechen}},\ }\href {https://doi.org/10.1088/1751-8121/ac19f8} {\bibfield
  {journal} {\bibinfo  {journal} {Journal of Physics A: Mathematical and
  Theoretical}\ }\textbf {\bibinfo {volume} {54}},\ \bibinfo {pages} {395304}
  (\bibinfo {year} {2021})}\BibitemShut {NoStop}%
\bibitem [{\citenamefont {Lokutsievskiy}\ \emph {et~al.}(2024)\citenamefont
  {Lokutsievskiy}, \citenamefont {Pechen},\ and\ \citenamefont
  {Zelikin}}]{LokutsievskiyPechen2}%
  \BibitemOpen
  \bibfield  {author} {\bibinfo {author} {\bibfnamefont {L.~V.}\ \bibnamefont
  {Lokutsievskiy}}, \bibinfo {author} {\bibfnamefont {A.~N.}\ \bibnamefont
  {Pechen}},\ and\ \bibinfo {author} {\bibfnamefont {M.~I.}\ \bibnamefont
  {Zelikin}},\ }\href {https://doi.org/10.1088/1751-8121/ad5396} {\bibfield
  {journal} {\bibinfo  {journal} {Journal of Physics A: Mathematical and
  Theoretical}\ }\textbf {\bibinfo {volume} {57}},\ \bibinfo {pages} {275302}
  (\bibinfo {year} {2024})}\BibitemShut {NoStop}%
\bibitem [{\citenamefont {Sussmann}(1982)}]{Sussmann_1982}%
  \BibitemOpen
  \bibfield  {author} {\bibinfo {author} {\bibfnamefont {H.~J.}\ \bibnamefont
  {Sussmann}},\ }in\ \href@noop {} {\emph {\bibinfo {booktitle} {Feedback
  Control of Linear and Nonlinear Systems}}},\ \bibinfo {series} {Lecture Notes
  in Control and Information Sciences}, Vol.~\bibinfo {volume} {39}\ (\bibinfo
  {publisher} {Springer},\ \bibinfo {year} {1982})\ pp.\ \bibinfo {pages}
  {244--260}\BibitemShut {NoStop}%
\bibitem [{\citenamefont {Sussmann}(1987)}]{Sussmann_1987}%
  \BibitemOpen
  \bibfield  {author} {\bibinfo {author} {\bibfnamefont {H.~J.}\ \bibnamefont
  {Sussmann}},\ }\href {https://doi.org/10.1137/0325025} {\bibfield  {journal}
  {\bibinfo  {journal} {SIAM Journal on Control and Optimization}\ }\textbf
  {\bibinfo {volume} {25}},\ \bibinfo {pages} {433} (\bibinfo {year}
  {1987})}\BibitemShut {NoStop}%
\bibitem [{\citenamefont {Scully}\ and\ \citenamefont
  {Zubairy}(1997)}]{Scully_Zubairy_1997}%
  \BibitemOpen
  \bibfield  {author} {\bibinfo {author} {\bibfnamefont {M.~O.}\ \bibnamefont
  {Scully}}\ and\ \bibinfo {author} {\bibfnamefont {M.~S.}\ \bibnamefont
  {Zubairy}},\ }\href@noop {} {\emph {\bibinfo {title} {Quantum Optics}}}\
  (\bibinfo  {publisher} {Cambridge University Press},\ \bibinfo {address}
  {Cambridge},\ \bibinfo {year} {1997})\BibitemShut {NoStop}%
\bibitem [{Note1()}]{Note1}%
  \BibitemOpen
  \bibinfo {note} {The set $\protect \mathcal {R}(\rho _0)$ of all quantum
  states reachable from a given initial state $\rho _0$ always contains
  $\protect \mathcal {R}=\protect \mathcal {R}(N)$. Moreover, if $\rho _0\in
  \protect \mathcal {R}$, then $\protect \mathcal {R}(\rho _0)=\protect
  \mathcal {R}$, see~\cite {LokutsievskiyPechen1}.}\BibitemShut {Stop}%
\bibitem [{Note2()}]{Note2}%
  \BibitemOpen
  \bibinfo {note} {Note that \cite {LokutsievskiyPechen2} uses a slightly
  different convention ($r_z\DOTSB \mapstochar \rightarrow -r_z$ and $u\DOTSB
  \mapstochar \rightarrow -u$) corresponding to the opposite choice of the
  raising and lowering operators $\sigma _\pm $.}\BibitemShut {Stop}%
\end{thebibliography}%


\providecommand{\noopsort}[1]{}\providecommand{\singleletter}[1]{#1}%
\begin{thebibliography}{16}%
\makeatletter
\providecommand \@ifxundefined [1]{%
 \@ifx{#1\undefined}
}%
\providecommand \@ifnum [1]{%
 \ifnum #1\expandafter \@firstoftwo
 \else \expandafter \@secondoftwo
 \fi
}%
\providecommand \@ifx [1]{%
 \ifx #1\expandafter \@firstoftwo
 \else \expandafter \@secondoftwo
 \fi
}%
\providecommand \natexlab [1]{#1}%
\providecommand \enquote  [1]{``#1''}%
\providecommand \bibnamefont  [1]{#1}%
\providecommand \bibfnamefont [1]{#1}%
\providecommand \citenamefont [1]{#1}%
\providecommand \href@noop [0]{\@secondoftwo}%
\providecommand \href [0]{\begingroup \@sanitize@url \@href}%
\providecommand \@href[1]{\@@startlink{#1}\@@href}%
\providecommand \@@href[1]{\endgroup#1\@@endlink}%
\providecommand \@sanitize@url [0]{\catcode `\\12\catcode `\$12\catcode
  `\&12\catcode `\#12\catcode `\^12\catcode `\_12\catcode `\%12\relax}%
\providecommand \@@startlink[1]{}%
\providecommand \@@endlink[0]{}%
\providecommand \url  [0]{\begingroup\@sanitize@url \@url }%
\providecommand \@url [1]{\endgroup\@href {#1}{\urlprefix }}%
\providecommand \urlprefix  [0]{URL }%
\providecommand \Eprint [0]{\href }%
\providecommand \doibase [0]{https://doi.org/}%
\providecommand \selectlanguage [0]{\@gobble}%
\providecommand \bibinfo  [0]{\@secondoftwo}%
\providecommand \bibfield  [0]{\@secondoftwo}%
\providecommand \translation [1]{[#1]}%
\providecommand \BibitemOpen [0]{}%
\providecommand \bibitemStop [0]{}%
\providecommand \bibitemNoStop [0]{.\EOS\space}%
\providecommand \EOS [0]{\spacefactor3000\relax}%
\providecommand \BibitemShut  [1]{\csname bibitem#1\endcsname}%
\let\auto@bib@innerbib\@empty
\bibitem [{\citenamefont {D'Alessandro}(2021)}]{AlessandroBook2021}%
  \BibitemOpen
  \bibfield  {author} {\bibinfo {author} {\bibfnamefont {D.}~\bibnamefont
  {D'Alessandro}},\ }\href@noop {} {\emph {\bibinfo {title} {Introduction to
  Quantum Control and Dynamics}}},\ \bibinfo {edition} {2nd}\ ed.\ (\bibinfo
  {publisher} {Chapman and Hall/CRC},\ \bibinfo {address} {Boca Raton},\
  \bibinfo {year} {2021})\BibitemShut {NoStop}%
\bibitem [{\citenamefont {Scully}\ and\ \citenamefont
  {Zubairy}(1997)}]{Scully_Zubairy_1997}%
  \BibitemOpen
  \bibfield  {author} {\bibinfo {author} {\bibfnamefont {M.~O.}\ \bibnamefont
  {Scully}}\ and\ \bibinfo {author} {\bibfnamefont {M.~S.}\ \bibnamefont
  {Zubairy}},\ }\href@noop {} {\emph {\bibinfo {title} {Quantum Optics}}}\
  (\bibinfo  {publisher} {Cambridge University Press},\ \bibinfo {address}
  {Cambridge},\ \bibinfo {year} {1997})\BibitemShut {NoStop}%
\bibitem [{\citenamefont {Lokutsievskiy}\ and\ \citenamefont
  {Pechen}(2021)}]{LokutsievskiyPechen1}%
  \BibitemOpen
  \bibfield  {author} {\bibinfo {author} {\bibfnamefont {L.}~\bibnamefont
  {Lokutsievskiy}}\ and\ \bibinfo {author} {\bibfnamefont {A.}~\bibnamefont
  {Pechen}},\ }\href {https://doi.org/10.1088/1751-8121/ac19f8} {\bibfield
  {journal} {\bibinfo  {journal} {Journal of Physics A: Mathematical and
  Theoretical}\ }\textbf {\bibinfo {volume} {54}},\ \bibinfo {pages} {395304}
  (\bibinfo {year} {2021})}\BibitemShut {NoStop}%
\bibitem [{\citenamefont {Lokutsievskiy}\ \emph {et~al.}(2024)\citenamefont
  {Lokutsievskiy}, \citenamefont {Pechen},\ and\ \citenamefont
  {Zelikin}}]{LokutsievskiyPechen2}%
  \BibitemOpen
  \bibfield  {author} {\bibinfo {author} {\bibfnamefont {L.~V.}\ \bibnamefont
  {Lokutsievskiy}}, \bibinfo {author} {\bibfnamefont {A.~N.}\ \bibnamefont
  {Pechen}},\ and\ \bibinfo {author} {\bibfnamefont {M.~I.}\ \bibnamefont
  {Zelikin}},\ }\href {https://doi.org/10.1088/1751-8121/ad5396} {\bibfield
  {journal} {\bibinfo  {journal} {Journal of Physics A: Mathematical and
  Theoretical}\ }\textbf {\bibinfo {volume} {57}},\ \bibinfo {pages} {275302}
  (\bibinfo {year} {2024})}\BibitemShut {NoStop}%
\bibitem [{Note3()}]{Note3}%
  \BibitemOpen
  \bibinfo {note} {Note that \cite {LokutsievskiyPechen2} uses a slightly
  different convention ($r_z\DOTSB \mapstochar \rightarrow -r_z$ and $u\DOTSB
  \mapstochar \rightarrow -u$) corresponding to the opposite choice of the
  raising and lowering operators $\sigma _\pm $.}\BibitemShut {Stop}%
\bibitem [{Note4()}]{Note4}%
  \BibitemOpen
  \bibinfo {note} {Theorem numbering is consistent with the main
  text.}\BibitemShut {Stop}%
\bibitem [{Note5()}]{Note5}%
  \BibitemOpen
  \bibinfo {note} {The curve $\Gamma $ consists of the point $N$ and two
  symmetric branches, $\Gamma =\Gamma ^+ \sqcup \Gamma ^- \sqcup \{N\}$. The
  branches $\Gamma ^\pm $ are trajectories of \protect \eqref {eq: QCS in Bloch
  ball appendix} with $u=u_b$. Moreover, both of these solutions tend to $N$ as
  $t\to -\infty $. Thus, $\Gamma ^\pm $ are the unstable separatrices of the
  point $N$.}\BibitemShut {Stop}%
\bibitem [{Note6()}]{Note6}%
  \BibitemOpen
  \bibinfo {note} {In the previous papers \cite {LokutsievskiyPechen1} and
  \cite {LokutsievskiyPechen2}, we have used slightly different notation.
  Namely, $r_z\DOTSB \mapstochar \rightarrow -r_z$, $u\DOTSB \mapstochar
  \rightarrow -u$, and $\theta \DOTSB \mapstochar \rightarrow -\theta
  $.}\BibitemShut {Stop}%
\bibitem [{Note7()}]{Note7}%
  \BibitemOpen
  \bibinfo {note} {STL = small-time localized}\BibitemShut {NoStop}%
\bibitem [{\citenamefont {Boscain}\ \emph {et~al.}(2023)\citenamefont
  {Boscain}, \citenamefont {Cannarsa}, \citenamefont {Franceschi},\ and\
  \citenamefont {Sigalotti}}]{BoscainControllability}%
  \BibitemOpen
  \bibfield  {author} {\bibinfo {author} {\bibfnamefont {U.}~\bibnamefont
  {Boscain}}, \bibinfo {author} {\bibfnamefont {D.}~\bibnamefont {Cannarsa}},
  \bibinfo {author} {\bibfnamefont {V.}~\bibnamefont {Franceschi}},\ and\
  \bibinfo {author} {\bibfnamefont {M.}~\bibnamefont {Sigalotti}},\ }\href
  {https://doi.org/10.5802/crmath.538} {\bibfield  {journal} {\bibinfo
  {journal} {Comptes Rendus. Math\'ematique}\ }\textbf {\bibinfo {volume}
  {361}},\ \bibinfo {pages} {1813} (\bibinfo {year} {2023})}\BibitemShut
  {NoStop}%
\bibitem [{Note8()}]{Note8}%
  \BibitemOpen
  \bibinfo {note} {In fact, the interiors of the sets of ST-local
  controllability and L-local controllability coincide here, but this is not
  essential for what follows.}\BibitemShut {Stop}%
\bibitem [{\citenamefont {Agrachev}\ and\ \citenamefont
  {Sachkov}(2004)}]{AgrachevSachkov}%
  \BibitemOpen
  \bibfield  {author} {\bibinfo {author} {\bibfnamefont {A.~A.}\ \bibnamefont
  {Agrachev}}\ and\ \bibinfo {author} {\bibfnamefont {Y.~L.}\ \bibnamefont
  {Sachkov}},\ }\href@noop {} {\emph {\bibinfo {title} {Control Theory from the
  Geometric Viewpoint}}}\ (\bibinfo  {publisher} {Springer},\ \bibinfo {year}
  {2004})\BibitemShut {NoStop}%
\bibitem [{\citenamefont {Davydov}(1994)}]{Davydov}%
  \BibitemOpen
  \bibfield  {author} {\bibinfo {author} {\bibfnamefont {A.~A.}\ \bibnamefont
  {Davydov}},\ }\href@noop {} {\emph {\bibinfo {title} {Qualitative Theory of
  Control Systems}}}\ (\bibinfo  {publisher} {American Mathematical Society},\
  \bibinfo {year} {1994})\BibitemShut {NoStop}%
\bibitem [{Note9()}]{Note9}%
  \BibitemOpen
  \bibinfo {note} {Recall that system~\protect \eqref {eq: QCS in Bloch ball
  appendix} is equivalent to~\protect \eqref {eq: QCS in cylindrical
  coordinates}, \protect \eqref {eq: QCS in cylindrical coordinates
  theta}.}\BibitemShut {Stop}%
\bibitem [{Note10()}]{Note10}%
  \BibitemOpen
  \bibinfo {note} {An explicit formula for $R(r_x,\theta )$ is also easy to
  write: the equation $S=0$ is quadratic in $R$, but it is not important for
  what follows.}\BibitemShut {Stop}%
\bibitem [{Note11()}]{Note11}%
  \BibitemOpen
  \bibinfo {note} {Recall that $E_+$ is the exterior of the figure-eight curve
  $E$.}\BibitemShut {Stop}%
\end{thebibliography}%
\end{document}